\documentclass[12pt]{article}
\usepackage{etex}
\usepackage{graphicx}
\usepackage{setspace}
\usepackage{comment}
\usepackage{amsmath}
\usepackage{amsfonts}
\usepackage{amssymb}
\usepackage[usenames,dvipsnames]{color}
\usepackage{amsthm}
\usepackage{graphics,epsfig,verbatim,bm,latexsym,url,amsbsy}
\usepackage{rotating}
\usepackage{float}
\usepackage{subfig}
\usepackage{hyperref}
\usepackage[authoryear]{natbib}
\usepackage{multirow}
\usepackage{soul}
\usepackage{tikz}
\usetikzlibrary{calc, positioning, fit, shapes.misc}
\usepackage{pgfplots, pgfplotstable}
\usepackage[flushleft]{threeparttable}
\usepackage[normalem]{ulem}
\usepackage[ruled,vlined]{algorithm2e}
\usepackage{booktabs}
\usepackage{cleveref}
\usepackage{cancel}

\tikzset{
  he/.style={
    rounded corners,
    inner sep=4pt,
  }
}

\definecolor{airforceblue}{rgb}{0.36, 0.54, 0.66}
\definecolor{battleshipgrey}{rgb}{0.52, 0.52, 0.51}
\definecolor{charcoal}{rgb}{0.21, 0.27, 0.31}
\hypersetup{colorlinks=true, linkcolor=airforceblue, citecolor=airforceblue, urlcolor=battleshipgrey}

\newtheorem{lemma}{{\bf \sc Lemma}}

\newtheorem{corollary}{{\bf \sc Corollary}}
\newtheorem{proposition}{{\bf \sc Proposition}}

\newtheorem{observation}{{\bf \sc Observation}}

\def\Re{I\!\!R}
\def\eproof{\hbox{\hskip3pt\vrule width4pt height8pt depth1.5pt}}

\long\def\symbolfootnote[#1]#2{\begingroup \def\thefootnote{\fnsymbol{footnote}}\footnote[#1]{#2}\endgroup}

\usepackage{footnote}
\makesavenoteenv{tabular}

\renewcommand{\max}{\operatornamewithlimits{max}}
\renewcommand{\min}{\operatornamewithlimits{min}}

\pgfplotsset{width=10cm,compat=1.9}

\begin{document}

\title{Supply Chain Disruptions, the Structure of Production Networks, and the Impact of Globalization 
}
\author{Matthew L. Elliott and Matthew O. Jackson\thanks{%
Elliott is at the Faculty of Economics at Cambridge University.
Jackson is at the Department of Economics at Stanford University, and also an external faculty member of the Santa Fe Institute.  Web sites:  \href{https://sites.google.com/view/matthewlelliott/}{https://sites.google.com/view/matthewlelliott/} and
 \href{http://www.stanford.edu/\%7Ejacksonm/}{http://www.stanford.edu/$\sim$jacksonm/} .
We thank David Baqaee, Vasco Carvalho, Glenn Magerman, Ankur Mani, Mike Smith, Evan Storms, and Chao Wei for helpful discussions.
We gratefully acknowledge support under NSF grant SES-2018554, the European Research Council (ERC) under the European Union's Horizon 2020 research and innovation programme (grant agreement number \#757229) and the JM Keynes Fellowships Fund. Alastair Langtry, Shane Mahen, Sarah Taylor and Yifan Zhang provided excellent research assistance.
Disclosure: Jackson is a scientific advisor of (with stock options in)
Altana AI, which is involved in collecting supply chain data and advising companies and governments on supply chain issues.}
}
\date{January 2026}
\maketitle

\begin{abstract}
    We introduce a parsimonious multi-sector model of international production and use it to study the impact of a disruption in the production of some goods propagates to other goods and consumers, and how that impact depends on the goods' positions in, and overall structure of, the production network. We show that the short-run impact of a disruption can be dramatically larger than the long-run impact. The short-run disruption depends on the value of all of the final goods whose supply chains involve a disrupted good, while by contrast the long-run disruption depends only on the cost of the disrupted goods. We use the model to show how increased complexity of supply chains leads to increased fragility in terms of the probability and expected short-run size of a disruption. We also show how decreased transportation costs can lead to increased specialization in production, lowering the chances for disruption but increasing the impact conditional upon disruption.  We use the model to characterize the power that a country has over others via diversions of its production as well as quotas on imports and exports.

\smallskip


\smallskip


\end{abstract}

\setcounter{page}{0}\thispagestyle{empty}


\section{Introduction}\label{sec:introduction}

One among many global supply chain failures stemming from the labor and transport disruptions of the COVID pandemic was a worldwide-shortage of integrated circuits and, in particular, basic computer chips.\footnote{ Long lead times in the necessary capital equipment, a drought in Taiwan that impacted on the production capacity of the industry, and high demand as people switched expenditure from experiences to products due to the Covid-19 pandemic, all played a role, along with several other factors including a fire at a production plant. See, for example, \cite{McLain2021}; \cite{Jeong2021}; \cite{Davidson2021}.}
These appear in almost all consumer goods that involve electronics, and integrated circuits rank fourth among products traded internationally \citep{Jeong2021}. Although basic computer chips are a commodity good that in typical times can sell for a few cents each,\footnote{ For example, an AC-DC Switching Power Supply Pulse-Width Modulation Integrated Circuit sells for around 6 cents per unit with a minimum order of 1000 
.} their shortage stalled the production of many downstream goods, having global economic implications.\footnote{There are a variety of studies documenting that disruptions propagate through supply chains \cite{BarrotSuppliers2016}, \cite{BohemInput2019} and \cite{carvalho2021supply}.}
How can we predict the economic impact of the disruption in the supply of one or more goods in a production network? How does this vary over time as the economy adjusts?

Equilibrium approaches neglect important constraints on production that are needed to assess the short-run impact. We show that these constraints are important as they can significantly amplify shocks in the short-run and the magnitude of this amplification varies substantially with the network position of the shocked industry. 

It is important to emphasize that existing models cannot be adjusted to estimate short-run disruptions simply by reducing the elasticity of substitution between inputs and/or making labor immobile. Although this reduces the ability of the economy to adjust in some ways, it is conceptually and qualitatively distinct from an out of equilibrium analysis in which resources (especially labor) are wasted and production is temporarily halted. Constraints capture the reality that some shocks in the short run are such that alternative suppliers cannot be quickly found and that supply relationships governed by explicit or implicit contracts prevent prices from adjusting and efficiently rationing goods.  For illustration, consider the chip example discussed above, and suppose there is a product that must have a particular kind of fairly inexpensive chip to function (e.g., a car) and so the demand for that chip is very inelastic. If the car manufacturer sources that chip (directly or indirectly via its direct suppliers) from one particular plant or country and some event, for instance political, closes access to that plant then the production of the car is completely disrupted in the short run. This inefficient shutdown of production is a fundamentally out-of-equilibrium phenomenon. In an equilibrium model, there would instead be more expensive production of that chip, and that would drive up the cost of the chip, but as a percentage of the price of the car, the impact would be relatively small even with completely inelastic demand for the input. 
However, stockpiled inventories are quickly exhausted (even if optimally redirected) and it can take months if not years for new sources of production to arise.  Indeed, during the pandemic, this is exactly what we saw.\footnote{ See, for example, \href{https://www.reuters.com/legal/legalindustry/semiconductor-shortage-us-auto-industry-2021-06-22/}{Reuters, (06-222021)}.} Thus, an equilibrium model that omits the misallocation of resources that occurs in the aftermath of disruptions, does not account for 
lags in building new production, and abstracts from limitations in reallocating existing production, can fail to account for much if not most of the short-run impact.\footnote{ Extensive anecdotal evidence of production being disrupted or halted because of difficulty sourcing inputs is presented in Section 2 of \cite{elliott2022networks}.}

In this paper, we introduce an international production and trade model that incorporates such short-run constraints on production, and can also evaluate full equilibrium adjustments. Thus, we can use it to characterize and contrast the extremes of short-run as well as long-run impacts of shocks to production.
The extremes are both benchmarks, but both are of importance, as well as their contrast.
We use the model to characterize the impact of production disruptions and how this depends on the structure of supply chains, the position of the disrupted goods in these chains, and the value of the final goods produced via the impacted chains. We contrast the short- and long-run impact of a drop in some producers' productivities. We also show how the potential for short-term disruption increases linearly in the complexity of supply chains (for small disruption probabilities) and how this compares with the long run. 
In addition, we examine how the potential for disruption changes as transportation costs drop and supply chains evolve to be more specialized and involve more trade across countries.  We also use the model to characterize the power that one country has over some other(s) by disrupting trade flows in and/or out of its borders.  This has become increasingly important as international trade bargaining stakes have increased \citep{clayton2025putting}.

The model is important to have in hand given the increasingly complex supply chains that have emerged over time, and the fact that
the production of some key inputs is highly consolidated,\footnote{For example, studies suggest that the majority of complex computer chips are assembled in Taiwan (close to 90 percent for the most advanced), and yet the inputs cross borders more than 70 times before reaching the final stage of production (\url{https://www.gsaglobal.org/globality-and-complexity-of-the-semiconductor-ecosystem/}).} while others are not. 

Our model and analysis provides a look at both extreme benchmarks:  the short run which occurs ``out of equilibrium'' in which disruptions are constraints on production,  and the long run in which there is a full equilibrium adjustment, as well as a look in between. 
Regardless of whether either benchmark applies, there are several important reasons for the literature to have these benchmarks and a model in which to examine the full spectrum of possibilities.   
One is that having careful theoretical benchmarks allows us to measure where an industry or supply chain falls relative to those benchmarks and how that trajectory adjusts with time after a disruption.   Second, one can then see how such trajectories from short run to long run depend on the shocked industry and other aspects of the supply chain.  Do parts of a supply chain adjust more quickly?  Which characteristics of a supply chain or industry make it more flexible?  How does this depend on which part of a supply chain is disrupted?
When is it even worthwhile for an industry to rewire rather than just absorb short term disruption(s) and wait until production returns to normal?
Third, developing such a theoretical framework provides us with a laboratory in which we can develop intuitive understandings of these issues and further develop answers to the above questions, both theoretical and as a guide for deeper empirical analyses and policy development.
This is especially necessary in an age in which supply chains are increasingly global and subject to political influence and international disruptions.   

Moreover, beyond providing an important benchmark, 
short-run supply chain disruptions include all shocks: temporary shocks as well as the first periods of longer-term shocks, and hence are the norm rather than the exception. 
Moreover, temporary shocks are 
of high frequency.  
Even before the myriad of problems associated with Covid-19 and the tariff shocks of recent times, \cite{snyder2016or} (p. 89) write, ``It is tempting to think of supply chain disruptions as rare events. However, although a given type of disruption (earthquake, fire, strike) may occur very infrequently, the large number of possible disruption causes, coupled with the vast scale of modern supply chains, makes the likelihood that some disruption will strike a given supply chain in a given year quite high.''
Indeed, every month there are many such disruptions that cause companies nontrivial losses and are such that it may not make sense for them to rewire, but instead to simply absorb the losses.   
Just as a recent typical example, in the last quarter of 2025 Honda Motors reported a loss of 1 billion dollars (more than a quarter of their annual operating profit) due to a shortage of chips from the Dutch firm Nexperia, which was caught in a trade dispute.  This resulted in temporary shutdown of a Honda plant in Mexico and slowdowns in their US and Canadian plants.\footnote{Reuters: \url{https://www.reuters.com/world/asia-pacific/honda-reports-25-fall-q2-operating-profit-2025-11-07/ }.}
Finding a new supplier for such custom chips is expensive and can simply result in a different source of risk, and so such a short term shock was simply absorbed.  

More generally, diversifying production of custom parts or service relationships is time-consuming and expensive, as it requires large investments in firm-specific ties and technologies.
A firm has a limited ability to intervene following such a shock in time to mitigate its impact. More likely, the firm instead accepts the losses during the shock. Note that there is no equilibrium adjustment here, just out of equilibrium losses. Despite this, only equilibrium models exist for estimating such losses.
Even if such losses only account for a percent of a typical firm or industry's production, this would be the difference between a two and three percent growth in GDP,
and it could be not in firms' interests to pay large fixed costs to diversify against such losses.

Our analysis begins with the long run.
There we show that the change in productivity of an input can be completely compensated for by
equilibrium adjustments of quantities and sources of all inputs, a version of Hulten's \citeyearpar{hulten1978growth} Theorem holds.  
Thus, our model in the long run matches the usual competitive long-run marginal benchmark.   That is, the marginal impact of a shock is proportional to the total amount spent on that input relative to total GDP.
 For instance, a shock that reduces the productivity of a \$500 billion market like that for integrated circuits by around 5 percent, would have a long-run impact of around \$25 billion.
Roughly, at the margin, if productivity drops by 5 percent then the circuits become about 5 percent more expensive, and so the total resource loss is an extra 5 percent of what was being spent on
that input originally.\footnote{This intuition applies at the margin, and can even
overestimate the impact as it does not account for the possibility that the input mix can be changed to mitigate the impact of the shock as one moves away from the margin.} However, in the short run, the impact is much larger. A 5 percent shortage of integrated circuits propagates to prevent the production of 5 percent of all final goods that use them as inputs, directly or indirectly, at some stage of production. Those final goods, valued at around \$5 trillion, lose production valued at around \$250 billion, or ten times the long-run impact. This comparison is even starker if a shock to basic computer chips is considered given their low value and the breadth of their use as inputs. 
The contrast can be summarized as the short-run impact being proportional to the value of all downstream final products, while the long-run impact is proportional to the change in the cost of producing the input itself. 

Central to our theory is that shocks propagate through the supply network and are amplified in the short-run. Much evidence is consistent with this.  \cite{carvalho2021supply} analyze the Great East Japan Earthquake and the resulting propagation of supply shocks. The localized impact of the disaster accounts for less than a 0.1 percentage point reduction in overall GDP growth, and this includes the propagation of disruptions within the affected areas. So, the direct impact of the shock is likely to be considerably below this. At the same time, there was a 0.4 percentage point reduction in the growth of GDP year-on-year and their counterfactual analysis estimates the actual impact at 0.47 percentage points. Further, tracing through affected customers and suppliers in the year following the shock, they find that the impact diminishes, but slowly. A further point of interest is that while they find inputs to be weak substitutes on the time frame of a year, \cite{BarrotSuppliers2016} find them to be complements on a quarterly basis. As we show, this is a prediction of our theory. Intuitively, as time progresses (which can be years, as in the \cite{carvalho2021supply} case) constraints are relaxed. This allows alternative suppliers to increase their production where multisourcing is already present, and for new suppliers to be found. This makes inputs more substitutable over time.

The lack of a short-run model able to capture out-of-equilibrium propagation through a network, and inadequacy of a long-run equilibrium approach to substitute for such calculations has been lamented by Larry Summers:

``I always like to think of these crises as analogous to a power failure or analogous to what would happen if all the telephones were shut off for some time. Consider such an event. The network would collapse. The connections would go away. And output would, of course, drop very rapidly. There would be a set of economists who would sit around explaining that electricity was only 4\% of the economy, and so if you lost 80\% of electricity, you couldn’t possibly have lost more than 3\% of the economy. Perhaps in Minnesota or Chicago there would be people writing such a paper, but most others would recognize this as a case where the evidence of the eyes trumped the logic of straightforward microeconomic theory. ...  And we would understand that somehow, even if we didn’t exactly understand it in the model, that when there wasn’t any electricity, there wasn’t really going to be much economy.''\citep{summers2013}.\footnote{ See \cite{carvalho2019production} for some further discussion around this quote and the impact of equilibrium disruptions when Hulten's theorem doesn't hold. Of particular note is \cite{baqaee2019macroeconomic}. See also \cite{wei2003energy} for a discussion of the timing of adjustments.}

We develop an algorithm to characterize how the impact of a short-run disruption depends on the structure of the supply chains involved, and show that this algorithm finds the solution to a well-defined constrained optimization problem.
An upper bound is that the disruption is equal to the percentage reduction in the output of the shocked technologies multiplied by the value of production of all final goods that use the shocked good as an input directly or indirectly. We identify several natural and intuitive situations in which this bound is tight. One such case is when there is conformity in the inputs used by producers within industries, and there is an industry shock---i.e., all producers of a certain good are shocked. Moreover, when all transportation costs are reduced sufficiently, countries specialize in what they produce and hence technology-specific shocks become industry-specific shocks, such that the bound is always obtained.

The more general calculation can differ from the bound depending on two issues. One is diversity of sourcing: some technologies might source the same input from multiple suppliers, and if not all of those suppliers are affected, then impact of the shock is reduced. The other is diversity of production technologies: different technologies that require different inputs might be used for producing the same good, and so be unaffected by a shock that propagates downstream to disrupt a competitor. However, if there are cycles in the supply network, then these can feed back and amplify disruptions allowing the bound to be obtained even in the presence of diversity in sourcing and technology. To calculate the short-run impact of a shock and how it propagates we provide a (convergent) algorithm,  allowing for cycles in the production network. The algorithm converges to the maximum amount that can be produced of the final goods subject to the shocks.

We also examine the medium-run in which supply is constrained, but in which there is equilibrium adjustment so that goods flow to the highest-value downstream goods that use a given input.
For example, if production of chips are disrupted in the very short run, all manufacturers using those chips may face delayed deliveries and lost production. Over time, the prices of chips can adjust and redirect the chips to the uses where they are most valued, and this can improve production.  We show that this mitigation can lead to an impact that falls anywhere between the effect of the short versus long run.
We show that this medium-run impact depends on the diversity in the values of the downstream production using an input.

We use our characterizations to examine the expected disruption due to an independent probability of a shock to different inputs.  We show that as the complexity of the supply chains, as measured by the number of inputs, increases, the expected loss in GDP holding all else fixed increases linearly in the complexity.
We go on to discuss how the long-run impact varies with the depth vs breadth of supply chains, while the short-run impact is more sensitive to how many final goods lie downstream of the shocked inputs.

We then use this to examine some comparative statics in trade costs. When all trade costs are reduced sufficiently the local prices of intermediate goods are equalized across locations and are sourced from a lowest cost technology and so production becomes specialized. This leads to lower diversity in sourcing and diversity in production technologies is lost.  There are countervailing effects:  there are fewer sources to disrupt, which can lower the chance of a disruption occurring, but conditional upon occurrence, the impact is larger. This is consistent with
empirical evidence on globalization, specialization, and fragility \citep{giovanni2009trade,magerman2016heterogeneous,di2022global,bernard2022origins,baldwin2022risks}.

Our results can be applied to a variety of situations in which supply chains become disrupted, including the imposition of sanctions or other directed trade disruptions and anticipating their effects.\footnote{In this context, \cite{Moll2022} estimate the impact of Russian energy sanctions for the German economy, varying assumptions about elasticities of substitution to represent the shorter run and longer run effects.} As such, the model serves as a foundation to understand the {\sl power} that one country holds over another in terms of how much disruption the first can cause in the second by disrupting imports and/or exports\textemdash e.g., via sanctions, quotas, or tariffs.
We measure this by examining the damage a disrupting country can cause to a target country per unit of damage to the disrupting country's domestic industries.   
In measuring the impact of intentional disruptions we need to track the distribution of disrupted (idle) labor across countries.  We develop an algorithm for evaluating the total disrupted labor, up and downstream from a given disrupted technology and show how the potentially combinatorially intensive exercise of identifying the disruptions that maximize one country's power over another, simplifies. Intuitive aspects emerge from our characterization of power: one country wields substantial power over another when it can disrupt a relatively low value-added industry domestically and this has far-reaching up and/or downstream consequences in the target country, that do not cause significant further disruptions at home. An example of such an industry in practice is the industry for rare-earth metals and magnets in China, which are (directly and indirectly) important inputs into high-value production in the US among other countries.\footnote{See, for example, \href{https://www.reuters.com/world/china-hits-back-us-tariffs-with-rare-earth-export-controls-2025-04-04/}{Reuters (04-04-2025)}.}  Our analysis also shows how it is possible for two countries each to have substantial power over each other (via different supply chains), or for power to be very asymmetric, depending on the structure of the supply network.

It is worth emphasizing that short-run disruptions in supply chains add up.  Although widespread short-run shocks that cause multiple percentage point losses in GDP such as COVID-19 are rare, shocks to industries that cause temporary shortages and price spikes to consumers, such as the car industry example discussed above, disruptions in shipping lanes, local wars, strikes, natural disasters, climate shocks, trade wars, etc., are common and can lead to death-by-a-thousand-cuts.  Thus, it is important to understand and better manage such short-run disruptions.  It could be that in an ideal world without any such disruptions, GDP would be percentage points higher.

\paragraph{Related Literature}
More generally, the three most closely related strands of literature to our work are: a macroeconomic literature on production networks, a more microeconomic literature on supply chain robustness, and the literature on international trade and global value chains.

Building on the seminal work of \cite{leontief1936quantitative}, \cite{long1983real} and \cite{acemoglu2012network}, a series of papers examine the propagation of shocks through sectoral and firm inter-linkages in the economy.\footnote{See \cite{carvalho2019production} and \cite{baqaee2022micro} for recent surveys.} Some of the recent work has incorporated cascading failures, production shut-downs and endogenized the network structure \citep{costinot2013elementary,dhyne2015belgian,magerman2016heterogeneous,brummitt2017contagious,baqaee2018cascading, oberfield2018endogenous,acemoglu2020firms,acemoglu2020endogenous,baqaee2021entry, kopytov2021endogenous,konig2022aggregate,chen2023invest,grossman2023supply,pin2025network}. However, within that literature, the focus has been on the (long-run) equilibrium impacts of shocks in which all factors are perfectly flexible and the economy re-equilibrates. We contribute to this literature by considering the short-run impact of a shock, with no adjustments by introducing a first model that incorporates constraints into networked production.\footnote{There is a  literature on ``disequilibrium'' models, beginning with \cite{galbraith1947disequilibrium,barro1971general,benassy1977neokeynesian}.  That literature is more interested in modeling macroeconomic dynamics under price controls or other frictions on convergence to equilibrium prices, rather than on supply chain disruptions, and hence those models are not closely related to ours.}

Perhaps closest to us is research that considers the impact of shocks in the presence of frictions, including limited substitutability of inputs, fixed factors of production, building delays and search and matching frictions (e.g.,\cite{baqaee2019macroeconomic}, \cite{Moll2022, bachmann2024if}, \cite{bui2022information}, \cite{pellet2023rigid}, \cite{bai2024causal}, \cite{baslandze2025price}, \cite{brancaccio2025rigidities}, \cite{arkolakis2025production} and \cite{schaal2025echoes}). In terms of the typology of time horizons we develop in Section \ref{sec:time_horizons}, that work relates to our long-run or medium-run analysis. Our main focus is on the short-run which fundamentally differs by, instead of being in (constrained) equilibrium, being out-of-equilibrium (through the inability of prices to adjust to clear markets).

There is also considerable work studying networks and fragility.\footnote{See \cite{elliott2022networks} for a recent survey, while \cite{baldwin2022risks} reviews the literature on risks in global supply chains.} Like ours, that work is mainly theoretical. The work closest to us in this area focuses on the fragility of supply chains and is complementary in so far as it focuses on better understanding the frictions that lead to the formation of inefficient supply networks. However, it abstracts from general equilibrium considerations and tends to focus on supply networks that are only a couple of layers deep. It includes, for example, \cite{bimpikis2018multisourcing}, \cite{bimpikis2019} and \cite{amelkin2020strategic}. Perhaps closest is \cite{elliott2022supply}. Like us the focus there is on the macroeconomic implications of shocks in the short-run and deep networks are accommodated. However, they investigate firms' strategic investments into the local robustness of their supply chains, which we do not consider, and study when equilibrium investments will yield fragile networks.\footnote{ The set up of their model (simple networks with no cycles to facilitate percolation and small firms) is also very different, as is the amplification mechanism. In \cite{elliott2022supply} individual firms are small and in equilibrium production is robust with respect to the failure of any one firm. In contrast, we have representative firm/technologies and study how shocks to individual technologies propagate.}  In contrast our analysis provides details on the how the size and scope of a disruption depend on specific details of the supply chain, which are abstracted away from in their work.
On the empirical side, \cite{diem2022quantifying} develop and estimate using Hungarian value added tax data, a measure of the systemic risk different firms pose to the economy by considering the short-run impact on aggregate output of a randomly selected firm failing.  They find that the empirical impacts are large in magnitude, consistent with the contrast that we draw between the short- and long-run.

Finally, there is related research on networks and international trade.\footnote{For recent surveys see \cite{bernard2018networks} and \cite{antras2022global}.} Some of the more closely related work in that literature seeks to better understand which importers match to which exporters, and how this is influenced by various frictions (see, for example, \cite{chaney2014network,Chaney2016, bernard2019production,grossman2024tariffs}). In terms of the macro modelling of international trade, the approach taken has tended to be very different from ours. For example, the workhorse models of \cite{melitz2003impact} and \cite{caliendo2015estimates}, while being well suited for answering a variety of questions and fitting various aspects of the trade data, are not so well suited to understanding how shocks' propagate and amplify depending upon the position of disruption in the trade network.

\section{A Model of (International) Supply Networks}\label{sec:model}

\subsection{The Model}

We use the same notation for each set and its cardinality.

\noindent \textbf{Goods and Countries:}
There is a set $N$ of countries indexed $n\in \{1,\dots ,N\}$.

Goods consist of a set $M$ of intermediate goods, including raw materials, indexed by $m\in \{1,\dots ,M\}$ that are used as inputs to production; and a set $F$ of final goods indexed $f\in \{1,\dots ,F\}$ that are consumed.  We assume that final goods are never used as inputs to production to simplify notation, but it is trivial to extend the model to permit this.\footnote{ Simply create a duplicate industry: If a country produces a good that is used as both an intermediate good and a final good, then let there be two industries producing the good, one of which sells it only as an intermediate good, and another that sells it only as a final good to consumers.  Prices depend on costs of inputs and so identically-produced goods end up with the same
prices in equilibrium.}

We use the term ``goods'' throughout the paper but emphasize that these include not only agricultural and manufactured goods, but also include services.

\noindent \textbf{Labor and Endowments:}
Country $n$'s endowment of labor is denoted $L_n>0$, and it is supplied completely inelastically so that in equilibrium $L_n$ units are all used in production.\footnote{This simply makes prices and examples easier to compute in our model, but would not affect our results (Hulten's Theorem would still hold in the long run, and our constrained approach does not adjust labor inputs in the short run).}

Access to raw materials (which are a special type of intermediate good) within a country is represented via the available production technologies: for instance, a country that has oil has a technology that outputs oil, while if there is no oil in a country then there is no technology available in that country to produce oil.

\noindent \textbf{Technologies:}
We work with Arrow-Debreu production economies. A technology is described, as in the classic model of \cite{arrow1954existence} (a production plan in their parlance), by a vector $\tau\in \Re^{1+M+F}$.
A technology lists the combinations of labor and intermediate goods required as inputs to produce positive amounts of output; with the interpretation that $\tau_k<0$ implies that good $k$ is an input and $\tau_k>0$ implies good $k$ is an output.

We focus on constant returns to scale technologies.  Although constant returns are not needed for our main results, we make it for simplicity as the model is already complicated by accounting for the supply network.  Importantly, for short-term disruptions and measurements, constant returns are appropriate.

A technology $\tau$ satisfies three additional conditions.

First, $\{k:\tau_k>0\}$ has exactly one element, interpreted as the output of the technology.  Other entries are either 0 or negative, with the negative ones being the inputs.
Let $O(\tau)=\{k:\tau_k>0\}$  denote the output good associated with technology $\tau$ and $I(\tau)=\{k:\tau_k<0\}$ denote the input goods.

Second, we normalize the output to $1$ so that $\max_k \tau_k=1$. Thus, a technology indicates the amounts of inputs needed to produce one unit of the output good, and this can be scaled to any level given the constant returns to scale.

Third, the leading entry in the technology vector, representing the amount of labor required, is strictly negative for all technologies, so that labor is needed in each production process.
This ensures that there are no infinite production paths that are costless, and ensures existence.

Throughout we associate each technology with a unique producer, but nothing would change if we instead associated the production of a given good $g$ in a given country $n$ with a unique producer endowed with all technologies $\{\tau\in T_n:O(\tau)=g\}$.

An example of technologies appears in Figure \ref{fig:flow_network}.

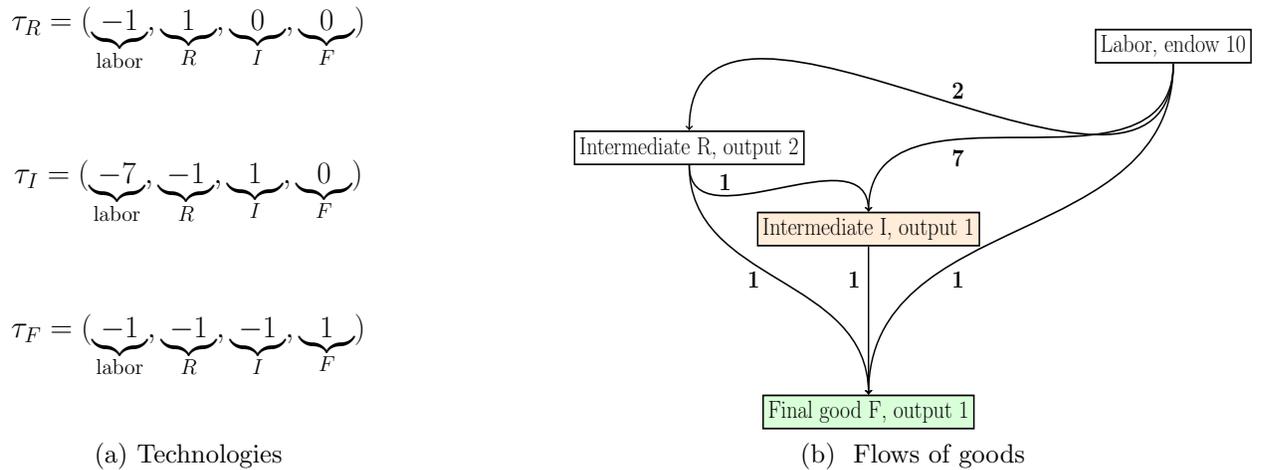
\begin{figure}[ht!]
\subfloat[Technologies\label{subfig:flow}]
{
\centering			
\resizebox{2in}{2.25in}{
\begin{tikzpicture}
 \node[] at (4.5, -1.5) {};
 \node[] at (4.5, 6) {\Large $\tau_R=(\underbrace{-1}_{\text{labor}},\underbrace{1}_{R},\underbrace{0}_{I},\underbrace{0}_{F})$};
 \node[] at (4.5, 3) {\Large $\tau_I=(\underbrace{-7}_{\text{labor}},\underbrace{-1}_{R},\underbrace{1}_{I},\underbrace{0}_{F})$};
 \node[] at (4.5, 0) {\Large $\tau_F=(\underbrace{-1}_{\text{labor}},\underbrace{-1}_{R},\underbrace{-1}_{I},\underbrace{1}_{F})$};	
 \end{tikzpicture}
}	
}
\hfill
\subfloat[\label{subfig:techs} Flows of goods]
{
\centering
\resizebox{3.6in}{2.7in}{						
\begin{tikzpicture}[
		sqRED/.style={rectangle, draw=black!240, fill=red!15, very thick, minimum size=7mm},
		sqBLUE/.style={rectangle, draw=black!240, fill=blue!15, very thick, minimum size=7mm},
		roundnode/.style={ellipse, draw=black!240, fill=green!15, very thick, dashed, minimum size=7mm},
		sqGREEN/.style={rectangle, draw=black!240, fill=green!15, very thick, minimum size=7mm},
		roundRED/.style={ellipse, draw=black!240, fill=red!15, very thick, dashed, minimum size=7mm},
		roundREDb/.style={ellipse, draw=black!240, fill=red!15, very thick, minimum size=7mm},
		roundYELL/.style={ellipse, draw=black!240, fill=yellow!20, very thick, dashed, minimum size=7mm},
		sqYELL/.style={rectangle, draw=black!240, fill=yellow!20, very thick, minimum size=7mm},
		roundORA/.style={ellipse, draw=black!240, fill=orange!15, very thick, dashed, minimum size=7mm},		
		sqORA/.style={rectangle, draw=black!240, fill=orange!15, very thick, minimum size=7mm},
		squaredBLACK/.style={rectangle, draw=black!240, fill=white!15, very thick, minimum size=7mm},	
		roundBROWN/.style={ellipse, draw=black!240, fill=brown!15, very thick, dashed, minimum size=7mm},
		sqBROWN/.style={rectangle, draw=black!240, fill=brown!15, very thick, minimum size=7mm},		
        sqPURPLE/.style={rectangle, draw=black!240, fill=purple!15, very thick, minimum size=7mm},				
		]

		\node[squaredBLACK](R1) at (-3, 5.5) {\Large Intermediate R, output 2};
        \node[] at (4.5, 6.9) {\Large $\bm{2}$};
        \node[] at (4.5, 5.25) {\Large $\bm{7}$};
        \node[] at (4.5, 2.25) {\Large $\bm{1}$};

		\node[sqORA, align=left](I1) at (2, 3.5) {\Large Intermediate I, output 1};
        \node[] at (-2.0, 4.65) {\Large $\bm{1}$};

        \node[] at (-1.2, 2.25) {\Large $\bm{1}$};

        \node[] at (1.6, 2.25) {\Large $\bm{1}$};

		\node[sqGREEN, align=left](F1) at (2, -1) {\Large Final good F, output 1};

		\node[squaredBLACK, align=left](L) at (10.5, 8) {\Large Labor, endow 10};

		\tikzset{thick edge/.style={-, black, fill=none, thick, text=black}}
		\tikzset{thick arc/.style={->, black, fill=black, thick, >=stealth, text=black}}

		\draw[line width=0.4mm, black,-> ] (R1) to[out=-90,in=90] (I1);
		
        \draw[line width=0.4mm, black,-> ] (L) to[out=-90,in=90] (R1);
        \draw[line width=0.4mm, black,-> ] (L) to[out=-90,in=90] (I1);
        \draw[line width=0.4mm, black,-> ] (L) to[out=-90,in=90] (F1);		
		
        \draw[line width=0.4mm, black,-> ] (R1) to[out=-90,in=90] (F1);

		\draw[line width=0.4mm, black,-> ] (I1) to[out=-90,in=90] (F1);

		\end{tikzpicture}
}
}
\caption{\footnotesize{\label{fig:flow_network}An example of an economy with three technologies and 10 units of labor, outputting 1 unit of a final good. }}
\end{figure}

Let $T_n$ be the set of technologies assigned to country $n$, and let $T=\cup_n T_n$ be the set of technologies.
We take $T$ to be finite.

This is for several reasons.

One is that any given Arrow-Debreu equilibrium in our setting can be rationalized with a finite set of production plans (our technologies) and this is without loss of generality for our results.
For example, a Cobb Douglas production function of the form $y=\ell^\alpha m^{1-\alpha}$ is represented by following the set of triples as $\ell$ is varied: $(\tau_\ell)_\ell=(-\ell,-\ell^{\frac{-\alpha}{1-\alpha}},1)_\ell$ where the first entry represents the quantity of labor input, the second entry represents the quantity of the intermediate good $m$ input and a single unit of the final good is produced.
This is pictured in Figure \ref{fig:AD}.
Moreover, other productions (CES or some much more general function) can be approximated by our sets, and so by using general $T_n$ we can avoid imposing restrictions on the nature of those functions and still derive our results.

\begin{figure}[ht!]
{
\centering
\includegraphics[height=2.2in]{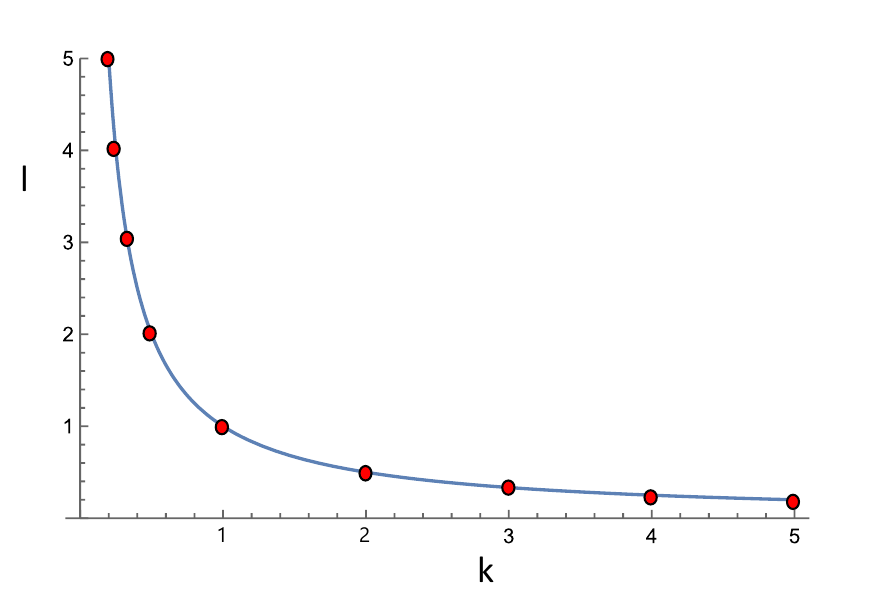}
\caption{\footnotesize{Long-run production of good 2 according to a Cobb-Douglas function of labor and good 1: $y_2=L^{\alpha}x_1^{1-\alpha}$, can be approximated by setting $T_n$ to be the intersection of $\{(-l,-x_1,1):l^{\alpha}x_1^{1-\alpha}=1\} $
with some grid.
\label{fig:AD}
}}}
\end{figure}

Second, and importantly, is that this realistically matches short-run implications, and more realistically than using some production function in which inputs can be readjusted.
For instance, if a business loses electricity, or a car manufacturer loses steering wheels, or a phone producer loses chips,  they cannot overcome this simply by shifting to use more labor.  In the longer run, they can reoptimize their production recipe to adjust inputs and re-source or adopt different technologies, but in the short run such substitution is impossible.
Thus, working with fixed sourcing and vectors (technologies) in the short run and then allowing for more potential variations in the longer run provides a natural contrast.

Third is that a finite set of available technologies that describe recipes for production is realistic when it comes to addressing things like innovations to technologies, which is something that we wish our model to address.

\noindent \textbf{Shipped units:}
The matrix $\mathbf{x}\in\Re_+^{N+T}\times \Re_+^{T}$ denotes the number of units shipped.  $x_{\tau'\tau}$ denotes the amount of the output, $O(\tau')$, produced by technology $\tau'$ that is shipped for use as an input by technology $\tau$.
We let $x_{n\tau}$ denotes the amount of labor endowed to country $n$ that is used by technology $\tau$, possibly in another country.

\noindent \textbf{Transportation costs:}
Transportation costs are captured via the matrix $\mathbf{\theta}\in \Re_+^{N+T}\times \Re_+^{T}$.
Entry $\theta_{n \tau}\geq 1$ denotes the amount of labor that must be supplied from country $n$ to technology $\tau$ in order for technology $\tau$ to get one unit of labor input. This could reflect a cost of remote working, or a cost of migration, among other things.
Entry $\theta_{\tau' \tau}\geq 1$ for $\tau',\tau\in T$ denotes the number of units of good $O(\tau')$ that must be shipped for technology $\tau$ to receive one unit of this good to use as an input.
These costs can reflect many things, for instance international shipping costs or tariffs,\footnote{These capture tariffs in terms of the calculation of equilibrium, but ignore the revenue that they generate which could impact welfare in ways ignored here.} among other things, but can also represent shipping costs internal to a country (for instance, shipping from the place where a raw material is produced to where a manufacturing plant is located). Thus, implicitly, technologies when coupled with this matrix encode transportation costs and locations.

Note also that differences in the effectiveness of labor and any intermediate goods across countries can be captured via transportation costs and available technologies.

Having final goods be costless to transport simplifies the consumption problem and enables us to concentrate on the production process. It is not necessary for our main results, but makes comparisons to key existing results (i.e., Hulten's Theorem) possible and so we maintain the assumption.

\noindent \textbf{Prices:}
Prices for each good and labor are described by a vector $p \in \Re_+^{N+T}$, providing the cost of hiring labor in each country and technology-specific prices for all goods. These prices are prices at the point of sale per unit.  Adjusting for transportation costs yields a price for use in any given technology.

As final goods are costless to ship, there is a world price for each final good. Abusing notation we let $p_f$ denote the price of final good $f$. So, in equilibrium $p_{\tau}=p_{\tau'}=p_f$ for all $(\tau,\tau')$ such that $O(\tau)=O(\tau')=f$.

\noindent \textbf{Preferences:}
Laborers are the consumers of the final goods. Consumers have preferences represented by
$$U(c_1, \ldots, c_F)$$ defined on $\Re_+^F$ that is increasing ($c\geq c', c\neq c'$ implies $U(c)>U(c')$)
and strictly quasi-concave and homogeneous of degree $1$, and continuously differentiable on the domain.

As there are no transportation costs on final goods, equilibrium prices for final goods are the same across different countries. Further, as preferences are not country specific, and represented by a utility function that is homogeneous of degree $1$, all agents consume scalings of the same bundle of final goods that are proportional to their wages. Thus, as we show in Lemma \ref{lemma:representative_consumer} in Appendix \ref{sec:proofs}, our formulation admits a representative consumer with preferences represented by $U(\cdot)$.

\noindent \textbf{Time horizons}\label{sec:time_horizons}
We consider the response of the economy to a shock over four different time horizons, which we refer to as the "Short Run", "Medium Run", "Long Run", and "Very Long Run". 
In the very long run, the economy fully adjusts and technologies can be changed. In the long run, technologies are held fixed (so it is as if each economy is only endowed with the technologies that were active when the shock hit), but everything else adjusts. 

If, for example, each country is endowed with a set of technologies for producing a given good that corresponded to a constant elasticity of substitution production function, Then the difference between the long run and very long run would be captured by having the elasticity of substitution be zero in our long run (so the production function becomes Leontief). Adjusting these elasticities to capture the difference between time horizons is an approach taken in some of the literature  (e.g. \cite{bachmann2024if}). There is evidence consistent with the elasticity of substitution being approximately zero for quite considerable lengths of time.\footnote{For example, \cite{BarrotSuppliers2016}, \cite{atalay2017important} and \cite{boehm2019input}.} 

However, this differs completely from our analysis of the short run:  the adjustments are not equilibrium ones in which prices adjust, but shortages on existing contracts. In the short run, prices are unable to adjust and shortages are rationed proportionally to the initial flows.  
For example, an automobile manufacturer does not receive the steering wheels it has contracted for, and cannot simply buy alternatives.  Finding a new supplier that could deliver the custom-designed steering wheels could take months if not years.  The manufacturer can only produce as many cars as it has steering wheels.  
This short run perspective is relevant for most manufacturing, where plants are designed and optimized relative to particular recipes for producing parts and assembling them.  It can take months to make small changes, and years for larger changes.  
Even in the service industry specific recipes are followed.  Fast food franchises have very specific and detailed tasks that are allocated to different workers, with extremely explicit workflows and dedicated inputs.  Shortages of inputs can lead some goods to be completely deleted from menus, or become sold out, for months at a time.

The medium run is a sort of hybrid: labor allocations to technologies are held fixed (so it takes time for workers, if made idle, to find alternative employment, and it takes firms time to hire and train new workers), but prices and flows of goods can adjust. It is equivalent to a constrained equilibrium in which labor allocations and technologies are fixed.

Table \ref{TA:Time_Horizon} summarizes these differences. 

\begin{center}
\begin{tabular}{r|c|c|c|c}
  & Short Run/Temporary & Medium Run & Long Run & Very Long Run \\
  \hline
  Suppliers/Prices & Fixed & Flexible & Flexible & Flexible \\
  Labor & Fixed & Fixed & Flexible & Flexible \\
  Technology & Fixed & Fixed & Fixed & Flexible \\
\label{TA:Time_Horizon}
\end{tabular}
\end{center}
\vspace{0.1in}

While the outcomes obtained in all the other time horizons are either an equilibrium of the economy (very long run) or an appropriately constrained equilibrium for the economy (long run and medium run), in the short run the economy is out of equilibrium, and thus it differs fundamentally from the other time horizons.

For simplicity, and in order to fully consider all the ways in which the economy adjusts, we consider shocks that are permanent (or at least are not anticipated to revert).\footnote{ See \cite{liu2024supply} for adjustments in response to temporary shocks.}

\noindent \textbf{Equilibrium}
We assume the economy is in equilibrium when a shock hits. In this subsection we formally define an equilibrium (which corresponds to the very long run outcome) and also introduce constrained equilibria to capture the economy in the medium and long run. We  define short run disruptions in Section \ref{sec:shortrun}. 

Competitive equilibrium and constant returns to scale imply that there are zero profits, and thus we ignore firm ownership and profits for the sake of eliminating unnecessary notation and definitions.

An \emph{economy} is therefore a list specifying the set of countries, goods, technologies, labor endowments and transportation costs: $(N, M, F, \{T_n\}_n,  {\{L_{n}\}_n}, \theta)$.

An \emph{equilibrium} is defined in the usual way (following \cite{arrow1954existence}) and is a list of prices of each good at each point of sale,  productions from each technology, the inputs (including labor) that it uses, and final consumptions of the consumers/laborers, such that
 (a) laborers supply their endowment of labor inelastically and choose final goods to consume to maximize their preferences,
(b) the output produced by each technology is feasible given the inputs it uses and it maximizes profits over all potential choices of outputs calculated as the revenue from the output minus the costs of corresponding inputs (including transportation costs), and
(c) markets for all goods clear.

Given that details of the equilibrium and existence are standard, we present them in Supplementary Appendix \ref{sec:equilibrium} where we:
(i) offer a fully formal definition of a general equilibrium of an economy,
(ii) show that an equilibrium exists,
(iii) show that in all equilibria the same amount of each final good is produced, and
(iv) show that an equilibrium is fully specified by the flow of country-specific labor to technologies, the flow of goods between all technology pairs and (local) prices for the output of all technologies.

In the long run, a constrained equilibrium is equivalent to an equilibrium defined above but with countries' technology endowments restricted to only include those technologies in use at the time of the shock.  In the medium run, labor allocations across technologies are also fixed.\footnote{More formally, this is equivalent to looking for an equilibrium of the following adjusted economy: (i) Set the number of countries equal to the number of technologies currently in use; (ii) Assign each technology currently in use to a different country; (iii) Endow each country with the labor used by their technology prior to the shock; and (iv) set the iceberg cost to infinity for all flows of labor.}

An equilibrium (or constrained equilibrium) specifies the flow of goods between technologies. It is helpful to represent these flows as a directed, weighted network like that shown in Figure \ref{fig:flow_network}.

\section{Contrasting Short- and Long-Run Impacts of Supply Chain Shocks}

So far, we have normalized the output of a technology $\tau$ producing good $k=O(\tau)$ to be $\tau_k=1$.
In what follows it is convenient to let $\tau_k$ vary at the margin to represent changes in productivity of that  technology.
To identify the impact of a shock, we consider a shock that changes the output of some technology $\tau$, given by $\tau_k$, from its initial value of $1$.

We begin with the {\sl long-run} impact of a shock, showing that Hulten's \citeyearpar{hulten1978growth} Theorem holds in our setting.

\subsection{Hulten's Theorem: The Long-Run Impact of a Change in the Supply Chain}

Let World $GDP$ ($GDP$ for short, henceforth) denote the total expenditures on final goods:\footnote{This is nominal GDP, but prices can be normalized so that the total price of the consumed bundle is 1, and then nominal and real GDP are equivalent. Rescaling all prices in our model (including labor) by some common factor has no effect on the equilibrium.}
$$GDP=\sum_n\sum_f p_f c_{fn}.$$
Because the consumers in different countries have the same homothetic preferences, they demand (potentially) different quantities of the same bundle of goods.
Thus, final total consumer demand equals the demand induced by a representative consumer with the same preferences and wealth equal to total labor income (Lemma \ref{lemma:representative_consumer}, Appendix \ref{sec:proofs}). Therefore, the utility of the representative consumer, denoted by $U$, is a measure of overall welfare and, given the homogeneity of preferences, it is proportional to GDP.

\begin{proposition}[Hulten's Theorem]\label{prop:hulten}
Consider an equilibrium of an economy,
and a technology $\tau$ used in positive amounts in equilibrium to produce good $k=O(\tau)$ such that there is no other technology that can produce the good at the same (transportation-adjusted) price.
The marginal impact on aggregate utility, and on GDP, of a change in the total factor productivity of $\tau$ is equal to the total expenditures on good $k$ produced using technology $\tau$, relative to overall GDP.
That is,
$$
\frac{\partial \log(U)}{\partial \log(\tau_k)} =\frac{\partial  \log(GDP)}{\partial \log(\tau_k)} = \frac{p_{\tau} y_{\tau}}{GDP}.
$$
\end{proposition}

The condition that there is no alternative technology that can produce the good at the same (transportation-adjusted) price 
serves the same role as the condition that there are unique producers for each good (or that all producers of the same good are shocked and use the same equilibrium mix of inputs) that is usually implicitly or explicitly assumed in statements of Hulten's Theorem.
If we relax this condition and allow for duplicate production of the same good at the same effective price and shock only one of the technologies/producers, then this is an upper bound on the impact, and so the equality becomes less than or equal to in the above equation.\footnote{As an example, given costs of transportation, it can be optimal to produce the exactly same good using exactly the same technology separately in different regions to be closer to inputs.  For instance, there are lumber mills using the same technology in many locations around the world to be close to the forests, and much of that lumber is later shipped as it is cheaper to ship the lumber than the trees.  As such, having multiple suppliers of the same input can be a feature of some equilibria.  Nonetheless, those could end up with different transportation-adjusted costs and hence not be perfect substitutes for each other.}
This condition rules out that there are exact duplicates for the shocked technologies and so complete substitution away from the shocks is  costly.\footnote{ \cite{bahal2023beyond} consider larger technology shocks that change a sector's input mix, and aggregate the impact of these shocks into an alternative measure of a sector's importance.}

Hulten's Theorem identifies the long-run marginal effect of a change in the productivity of a technology: it measures the full equilibrium adjustment of the economy to a new equilibrium. It shows that a sufficient statistic for long run marginal impact of a shock to an industry on GDP is the equilibrium value of the shocked industry. A simple intuition for the result is that, at the margin, the reduced productivity is compensated for by sourcing more inputs at their current prices.

We note that there is no adjustment in the statement of Hulten's Theorem for transportation costs.  This is because $y_\tau$ already includes the extra amount produced to account for shipping costs.   That is,  for any intermediate good, $y_\tau = \sum_{\tau''} x_{\tau \tau''}$ and each $x_{\tau \tau''}$ is $\theta_{\tau \tau''}$ times the amount that technology $\tau''$ finally receives as an input.  Technology $\tau$'s price also includes the  transportation costs that accumulated along its supply chain.   

Hulten's Theorem is illustrated in Figure \ref{subfig:LRdisrupt}. 
Before a ten percent decrease in the productivity of Intermediate $R$, the prices of the goods are $(0.1,0.1,0.8,1)$ for Labor, Intermediate $R$, Intermediate $I$ and Final good $F$, respectively. Applying Hulten's Theorem, the long run impact of the shock is approximately $1/50$th of GDP. Accounting for the full equilibrium adjustment in this case results in the updated flows shown in Figure \ref{subfig:LRdisrupt} (to two decimal places), showing that the approximation does a good job in this case.

\begin{figure}
  \subfloat[Long Run Disruption\label{subfig:LRdisrupt}]
{
\centering
\resizebox{3.25in}{2.7in}{	
\begin{tikzpicture}[
		sqRED/.style={rectangle, draw=black!240, fill=red!15, very thick, minimum size=7mm},
		sqBLUE/.style={rectangle, draw=black!240, fill=blue!15, very thick, minimum size=7mm},
		roundnode/.style={ellipse, draw=black!240, fill=green!15, very thick, dashed, minimum size=7mm},
		sqGREEN/.style={rectangle, draw=black!240, fill=green!15, very thick, minimum size=7mm},
		roundRED/.style={ellipse, draw=black!240, fill=red!15, very thick, dashed, minimum size=7mm},
		roundREDb/.style={ellipse, draw=black!240, fill=red!15, very thick, minimum size=7mm},
		roundYELL/.style={ellipse, draw=black!240, fill=yellow!20, very thick, dashed, minimum size=7mm},
		sqYELL/.style={rectangle, draw=black!240, fill=yellow!20, very thick, minimum size=7mm},
		roundORA/.style={ellipse, draw=black!240, fill=orange!15, very thick, dashed, minimum size=7mm},		
		sqORA/.style={rectangle, draw=black!240, fill=orange!15, very thick, minimum size=7mm},
		squaredBLACK/.style={rectangle, draw=black!240, fill=white!15, very thick, minimum size=7mm},
squaredBLACK2/.style={rectangle, draw=red!240, fill=white!15, very thick, dashed, minimum size=7mm},		
		roundBROWN/.style={ellipse, draw=black!240, fill=brown!15, very thick, dashed, minimum size=7mm},
		sqBROWN/.style={rectangle, draw=black!240, fill=brown!15, very thick, minimum size=7mm},		
        sqPURPLE/.style={rectangle, draw=black!240, fill=purple!15, very thick, minimum size=7mm},				
		]
        \node[](S4) at (0, -3.5) {};

		\node[squaredBLACK2](R1) at (-3, 3.6) {\Large Intermediate R, $\cancel{2}$ \textcolor{red}{$1.96$}};
        \node[] at (4.65, 5.35) {\Large $\cancel{2}$ \textcolor{red}{$2.17$}};
        \node[] at (4.5, 3.75) {\Large $\cancel{7}$ \textcolor{red}{$6.85$}};
        \node[] at (4.7, 0.5) {\Large $\cancel{1}$ \textcolor{red}{$.98$}};

		\node[sqORA, align=left](I1) at (2, 1.9) {\Large Intermediate I, $\cancel{1}$ \textcolor{red}{$.98$}};
        \node[] at (-1.8, 2.85) {\Large $\cancel{1}$ \textcolor{red}{$.98$}};

        \node[] at (-1.4, 0.5) {\Large $\cancel{1}$ \textcolor{red}{$.98$}};

        \node[] at (1.35, 0.5) {\Large $\cancel{1}$ \textcolor{red}{$.98$}};

		\node[sqGREEN, align=left](F1) at (2, -2.5) {\Large Final good F,  $\cancel{1}$ \textcolor{red}{$.98$} };

		\node[squaredBLACK, align=left](L) at (12, 6) {\Large Labor, endow 10};

		\tikzset{thick edge/.style={-, black, fill=none, thick, text=black}}
		\tikzset{thick arc/.style={->, black, fill=black, thick, >=stealth, text=black}}

		\draw[line width=0.4mm, black,-> ] (R1) to[out=-90,in=90] (I1);
		
        \draw[line width=0.4mm, black,-> ] (L) to[out=-90,in=90] (R1);
        \draw[line width=0.4mm, black,-> ] (L) to[out=-90,in=90] (I1);
        \draw[line width=0.4mm, black,-> ] (L) to[out=-90,in=90] (F1);		
		
        \draw[line width=0.4mm, black,-> ] (R1) to[out=-90,in=90] (F1);

		\draw[line width=0.4mm, black,-> ] (I1) to[out=-90,in=90] (F1);

		\end{tikzpicture}
}
}
\hfill
  \subfloat[Short Run Disruption\label{subfig:SRdisrupt}]
{
\centering			
\resizebox{3.25in}{2.7in}{	
\begin{tikzpicture}[
		sqRED/.style={rectangle, draw=black!240, fill=red!15, very thick, minimum size=7mm},
		sqBLUE/.style={rectangle, draw=black!240, fill=blue!15, very thick, minimum size=7mm},
		roundnode/.style={ellipse, draw=black!240, fill=green!15, very thick, dashed, minimum size=7mm},
		sqGREEN/.style={rectangle, draw=black!240, fill=green!15, very thick, minimum size=7mm},
		roundRED/.style={ellipse, draw=black!240, fill=red!15, very thick, dashed, minimum size=7mm},
		roundREDb/.style={ellipse, draw=black!240, fill=red!15, very thick, minimum size=7mm},
		roundYELL/.style={ellipse, draw=black!240, fill=yellow!20, very thick, dashed, minimum size=7mm},
		sqYELL/.style={rectangle, draw=black!240, fill=yellow!20, very thick, minimum size=7mm},
		roundORA/.style={ellipse, draw=black!240, fill=orange!15, very thick, dashed, minimum size=7mm},		
		sqORA/.style={rectangle, draw=black!240, fill=orange!15, very thick, minimum size=7mm},
		squaredBLACK/.style={rectangle, draw=black!240, fill=white!15, very thick, minimum size=7mm},
squaredBLACK2/.style={rectangle, draw=red!240, fill=white!15, very thick, dashed, minimum size=7mm},		
		roundBROWN/.style={ellipse, draw=black!240, fill=brown!15, very thick, dashed, minimum size=7mm},
		sqBROWN/.style={rectangle, draw=black!240, fill=brown!15, very thick, minimum size=7mm},		
        sqPURPLE/.style={rectangle, draw=black!240, fill=purple!15, very thick, minimum size=7mm},				
		]
        \node[](S4) at (0, -3.5) {};

		\node[squaredBLACK2](R1) at (-3, 3.6) {\Large Intermediate R,  $\cancel{2}$ \textcolor{red}{$1.8$}};
        \node[] at (4.65, 5.35) {\Large $\bm{2}$};
        \node[] at (4.5, 3.75) {\Large $\bm{7}$};
        \node[] at (4.7, 0.5) {\Large $\bm{1}$ };

		\node[sqORA, align=left](I1) at (2, 1.9) {\Large Intermediate I, $\cancel{1}$ \textcolor{red}{$.9$}};
        \node[] at (-1.8, 2.85) {\Large $\cancel{1}$ \textcolor{red}{$.9$} };

        \node[] at (-1.4, 0.5) {\Large $\cancel{1}$ \textcolor{red}{$.9$}};

        \node[] at (1.35, 0.5) {\Large $\cancel{1}$ \textcolor{red}{$.9$}};

		\node[sqGREEN, align=left](F1) at (2, -2.5) {\Large Final good F, $\cancel{1}$ \textcolor{red}{$.9$}};
			
		\node[squaredBLACK, align=left](L) at (12, 6) {\Large Labor, endow 10};

		\tikzset{thick edge/.style={-, black, fill=none, thick, text=black}}
		\tikzset{thick arc/.style={->, black, fill=black, thick, >=stealth, text=black}}
			
		\draw[line width=0.4mm, black,-> ] (R1) to[out=-90,in=90] (I1);
		
        \draw[line width=0.4mm, black,-> ] (L) to[out=-90,in=90] (R1);
        \draw[line width=0.4mm, black,-> ] (L) to[out=-90,in=90] (I1);
        \draw[line width=0.4mm, black,-> ] (L) to[out=-90,in=90] (F1);		
		
        \draw[line width=0.4mm, black,-> ] (R1) to[out=-90,in=90] (F1);			
		\draw[line width=0.4mm, black,-> ] (I1) to[out=-90,in=90] (F1);

		\end{tikzpicture}
}	
}

\caption{\label{fig:shortcontrast} \footnotesize{An example of the short-run impact of the shock to a technology, and the contrast to the long run.   A 10 percent disruption of the production propagates through the network to the final good.  Even though labor is not
disrupted, it cannot produce the outputs without the corresponding inputs and so final good disruption is disrupted to the full extent of the input disruption.
The disruption is 5 times larger than the corresponding long-run impact.
In the long run, labor reallocates to even out the production needed as inputs downstream. }}
\end{figure}
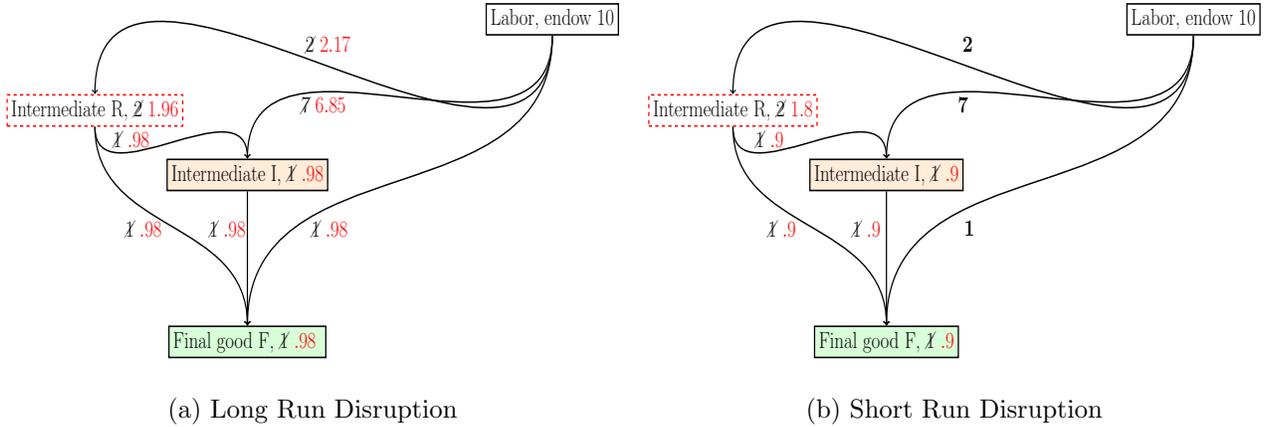

\subsection{Bottlenecks and Disruptions:  Short-Run Impacts of Shocks}
\label{sec:shortrun}

Hulten's Theorem applies in the long run as an approximation for sufficiently small shocks, and also in the very long run (for small shocks\footnote{And also for larger shocks if the set of technologies available to countries to produce a given output corresponds to a Cobb-Douglas production function, in which case the Theorem is exact for non-marginal shocks.}) under additional conditions that limit how substitution works across technologies.  Once one
examines large changes, or the very long run when certain nonlinearities as technology substitution occurs, then adjustments are needed.  
That is well-studied (e.g., see \cite{baqaee2019macroeconomic,baqaee2020productivity}), and thus
we next examine the impact of the technologies when some aspects of production are fixed in the short and medium run.  

In the short run, the impact of changes in productivity
can be much more dramatic and depend on the structure of the supply network.
This answers the point made by Larry Summers in the 2013 speech that was quoted in the introduction, highlighting the difference between the long and short run.
Short-run supply disruptions can be substantial even for items whose costs are a small fraction of GDP,
as we show in this section.

We now examine the short-run such that if there is a shortage of an input, it is rationed so that each customer suffers the same percentage shortfall in supply, and there are no other (compensating) adjustments to the inputs. (In Section \ref{sec:medium}, we return to the medium run in which there are shortages but price adjustments and resourcing.)

To measure the impact of a disruption of some technology, we examine the most that each technology can still produce given the shortage of some input(s) that it faces.  Each technology that sources an input produced by a shocked technology is thus shocked, and so we also trace those impacts as they propagate downstream.
Suppose, for example, that a single technology is used to produce a given final good and the supply chain for this final good technology involves a string of single-sourced inputs without any cycles. Then, the impact of an $X$ percent drop in the output of an upstream good is to reduce the production of the good downstream of it in the supply chain by $X$ percent, and this propagates directly down the chain and disrupts the output of the final good by $X$ percent.

Although the short-run impact in Figure \ref{fig:shortcontrast} is clear in how it works as the shock propagates directly downstream,
there can be three complications in more intricate settings.
First, some technologies downstream of the shock might source the same input from multiple suppliers, and if not all of those sources are shocked, then the reduction can be less than X percent. Second, the same good may be produced using different technologies that require different inputs, only some of which are affected. Third, if there are cycles in the supply network, then these can feed back leading to repeated reductions that can amplify the effect of the shock in a way that helps negate the impact of sourcing and technological diversity.

The impact of a short run shock is found by the following (intuitive) {\sl Shock Propagation Algorithm}:  
\begin{itemize}
\item Begin with the equilibrium flows and final good outputs $(x_{n\tau'})_{n\tau'},(x_{\tau\tau'})_{\tau\tau'}, (y_{\tau})_{O(\tau)\in F} $
\item Reduce the outputs of the shocked technologies $\tau\in T^{Shocked}$ by $(1-\lambda)$, so $y^1_\tau = \lambda y_{\tau}$ and $x^1_{\tau\tau'}=\lambda x_{\tau\tau'}$, and leave
other flows and outputs are unchanged (so,  $y^1_\tau =  y_{\tau}$ and $x^1_{\tau\tau'}= x_{\tau\tau'}$ for all $\tau\notin T^{Shocked}$ and $\tau'$).
\item Iteratively in $k$ (starting from $k=2$):
\begin{itemize}
\item   
let $y^k_{\tau'} = y_{\tau'}\left(\min_{\tau:x_{\tau\tau'}>0  }  \frac{x^{k-1}_{\tau \tau'}}{x_{\tau\tau'}} \right)$ and
$x^k_{\tau' \tau''} = \frac{y^k_{\tau'}}{y_{\tau'}} x_{\tau'\tau''}$ for all $\tau'$ and $\tau''$.
\end{itemize}
\item Iterate until $y^{k-1}=y^k$ (or  $\max_\tau (y^k_\tau-y^{k-1}_\tau) \leq\Delta$  for some threshold $\Delta\geq 0$).
\end{itemize}
At each stage, we examine nodes downstream from any shocked node and calculate what it can produce given the available inputs.
Any nodes that have outputs change are then shocked, and we repeat the calculation.
In particular, we calculate, for each affected input, the proportion of the original supplied level of that input type which can still be successfully sourced from the shocked flow network. From this, we calculate how much the output of each node declines, and continue to trace affected technologies.   The overall impact on the economy is then given by the value of lost final good production resulting from the shock(s). Note that this does not allow for the redistribution of labor and so some labor becomes idle. This reduces the labor income of the corresponding representative consumer.

The Shock Propagation Algorithm is illustrated in Figure \ref{shockalg} in Supplementary Appendix \ref{sec:extended_eg}. If there are no cycles, the shock propagation algorithm terminates in a finite number of steps at a fixed point which is the unique solution to a minimum disruption problem. If there are cycles, the algorithm may not terminate in finite time. However, in that case, it still converges to the fixed point which is a solution to a minimum disruption problem. This is formalized in Proposition \ref{prop:shock_prop_alg}.

Although the short-run outcome is out of equilibrium, a variation on the First Welfare Theorem nevertheless applies to it: the flows that the shock propagation algorithm converge to solve a suitably constrained value maximization problem. We call this problem the minimum disruption problem. 

Starting from an equilibrium outputs and flows $(y_{\tau})_{\tau}$ and $(x_{\tau\tau'})_{\tau\tau'}$, respectively, let $T^{Active}:=\{\tau\in T :y_{\tau}>0\}$ be the set of active technologies.
Consider a shock that reduces the outputs of technologies $T^{Shocked}\subseteq T^{Active}$ to $\lambda<1$ of their initial level.
Consider outputs and flows  $(\widehat{y}_{\tau})_{\tau}$ and $(\widehat{x}_{\tau\tau'})_{\tau\tau'}$
that are the solution to the following minimum disruption problem:

$$\max_{(\widehat{x}_{\tau\tau'})_{\tau\tau'}} \sum_{\tau:O(\tau)\in F} p_{\tau}\widehat{y}_{\tau}$$ subject to
\begin{align*}
\widehat{y}_{\tau}&\leq \lambda y_{\tau} \quad  \text{ for all }\tau \in T^{Shocked},     & \text{(shock constraints)}\\
\widehat{y}_{\tau}&\leq \left(\min_{k:\tau_k<0} \frac{\sum_{\tau':O(\tau')=k}\widehat{x}_{\tau'\tau}}{\sum_{\tau':O(\tau')=k} x_{\tau'\tau}} \right)y_{\tau}  \quad  \text{ for all }\tau\in T^{Active},     &  \quad \text{(technology constraints)}\\
\widehat{y}_{\tau}&= y_{\tau}=0  \quad  \text{ for all }\tau\not\in T^{Active},     &  \quad \text{(technology switching constraints)}\\
\widehat{x}_{\tau'\tau}&=x_{\tau'\tau} \left(\frac{\widehat{y}_{\tau'}}{y_{\tau'}}\right)   \quad \text{ for all }\tau',\tau \in T^{Active}.   &    \text{(proportional rationing)}
\end{align*}

The minimum disruption problem defines the maximum final production that can still be produced, subject to the reduced output of directly-shocked technologies, as well as technology constraints
that do not allow new sources of inputs and are based on proportional rationing constraints that equally spread the impact of each technology's reduced production among its customers.

Proportional rationing applies in the short run in which firms are not able to seek out new suppliers, renegotiate contracts, or switch technologies to substitute disrupted inputs for non-disrupted ones.

\begin{proposition}\label{prop:shock_prop_alg}
The Shock Propagation Algorithm (with $\Delta=0$, which may continue ad infinitum) converges to a flow of goods that is weakly lower for all links, and strictly lower for all links on a directed path from any shocked node.
The limit output vector solves the minimum disruption problem.
\end{proposition}

In the medium run prices, and thus the rationing, adjust. We explore the extent to which this mitigates short run disruptions in Section \ref{sec:medium}. An implication of Proposition \ref{prop:shock_prop_alg} is that allowing for flexible prices will (generically) reduce overall disruption because it relaxes the proportional rationing constraint.

Proposition \ref{prop:shock_prop_alg} shows that the Shock Propagation Algorithm converges to a solution of the minimum disruption problem. The uniqueness of the limit, and that it solves the minimum disruption, follow from several facts. First, the algorithm iterates on the flows between technology pairs, and these flows can be represented by vectors. Bounding each entry of the vector from above by the initial equilibrium flow, and from below by $0$, the vectors can then be partially ordered such that one vector is weakly ordered above another when all flows are weakly higher. Moreover, this partially ordered set is a complete lattice. Second, reductions in inputs in the short run are complementary to each other, so one iteration of the algorithm outputs weakly higher flows when the initial flows are weakly higher. As such, an iteration of the algorithm is an isotone (monotonic) function, and Tarski's fixed point theorem tells us that the fixed points of this mapping form a complete lattice (given the same partial order). Third, the algorithm terminates at/converges to a fixed point in which no more reductions are necessary, and every reduction implemented by the algorithm is necessary. Thus the highest possible production levels that satisfy all the constraints is the limit, and the algorithm finds a solution to the minimum disruption problem. The full proof, including the infinite case appears in the appendix.  It is worth noting that flows that solve the minimum disruption problem solution are not uniquely tied down: there could be several flows that are changing in combination and only some of them end up binding the production, and so nonbinding flows can be lowered and not change overall production.
One can check directly that convex combinations of flows that solve the problem also solve the problem, and hence given that the flows that can solve the problem are bounded above and below, the set of flows that solve the problem form a complete lattice, of which the Shock Propagation Algorithm finds the maximum.

We have specified the minimum disruption problem to minimize the lost value of final goods at the initial equilibrium prices. One might wonder whether the new long-run equilibrium prices should be used instead, or perhaps some other prices.
It turns out that the flows that solve the minimum disruption problem are invariant to such choices.
This is because proportional rationing is imposed in the problem, and this determines the final good levels.
In Section \ref{sec:medium} we show that this can change when that proportionality constraint is removed. 

\begin{observation}\label{cor:minimum_disruption_price_invariance}
Flows that solve the minimum disruption problem remain a solution as prices in the objective function are changed, and the total production of each final good is the same because this solution maximizes output good by good over flows consistent with the constraints.
\end{observation}

It is worth reflecting on our model of the short-run and contrasting it with the existing literature.  A typical way to capture the difference between the short-run and long-run is via the elasticity of substitution (e.g., \cite{Moll2022}). By making this more inelastic in the short run, less adjustment in the input mix 
occurs as one changes productivity parameters and compares across equilibria. However, this is a different approach
that answers a fundamentally different question from our analysis that constrains production. Our assumption that firms are unable to change technologies and hence are stuck with a short-run production function that is Leontief captures this lack of substitution, and one could use this to compare outcomes across different (long-run) equilibria as productivity parameters change.
However, that is not what we are doing and is not what is driving the difference between the short- and long-run in our model. To emphasize this, consider an example where a pandemic shuts down production of some input for a quarter of the year.  So, 1/4 of the annual output of this input is lost.  In our model that cascades through the supply chain and up to 1/4 of the output of downstream goods using this as an input (directly or indirectly) are lost.  If one uses does not treat this as a constraint, but recalculates equilibrium with a new productivity parameter of 3/4 of what it used to be, the solution would be to scale up production of this good to compensate for its lower productivity.  That would come at some cost (e.g., labor cost that is pulled from elsewhere in the economy), and would lead to lower production, but a (possibly much) smaller impact than the 1/4 lost production that our model generates.  But in a pandemic that scaling up would be impossible: the production is constrained by the labor restrictions of the pandemic and there is a fundamental cap put on the production function that operates differently from a parameter change and a new equilibrium.  This applies to many short-term shocks such as natural disasters, sanctions, political crises, etc.  It also captures situations where equilibrium adjustments and reallocation of labor, etc., across different parts of production take some time (which is true of many settings), so that the economy is at least temporarily out of equilibrium due to an unanticipated shock.  This is not something that can be captured just by changing elasticities in a production function, but needs to be modeled as short-run constraints.  

To further contrast the approaches, consider a temporary shock. A standard approach here would be to adjust to a new equilibrium in the short run for the duration of the shock, and then to adjust back to the original equilibrium in the long run, perhaps slowly if there are adjustment costs (see, for example, \cite{liu2024supply}). However, our short-run analysis would let the economy be out of equilibrium, and adjustments back to the long-run equilibrium would occur as short-run constraints are relaxed and the economy can adjust on additional dimensions. Including the constraints leads to markedly different short-run outcomes, and identifies different points in the supply network as crucial.

Given the monotonicity of the algorithm, from it we can also deduce an upper bound on the impact of a shock that is easy to calculate, applies in many cases of interest, and contrasts starkly with Hulten's Theorem.

Consider a shock to some subset of the technologies $T^{Shocked}$. Let $ F(T^{Shocked})$ denote the set of final good technologies that lie downstream (on a directed path) from a shocked technology, so either (i) are in $T^{Shocked}$ directly or (ii) use an input good from a technology (directly or indirectly) that is shocked.
These are the final goods that ever change at some stage $k$ of the algorithm.  Thus $ F(T^{Shocked})$ gives the set of all final good technologies that are impacted by the shock.

\begin{proposition}\label{prop:shocked_GDP}
Consider a shock that reduces the output of all technologies $k\in T^{Shocked}$ to $\lambda<1$ of their original levels. Then the proportion of GDP that is lost to this shock is bounded above by $$(1-\lambda) \left(\frac{\sum_{\tau \in F(T^{Shocked})} {p}_{\tau} y_{\tau}}{GDP}\right).$$
\end{proposition}

Proposition \ref{prop:shocked_GDP} provides an upper bound on the impact of a shock that is tied to the value of all final goods related to the shocked technologies, rather than the cost of production of the shocked technologies, which is what matters in the long-run by Hulten's Theorem. We turn now to identifying sufficient conditions under which the upper bound is achieved.

Consider the sub-network which describes the supply chains of all final goods that are impacted by the shock. Specifically, the \emph{disrupted industries sub-network}, ${{\mathcal{G}}}(T^{Shocked})$, is the sub-network induced on ${\mathcal{G}}$ by all technologies that are on a directed path that terminates at a final good technology in $F(T^{Shocked})$. Figure \ref{fig:disruption} in Supplementary Appendix \ref{sec:extended_eg} illustrates a disrupted industries sub-network. 

When a \emph{disrupted industries sub-network}, ${{\mathcal{G}}}(T^{Shocked})$, is acyclic, there exist a set of nodes (technologies) that have no in-links. We denote these technologies $R(T^{Shocked})$. In this case, it is helpful to introduce the notion of a \emph{(s,t)-cut set}.
For a directed network with disjoint sets of nodes $s$ and $t$, an $(s,t)$-cut set is a set of edges that when removed from the network there are no remaining paths between $s$ and $t$.

\begin{proposition}\label{prop:cut}
If the disrupted industries sub-network ${{\mathcal{G}}}(T^{Shocked})$ is acyclic and the set of edges adjacent to the set of shocked technologies $T^{Shocked}$ forms an $( R(T^{Shocked}), F(T^{Shocked}))$-cut set of the disrupted industries sub-network ${{\mathcal{G}}}(T^{Shocked})$,\footnote{Note that this permits the set of shocked technologies $T^{Shocked}$ to intersect both $R(T^{Shocked})$ and $ F(T^{Shocked})$.} then the bound from Proposition \ref{prop:shocked_GDP} binds.
\end{proposition}

Proposition \ref{prop:cut} shows that the upper bound for the impact of a shock from Proposition \ref{prop:shocked_GDP} is tight when the technology network is acyclic and removing the
edges adjacent to the shocked technologies $T^{Shocked}$ constitutes a cut on the network that restricts attention to related goods (i.e., disconnects raw materials from their related final goods). The example shown in Figure \ref{subfig:disruption-1} in Supplementary Appendix \ref{sec:extended_eg} illustrates a shock which, by Proposition \ref{prop:shocked_GDP}, must obtain the bound.

Although the shocked technologies constituting a cut set is sufficient for the bound to apply, it is not necessary. Figure \ref{shockalg} in Supplementary Appendix \ref{sec:extended_eg} provides an example of a network in which the bound holds when the shocked technologies are not a cut set and the supply network contains cycles. Further, a variation on Figure \ref{fig:shortcontrast}(a) shows that even when the technology network is acyclic, the bound can be obtained when the edges adjacent to the shocked technologies do not constitute a cut set. To see this simply relabel the Labor node to be a technology, and the shocked technologies do not constitute a cut set.
Similarly, one can easily construct examples where the bound holds that have extraneous cycles that are irrelevant to the cut set.  Ruling out cycles is helpful in our partially constructive proof.

We say there is \emph{no technological diversity} if all technologies for producing any given good use the same set of inputs (albeit possibly in different ratios or with different efficiencies). We say there are \emph{industry-wide shocks} (as opposed to technology-specific shocks) if for any shocked technology $\tau\in T^{Shocked}$ all producers of good $O(\tau)$ are also in $T^{Shocked}$.

\begin{proposition}\label{prop:general_industry_shocks}
If there is no technological diversity and there are industry-wide shocks, then the bound from Proposition \ref{prop:shocked_GDP} binds.
\end{proposition}

Proposition \ref{prop:general_industry_shocks} does not follow from Proposition \ref{prop:cut} as it allows for cyclic production and also for cases in which the shocked technologies do not constitute a cut, as there could be other technologies that are not affected that are on separate paths to the affected
final goods.

The proof is straightforward, and so we just sketch it.
Consider a final good technology $\tau\in F(T^{Shocked})$. By definition of $F(T^{Shocked})$, there is a path from a shocked technology $\tau' \in T^{Shocked}$ to $\tau$. Consider that path, starting with a shocked technology $\tau'$ that is at maximal distance from a final good. The output of $\tau'$ is reduced to $\lambda$ of its initial level as it is shocked. Consider the next technology $\tau''$ on the path, which sources good $O(\tau')$ from $\tau'$. As shocks are industry wide, all producers of good $O(\tau')$ are shocked, and so $\tau''$ is only able to source $\lambda$ of the initial amount of good $O(\tau')$ that it sourced, and thus its output will be reduced to $\lambda$ its initial level.  As all shocked inputs are shocked to $\lambda$ of their initial levels, it does not matter if just one input or multiple inputs are shocked.  Moreover, as there is no technological diversity, all producers of good $O(\tau'')$ also use good $O(\tau')$ as an input and are only able to source $\lambda$ units per unit they initially source of this input. Thus all producers of good $O(\tau'')$ have their output reduced to $\lambda$ their initial levels. This implies that the next technology on path, $\tau'''$, also has its output reduced to $\lambda$ of its initial level and so on. Thus the shock propagates down the path, reducing output to $\lambda$ of its initial level at each step, until the output of final good $\tau\in F(T^{Shocked})$ has its output reduced to $\lambda$ of its initial level. This implies that all final goods in $F(T^{Shocked})$ have their output reduced to $\lambda$ of their initial level, and so bound from Proposition \ref{prop:shocked_GDP} is obtained.

Proposition \ref{prop:general_industry_shocks} highlights the value of diversity in production technologies and the location of industries. For example, if all production of a given good is located in the same country, then country-specific shocks become industry shocks, and by Proposition \ref{prop:general_industry_shocks} the upper bound is obtained.

\subsection{Contrasting the Short- and Long-Run Impacts of Shocks}

Comparing Figures \ref{subfig:SRdisrupt} and \ref{subfig:LRdisrupt} gave us some idea of the differences that can occur between the long and the short run. This comparison is to a version of the long run in which we endow each firm with a unique (Leontief) technology and do not allow for  substitutability in the input mix. 
Substitutability in the input mix would dampen the long-run impact of the shock further, making the contrast with the short-run even starker.  Our approach highlights the difference between our short-run out-of-equilibrium analysis and models that compare the long-run and short-run to each other by reducing the elasticities of substitution between inputs. 

The contrast between the long- and short-run is further reinforced when the bound from Proposition \ref{prop:shocked_GDP} is tight. For example, for an industry-wide shock, the long-run (marginal) impact on GDP is, by Hulten's theorem, proportional to the value of the output of the affected industry. By contrast, if there is no technological diversity, the short-run impact on GDP is proportional to the value of the output of all final goods industries that use the output of the affected industry directly or indirectly by Proposition \ref{prop:general_industry_shocks}. If, for example, production of several final good depends on a basic type of computer chip, and there is a 20 percent disruption in the supply of these computer chips, then 20 percent of final good production would be lost for all affected final goods. This can constitute a substantial short run impact, while in the long run the impact on GDP would only be 20 percent of the amount spent on these basic computer chips, which could be tiny in comparison.

An important point to make is that the long-run impact of a shock is dependent only on the expenditures on the shocked technology, and not on the details of the network beyond that.\footnote{Again, those expenditures are dependent on the network, so this is not to say that the network is irrelevant.  It simply says that the information needed to determine the impact is captured by a very simple sufficient statistic.}
To see this, consider an economy with two variations of the sets of technologies from the example in Figures \ref{fig:flow_network} and \ref{fig:shortcontrast}, but, for simplicity, let
us omit the labor inputs. As we see in Figure \ref{fig:ContLR},  the long-run impact of the 10 percent shock to technology $\tau1$ is the same regardless of whether which combination of technologies is used to produce which final good. The new long-run equilibrium differs across the two networks, but the ultimate impact of the shock on GDP does not.

\begin{figure}[!ht]
    \subfloat[Impact is 1/100th of GDP\label{subfig:LR1}]
    {
\centering	
\resizebox{2.75in}{1.65in}{		
\begin{tikzpicture}[
		sqRED/.style={rectangle, draw=black!240, fill=red!15, very thick, minimum size=7mm},
		sqBLUE/.style={rectangle, draw=black!240, fill=blue!15, very thick, minimum size=7mm},
		roundnode/.style={ellipse, draw=black!240, fill=green!15, very thick, dashed, minimum size=7mm},
		sqGREEN/.style={rectangle, draw=black!240, fill=green!15, very thick, minimum size=7mm},
		roundRED/.style={ellipse, draw=black!240, fill=red!15, very thick, dashed, minimum size=7mm},
		roundREDb/.style={ellipse, draw=black!240, fill=red!15, very thick, minimum size=7mm},
		roundYELL/.style={ellipse, draw=black!240, fill=yellow!20, very thick, dashed, minimum size=7mm},
		sqYELL/.style={rectangle, draw=black!240, fill=yellow!20, very thick, minimum size=7mm},
		roundORA/.style={ellipse, draw=black!240, fill=orange!15, very thick, dashed, minimum size=7mm},		
		sqORA/.style={rectangle, draw=black!240, fill=orange!15, very thick, minimum size=7mm},
		squaredBLACK/.style={rectangle, draw=black!240, fill=white!15, very thick, minimum size=7mm},	
		squaredBLACK2/.style={rectangle, draw=red!240, fill=white!15, very thick, dashed, minimum size=7mm},	
		roundBROWN/.style={ellipse, draw=black!240, fill=brown!15, very thick, dashed, minimum size=7mm},
		sqBROWN/.style={rectangle, draw=black!240, fill=brown!15, very thick, minimum size=7mm},		
        sqPURPLE/.style={rectangle, draw=black!240, fill=purple!15, very thick, minimum size=7mm},				
		]

		\node[squaredBLACK2](R1) at (-2, 2.5) {$\tau1$, $\cancel{2}$ \textcolor{red}{$1.96$}};
		\node[squaredBLACK](R2) at (2, 2.5) {$\tau2$, $2$};

        \node[] at (1.25, 1.25) { $1$};
        \node[] at (-1.25, 1.25) { $\cancel{1}$ \textcolor{red}{$.98$}};

		\node[sqORA](I1) at (-4, 0) { $\tau3$, $\cancel{1}$ \textcolor{red}{$.98$}};
        \node[] at (-3.9, -1.25) { $\cancel{1}$ \textcolor{red}{$.98$}};
		\node[] at (-3.9, 1.25) { $\cancel{1}$ \textcolor{red}{$.98$}};

        \node[sqORA](I2) at (4, 0) { $\tau4$, $1$};
        \node[] at (3.65, -1.25) { $1$};
		\node[] at (3.65, 1.25) { $1$};

		\node[sqGREEN](F1) at (-2, -2.5) { $\tau5$, $\cancel{1}$ \textcolor{red}{$.98$}};
		\node[sqGREEN](F2) at (2, -2.5) { $\tau6$, $1$};
		
		\tikzset{thick edge/.style={-, black, fill=none, thick, text=black}}
		\tikzset{thick arc/.style={->, black, fill=black, thick, >=stealth, text=black}}

		\draw[line width=0.4mm, black,-> ] (R1) to[out=-90,in=90] (I1);
		\draw[line width=0.4mm, black,-> ] (R2) to[out=-90,in=90] (I2);

		\draw[line width=0.4mm, black,-> ] (I1) to[out=-90,in=90] (F1);
		\draw[line width=0.4mm, black,-> ] (I2) to[out=-90,in=90] (F2);

		\draw[line width=0.4mm, black,-> ] (R1) to[out=-90,in=90] (F1);
		\draw[line width=0.4mm, black,-> ] (R2) to[out=-90,in=90] (F2);

		\end{tikzpicture}	
		}
	}
    \hfill
    \subfloat[Impact is 1/100th of GDP\label{subfig:LR2}]
    {
     \centering
\resizebox{2.75in}{1.65in}{		
\begin{tikzpicture}[
		sqRED/.style={rectangle, draw=black!240, fill=red!15, very thick, minimum size=7mm},
		sqBLUE/.style={rectangle, draw=black!240, fill=blue!15, very thick, minimum size=7mm},
		roundnode/.style={ellipse, draw=black!240, fill=green!15, very thick, dashed, minimum size=7mm},
		sqGREEN/.style={rectangle, draw=black!240, fill=green!15, very thick, minimum size=7mm},
		roundRED/.style={ellipse, draw=black!240, fill=red!15, very thick, dashed, minimum size=7mm},
		roundREDb/.style={ellipse, draw=black!240, fill=red!15, very thick, minimum size=7mm},
		roundYELL/.style={ellipse, draw=black!240, fill=yellow!20, very thick, dashed, minimum size=7mm},
		sqYELL/.style={rectangle, draw=black!240, fill=yellow!20, very thick, minimum size=7mm},
		roundORA/.style={ellipse, draw=black!240, fill=orange!15, very thick, dashed, minimum size=7mm},		
		sqORA/.style={rectangle, draw=black!240, fill=orange!15, very thick, minimum size=7mm},
		squaredBLACK/.style={rectangle, draw=black!240, fill=white!15, very thick, minimum size=7mm},	
		squaredBLACK2/.style={rectangle, draw=red!240, fill=white!15, very thick, dashed, minimum size=7mm},	
		roundBROWN/.style={ellipse, draw=black!240, fill=brown!15, very thick, dashed, minimum size=7mm},
		sqBROWN/.style={rectangle, draw=black!240, fill=brown!15, very thick, minimum size=7mm},		
        sqPURPLE/.style={rectangle, draw=black!240, fill=purple!15, very thick, minimum size=7mm},				
		]

		\node[squaredBLACK2](R1) at (-2, 2.5) {$\tau1$, $\cancel{2}$ \textcolor{red}{$1.98$}};
		\node[squaredBLACK](R2) at (2, 2.5) {$\tau2$,$\cancel{2}$ \textcolor{red}{$1.98$}};

        \node[] at (1.05, 1.25) { $\cancel{1}$ \textcolor{red}{$.99$}};
        \node[] at (-1.05, 1.25) { $\cancel{1}$ \textcolor{red}{$.99$}};

		\node[sqORA](I1) at (-4, 0) { $\tau3$, $\cancel{1}$ \textcolor{red}{$.99$}};
        \node[] at (-3.9, -1.25) { $\cancel{1}$ \textcolor{red}{$.99$}};
		\node[] at (-3.9, 1.25) { $\cancel{1}$ \textcolor{red}{$.99$}};

        \node[sqORA](I2) at (4, 0) { $\tau4$, $\cancel{1}$ \textcolor{red}{$.99$}};
        \node[] at (3.9, -1.25) { $\cancel{1}$ \textcolor{red}{$.99$}};
		\node[] at (3.9, 1.25) { $\cancel{1}$ \textcolor{red}{$.99$}};

		\node[sqGREEN](F1) at (-2, -2.5) { $\tau5$, $\cancel{1}$ \textcolor{red}{$.99$}};
		\node[sqGREEN](F2) at (2, -2.5) { $\tau6$,$\cancel{1}$ \textcolor{red}{$.99$}};

		\tikzset{thick edge/.style={-, black, fill=none, thick, text=black}}
		\tikzset{thick arc/.style={->, black, fill=black, thick, >=stealth, text=black}}

		\draw[line width=0.4mm, black,-> ] (R1) to[out=-90,in=90] (I1);
		\draw[line width=0.4mm, black,-> ] (R2) to[out=-90,in=90] (I2);
			
		\draw[line width=0.4mm, black,-> ] (I1) to[out=-90,in=90] (F1);
		\draw[line width=0.4mm, black,-> ] (I2) to[out=-90,in=90] (F2);

		\draw[line width=0.4mm, black,-> ] (R1) to[out=-90,in=90] (F2);
		\draw[line width=0.4mm, black,-> ] (R2) to[out=-90,in=90] (F1);

		\end{tikzpicture}
  }	
		}
    \caption{\footnotesize{ \label{fig:ContLR} Long run:  In both cases have supply networks that have two copies of technologies similar to those in the example from Figure \ref{fig:flow_network}.
    Each final good needs one resource and one intermediate good, but which combination of inputs are needed downstream differs between the networks.
    In the long run the details of the network structure do not matter if the amount spent on the shocked technology is the same.
    In both cases the labor endowment is 20, the initial prices are $p=\left(\frac{1}{10},\frac{1}{10},\frac{4}{5},1\right)$ and
$GDP=\sum_f p_f c_f=2$.
Thus, from Hulten's Theorem, the marginal impact is
$\frac{p_{R1} y_{R1}}{\text{GDP}}= \frac{1}{10}$
and then extrapolating for a $10\%$ shock,
the {\sl long-run} impact is $1/100$th of GDP.  We do see, however, that the new long-run equilibrium flows differ across the two variations, but the GDP impact is similar.  }}
  \end{figure}
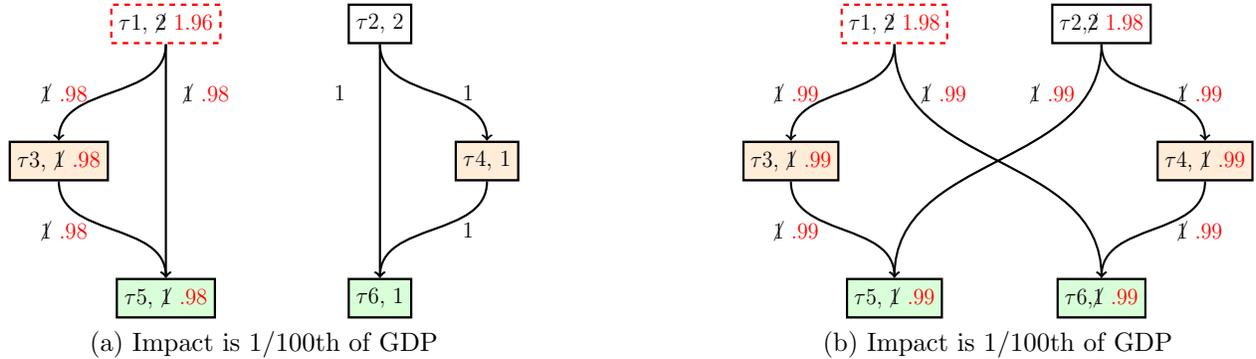

In contrast, the short-run impact in the same economies differs substantially depending on the details of the network as we see in Figure \ref{fig:ContSR}.

\begin{figure}[!ht]
    \subfloat[Impact is 1/20th of GDP\label{subfig:SR1}]
    {
\centering	
\resizebox{2.75in}{1.65in}{	
\begin{tikzpicture}[
		sqRED/.style={rectangle, draw=black!240, fill=red!15, very thick, minimum size=7mm},
		sqBLUE/.style={rectangle, draw=black!240, fill=blue!15, very thick, minimum size=7mm},
		roundnode/.style={ellipse, draw=black!240, fill=green!15, very thick, dashed, minimum size=7mm},
		sqGREEN/.style={rectangle, draw=black!240, fill=green!15, very thick, minimum size=7mm},
		roundRED/.style={ellipse, draw=black!240, fill=red!15, very thick, dashed, minimum size=7mm},
		roundREDb/.style={ellipse, draw=black!240, fill=red!15, very thick, minimum size=7mm},
		roundYELL/.style={ellipse, draw=black!240, fill=yellow!20, very thick, dashed, minimum size=7mm},
		sqYELL/.style={rectangle, draw=black!240, fill=yellow!20, very thick, minimum size=7mm},
		roundORA/.style={ellipse, draw=black!240, fill=orange!15, very thick, dashed, minimum size=7mm},		
		sqORA/.style={rectangle, draw=black!240, fill=orange!15, very thick, minimum size=7mm},
		squaredBLACK/.style={rectangle, draw=black!240, fill=white!15, very thick, minimum size=7mm},	
		squaredBLACK2/.style={rectangle, draw=red!240, fill=white!15, very thick, dashed, minimum size=7mm},	
		roundBROWN/.style={ellipse, draw=black!240, fill=brown!15, very thick, dashed, minimum size=7mm},
		sqBROWN/.style={rectangle, draw=black!240, fill=brown!15, very thick, minimum size=7mm},		
        sqPURPLE/.style={rectangle, draw=black!240, fill=purple!15, very thick, minimum size=7mm},				
		]

		\node[squaredBLACK2](R1) at (-2, 2.5) {$\tau1$, $\cancel{2}$ \textcolor{red}{$1.8$}};
		\node[squaredBLACK](R2) at (2, 2.5) {$\tau2$, $2$};

        \node[] at (1.5, 1.25) { $1$};
        \node[] at (-1.5, 1.25) { $\cancel{1}$ \textcolor{red}{$.9$}};

		\node[sqORA](I1) at (-4, 0) { $\tau3$, $\cancel{1}$ \textcolor{red}{$.9$}};
        \node[] at (-3.9, -1.25) { $\cancel{1}$ \textcolor{red}{$.9$}};
		\node[] at (-3.9, 1.25) { $\cancel{1}$ \textcolor{red}{$.9$}};

        \node[sqORA](I2) at (4, 0) { $\tau4$, $1$};
        \node[] at (3.65, -1.25) { $1$};
		\node[] at (3.65, 1.25) { $1$};

		\node[sqGREEN](F1) at (-2, -2.5) { $\tau5$, $\cancel{1}$ \textcolor{red}{$.9$}};
		\node[sqGREEN](F2) at (2, -2.5) { $\tau6$, ${1}$ };

		\tikzset{thick edge/.style={-, black, fill=none, thick, text=black}}
		\tikzset{thick arc/.style={->, black, fill=black, thick, >=stealth, text=black}}

		\draw[line width=0.4mm, black,-> ] (R1) to[out=-90,in=90] (I1);
		\draw[line width=0.4mm, black,-> ] (R2) to[out=-90,in=90] (I2);

		\draw[line width=0.4mm, black,-> ] (I1) to[out=-90,in=90] (F1);
		\draw[line width=0.4mm, black,-> ] (I2) to[out=-90,in=90] (F2);

		\draw[line width=0.4mm, black,-> ] (R1) to[out=-90,in=90] (F1);
		\draw[line width=0.4mm, black,-> ] (R2) to[out=-90,in=90] (F2);
		
		\end{tikzpicture}		
		}
	}
    \hfill
    \subfloat[Impact is 1/10th of GDP\label{subfig:SR2}]
    {
     \centering
\resizebox{2.75in}{1.65in}{		
\begin{tikzpicture}[
		sqRED/.style={rectangle, draw=black!240, fill=red!15, very thick, minimum size=7mm},
		sqBLUE/.style={rectangle, draw=black!240, fill=blue!15, very thick, minimum size=7mm},
		roundnode/.style={ellipse, draw=black!240, fill=green!15, very thick, dashed, minimum size=7mm},
		sqGREEN/.style={rectangle, draw=black!240, fill=green!15, very thick, minimum size=7mm},
		roundRED/.style={ellipse, draw=black!240, fill=red!15, very thick, dashed, minimum size=7mm},
		roundREDb/.style={ellipse, draw=black!240, fill=red!15, very thick, minimum size=7mm},
		roundYELL/.style={ellipse, draw=black!240, fill=yellow!20, very thick, dashed, minimum size=7mm},
		sqYELL/.style={rectangle, draw=black!240, fill=yellow!20, very thick, minimum size=7mm},
		roundORA/.style={ellipse, draw=black!240, fill=orange!15, very thick, dashed, minimum size=7mm},		
		sqORA/.style={rectangle, draw=black!240, fill=orange!15, very thick, minimum size=7mm},
		squaredBLACK/.style={rectangle, draw=black!240, fill=white!15, very thick, minimum size=7mm},	
		squaredBLACK2/.style={rectangle, draw=red!240, fill=white!15, very thick, dashed, minimum size=7mm},	
		roundBROWN/.style={ellipse, draw=black!240, fill=brown!15, very thick, dashed, minimum size=7mm},
		sqBROWN/.style={rectangle, draw=black!240, fill=brown!15, very thick, minimum size=7mm},		
        sqPURPLE/.style={rectangle, draw=black!240, fill=purple!15, very thick, minimum size=7mm},				
		]

		\node[squaredBLACK2](R1) at (-2, 2.5) {$\tau1$, $\cancel{2}$ \textcolor{red}{$1.8$}};
		\node[squaredBLACK](R2) at (2, 2.5) {$\tau2$, $2$};

        \node[] at (1.15, 1.25) { $1$};
        \node[] at (-1, 1.25) { $\cancel{1}$ \textcolor{red}{$.9$}};

		\node[sqORA](I1) at (-4, 0) { $\tau3$, $\cancel{1}$ \textcolor{red}{$.9$}};
        \node[] at (-3.9, -1.25) { $\cancel{1}$ \textcolor{red}{$.9$}};
		\node[] at (-3.9, 1.25) { $\cancel{1}$ \textcolor{red}{$.9$}};

        \node[sqORA](I2) at (4, 0) { $\tau4$, $1$};
        \node[] at (3.65, -1.25) { $1$};
		\node[] at (3.65, 1.25) { $1$};

		\node[sqGREEN](F1) at (-2, -2.5) {$\tau5$, $\cancel{1}$ \textcolor{red}{$.9$}};
		\node[sqGREEN](F2) at (2, -2.5) {$\tau6$, $\cancel{1}$ \textcolor{red}{$.9$}};
		
		\tikzset{thick edge/.style={-, black, fill=none, thick, text=black}}
		\tikzset{thick arc/.style={->, black, fill=black, thick, >=stealth, text=black}}

		\draw[line width=0.4mm, black,-> ] (R1) to[out=-90,in=90] (I1);
		\draw[line width=0.4mm, black,-> ] (R2) to[out=-90,in=90] (I2);
	
		\draw[line width=0.4mm, black,-> ] (I1) to[out=-90,in=90] (F1);
		\draw[line width=0.4mm, black,-> ] (I2) to[out=-90,in=90] (F2);

		\draw[line width=0.4mm, black,-> ] (R1) to[out=-90,in=90] (F2);
		\draw[line width=0.4mm, black,-> ] (R2) to[out=-90,in=90] (F1);

		\end{tikzpicture}
  }	
		}
    \caption{\footnotesize{Short run: the details of the network structure matter even with identical initial prices and technological structures.
    Here the impact is either 5 or 10 times more than the long-run impact (which was 1/100th), and here it depends on the network structure. }}
    \label{fig:ContSR}
  \end{figure}
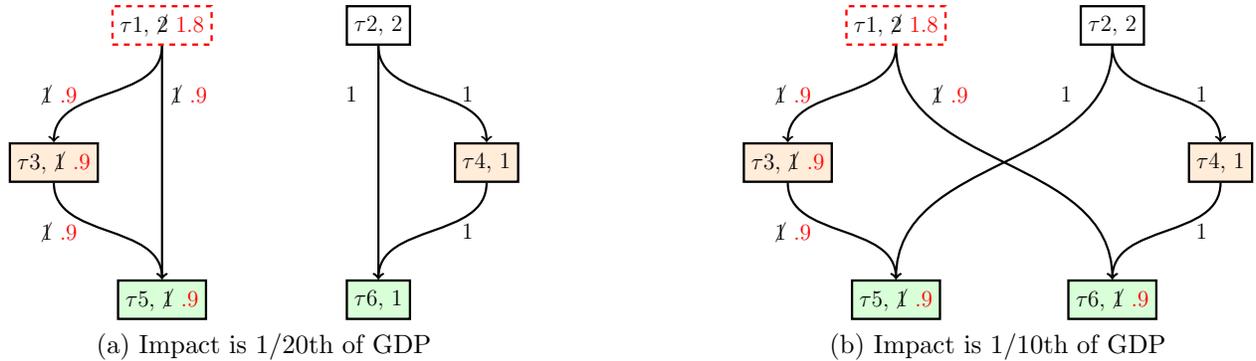

Although our short-run calculation is accurate for non-marginal shocks, unlike Hulten's theorem which only applies at the margin, it is also instructive to compare the short-run marginal impact of a shock to the long-run marginal impact. For example, if the bound from Proposition \ref{prop:shocked_GDP} is obtained, the short-run impact of a disruption to some technology $\tau$ with output good $k$ is
$$
\frac{\partial \log(U)}{\partial \log(\tau_k)} =\frac{\partial  \log(GDP)}{\partial \log(\tau_k)} = \frac{\sum_{f \in  F(\tau)} p_{f} y_{f}}{GDP},
$$ while by contrast the long-run (marginal) impact is, by Hulten's Theorem,
$$
\frac{\partial \log(U)}{\partial \log(\tau_k)} =\frac{\partial  \log(GDP)}{\partial \log(\tau_k)} = \frac{p_{\tau} y_{\tau}}{GDP},
$$
illustrating the large difference that is possible, given the huge potential difference between the value of all affected final goods $\sum_{f \in  F(\tau)} p_{f} y_{f}$ compared to the cost of the affected
input $p_{\tau} y_{\tau}$.

We remark on a couple of implications of these comparisons.  In the short run, the impact of a disruption is not dependent upon how expensive an input is, but instead by the value of all final goods
that lie downstream.  In the long run, it is the reverse.    That does not mean that the long run impact is independent of how upstream or downstream a good is.  Goods that are nearer to final goods and incorporate more inputs from upstream will be more expensive, all else held equal, and so the long run disruption of them is more costly.    So, how upstream or downstream a good is makes a difference in the short run since that might affect how many final goods it reaches, while how upstream or downstream a good is makes a difference in the long run since that affects how expensive it is.  Roughly, shocks to goods that are more upstream are more disruptive in the short run since they affect more final goods, while shocks to goods that are more downstream tend to be more disruptive in the long run since they are more costly.\footnote{If a good is unambiguously more downstream of another in that it uses all of the production of some good upstream (directly or indirectly), then the expenditures on the downstream good are strictly greater.  If one can order goods in levels, so that goods at some level use all the production of goods further upstream, then this also holds across the levels of production.  Thus, there is a well-defined way in which one can think of downstream expenditures being greater than upstream.   Of course, this is having the number of goods in different layers be comparable.  Having one raw material and many intermediate goods, could lead to any one intermediate good having lower expenditure than the raw material. }   The details are given by the formulas, but these rough intuitions are still useful to note.

Although we have so far discussed the implications of negative shocks, the Shock Propagation Algorithm can also be applied to consider positive shocks. Consider a positive shock to technology  $\tau$ that increases output to a proportion $\lambda>1$ of its initial equilibrium level. If $O(\tau)$ is a final good, the impact on final good production is an immediately apparent proportional increase in production of technology $\tau$. On the other hand, if $O(\tau)$ is an intermediate good and all producers that source from $\tau$ use intermediate goods other than just $O(\tau)$, then there is no impact on final good production. This again contrasts with the long-run impact of such changes as captured by Hulten's Theorem, where the (marginal) impact of positive and negative shocks is symmetric.

Finally, consider a producer $i$ downstream of a shocked firm and suppose, for the sake of illustration, that only its supply of one input from one direct supplier is affected. In the short run, $i$'s unaffected suppliers of that same input cannot increase their production or re-route additional goods towards it. Thus, even if $i$ is multi-sourcing the disrupted input, the flows of these alternative supplies remain fixed. For $i$'s suppliers of different inputs, the amount used by $i$ diminishes. Thus, in the short run, inputs are predicted to appear, in data, as complements. However, in the long run, prices adjust, flows reroute, productions change, and affected firms can adjust the technology they are using. All these adjustments reduce the observed complementarity between inputs. This is consistent with previous findings of short-run complementarity and longer-run substitutability between inputs. For example, on a quarter-long horizon \cite{BarrotSuppliers2016} present evidence for inputs being complements, while on a year-long horizon \cite{carvalho2021supply} find inputs are weak substitutes.

\subsection{Flexible Prices: Shock Impacts in the Medium Run}\label{sec:medium}

While in the short run existing contracts prevent prices from adjusting (as embedded within our proportional rationing assumption), once prices can adjust existing production can reroute to higher value uses. This decreases the impact of the shock, yielding an outcome between
the short-run and long-run impacts of the shock.\footnote{E.g., see \cite{iyoha2024exports}.} We now consider this medium-run situation, in which prices can adjust and 
shortages can be rationed efficiently, but in which production technologies are fixed, as is labor, which prevents any firms from increasing their output beyond pre-shock levels.

The medium run outcome can end up anywhere between the short- and long-run outcomes, depending on circumstances. Here we discuss some examples  illustrating the forces, as a full characterization is simply an equilibrium statement that does not provide much further insight. We discuss these ideas more extensively, aided by numerical examples, in Appendix \ref{sec:mediumrun}.

For a real world example where such a reallocation would have been valuable, consider the semiconductor / computer chip crisis. In the short run the production of very valuable downstream goods like cars was disrupted, along with many other consumer goods. With flexible prices, in the medium run, the disruptions should become concentrated on less valuable consumer goods (like, for example, cheap toys). The extent to and speed with which this happened is an empirical question of interest.

Another illustrative example along these lines is the rolling electricity blackouts in California in the summer of 2020. This was applied without discrimination (save some emergency services and legal restrictions) and so fits our proportional rationing assumption. However, in anticipation of further future blackouts, more complicated contracts have emerged that create priorities in rationing and corresponding differences in prices. These ensure that electricity is allocated more efficiently towards higher value uses when there are future shortages. Again, it is an empirical question regarding the extent to which efficient contracts emerge, and which are more challenging to implement in other industries with supplier-specific inputs.

The output loss in the medium run (with flexible prices) is weakly less than the output loss in the short run, because we are relaxing the proportional rationing constraint (and effectively, reoptimizing). However, there are situations in which there is no difference between the medium and short run. For example, if the conditions of Proposition \ref{prop:general_industry_shocks} hold, such that there is no technological diversity and there are industry specific shocks and each shocked technology is used---directly or indirectly---in the production of only one final good, then there is no scope for reducing the impact of the shock in the medium run. There is no scope for correlating the downstream instances of the disruption nor any scope for assigning the disruptions disproportionately to supply chains of lower-value final goods. On the other hand, it is possible to construct examples in which price flexibility makes a big difference. Specifically, the ratio of lost output in the short run to lost output in the medium run with price flexibility is unbounded (as we show in Supplementary Appendix \ref{sec:Potential_price_flexibility}).


\section{Disruption Centrality and Power}

We now return to the short run, and provide further results that leverage the characterization to understand the impact of specific disruptions and how they sit in the supply chain.  This allows us to not only define the ``centrality'' of a given technology based on short-run disruptions,\footnote{The idea of showing that a measure of network centrality determines the impact of shocking different industries dates to \cite{leontief1936quantitative,acemoglu2012network}. Recent work advancing this agenda includes \cite{bahal2023beyond}.} but also to examine how much power a given country might have over others.   

The two approaches in this section take different perspectives.  One is of inadvertent disruption, where disruptions travel as in the short run proportionally along all relevant flows.  The other is of deliberate and targeted disruptions for strategic purposes.

\subsection{Disruption Centrality}

Beyond cases in which the intuitive bound is tight, we can calculate the exact impact of a disruption in a much wider set of cases, even when the bound is not tight.
In this subsection, we use our model to estimate the impact of the disruption of any given technology in the short run where there is proportional rationing on affected flows.  We call the impact ``disruption centrality,'' which we now define.

Let $S(f,\tau')$ denote the percentage of final good $f$ that is produced by technology $\tau'$, and let $S(\tau',\tau'')$ denote the fraction of $O(\tau'')$ that $\tau'$ sources as an input from $\tau''$.
For most pairs this will be 0, but for a technology that sources from another it will be positive.

To keep the definition relatively intuitive, we examine acyclic supply chains.  The definition easily extends, similarly to how we dealt with cycles above, but is no longer defined by a finite algorithm.
Let us say that a set of technologies $T'$ is {\sl well-ordered} if there is an order $\succ$ over the technologies
such that, for any $\tau, \tau' \in T'$, if $\tau \succ \tau'$ then output good $O(\tau')$ is not used as an input by $\tau$. Thus technologies higher in the $\succ$ order are further upstream in a clearly defined sense. This rules out (directed) cycles, but still allows upstream goods to be used in the production of multiple downstream goods and downstream goods to source any of their inputs from multiple upstream sources, and so for rich structures (and for the existence of undirected cycles in the network). In general when a set of technologies is well-ordered, then there will be many such orders, as some goods are never used up or down stream from each other.\footnote{For example,
if the output of technology $\tau$ is good 1, and it is used by technologies $\tau'$ and $\tau''$ to produce goods 2 and 3 respectively, which are then used by technology $\tau'''$ to produce the final good,  then the orderings $\tau \succ \tau'\succ \tau''\succ \tau'''$ and $\tau\succ \tau''\succ \tau'\succ \tau'''$ are both valid.}

When technologies are well-ordered in some equilibrium, then we say that
$\tau$ is downstream from $\tau'$ (and $\tau'$ is upstream from $\tau'$) relative to $\succ$ if there is a sequence of technologies $\tau^1\succ \tau^2 \succ \cdots \succ \tau^L$ with $L\geq 2$, $\tau^1=\tau'$ and  $\tau^L=\tau$.   

Consider a technology $\tau\in T$ that has equilibrium paths to some set of final goods $F(\tau)$, and consider some $f\in F(\tau)$. Let $T_{\tau,f}$ be the set of all technologies that lie on any equilibrium path between $\tau$ and some $\tau^f$ producing $f$ (inclusive).
Suppose that $T_{\tau,f}$ is well-ordered with corresponding order $\succ_{\tau,f}$.
Let $G_{\tau,f}$ be the set of goods that are produced by any $\tau'\in T_{\tau,f}$, and number them 1 to $K_f$, ordered according to $\succ_{\tau,f}$ with good 1 being $O(\tau)$ and good $K_f$ being f.

Let $d_{\tau,f}(\tau) = 1$ and let $d_{\tau,f}(\tau') = 0$ for all $\tau'\not\in G_{\tau,f}$.
Inductively,  for $k\in \{1,\ldots, K_f\}$,
consider $\tau'\in T_{\tau,f}$ producing good $k$.
Let
$d_{\tau,f}(\tau') =  \max_{i \in I(\tau')} \left[ \sum_{\tau'':  O(\tau'')=i}  d_{\tau,f}(\tau'' ) S(\tau',\tau'') \right].$

Then the disruption centrality of $\tau$ is
$$
DC(\tau)\equiv \frac{\sum_{f\in F(\tau)}p_f y_f \left(\sum_{\tau': O(\tau') = f}  d_{\tau,f}(\tau') S(f,\tau')\right)}{GDP}.
$$ Figure \ref{fig:disruption_centrality} in Supplementary Appendix \ref{sec:extended_eg} provides an example of how the $d_{\tau,f}(\tau')$ terms are calculated.


If all technologies are single-sourced (i.e., $S(\cdot,\cdot)$ is always 0 or 1)
then 
$DC(\tau) =\sum_{f\in F(\tau)} p_f y_f,$ 
which corresponds to the upper bound we identified before.
More generally, however, $DC(\tau)$ captures all the downstream disruptions accounting for all the fractions and multiple paths
that lie between some technology and all the final goods that are produced downstream from it.

\subsection{Power}\label{sec:power}

The perspective of ``Disruption Centrality'' is agnostic on why a particular technology might be disrupted and considers a total factor productivity shock (as is standard). 
There are however situations in which disruption, or the threat of disruption, is strategic. This could come either in the form of sanctions, or other deliberate disruptions as part of bargaining or conflict, including various quotas or other trade limitations. To capture how the structure of the supply network feeds into such situations we consider the power that one country wields over another via purposeful disruptions of its own technologies.

Although we use the terminology of ``countries'', it should be clear that this can be any distinction of regions under consideration (e.g., parts of countries or collections of countries, etc.).

Total factor productivity shocks are not the best way to capture deliberate disruptions by an aggressor country. The aggressor can limit its own losses by selectively choosing
which inputs it sources for some goods (e.g., retaining domestic inputs and eschewing inputs from a target country) or similarly from selectively choosing to whom it sells various goods (continuing to sell domestically, but reducing exports to a target country). 
We therefore study the impact of disruptions in which countries  selectively withhold trade along some paths that involve their imports and exports.

To calculate the power that one country has over another (in the short run), it is  no longer sufficient to simply track final goods disruptions, as we now care about the distribution of the impact of the disruption across countries. 
A country's GDP (and welfare) can be measured via its labor income, and following a disruption some labor becomes idle throughout the supply chain. Indeed if two percent of a country's labor becomes idle, then two percent of its GDP is lost.\footnote{We work with homogeneous labor for simplicity, but more generally one can simply wage-weight labor to adjust for differences in productivity by labor in different tasks.} Additional losses in terms of utility could also come via the final bundle that is available to purchase.\footnote{If there is just one final good, then relative labor disruptions and corresponding reductions in relative purchasing power map exactly into GDP and utility losses. That is, with just one final good, the loss in labor wages (both up and down stream) equals the loss in consumption.}

To measure the power that one country has over another we analyze what happens when a country disrupts the production of some of its technologies. We allow the country instigating the disruption to reduce the output of any of its technologies, or combination of its technologies, and to direct these disruptions selectively. Thus a country can choose which customers receive how much less output from one of its disrupted technologies, and, when it sources the same input from multiple suppliers for a disrupted technology, how much less demand it sources from each supplier. 

We start off by considering the short run in which the other countries don't have time to react strategically and the indirect disruptions propagate up and downstream proportionally across flows. In Section \ref{sec:strategic_power} we turn to time horizons over which the other countries have discretion over how to allocate their disrupted production and how to source the inputs they still need, and consider how this mitigates the power countries have over one another.

\subsubsection{A Heuristic Definition of Power}

We define the power of an aggressor country looking to disrupt production in a target country, as the biggest GDP losses that can be induced in the target country at the expense of the least GDP losses to itself.  This is the disruption with the highest ``bang-for-the-buck,'' the disruption which maximizes the percentage GDP loss in the target country per percent of own GDP loss.

We denote the power that Country $i$ has over Country $j$ by $\text{Power}_{ij}$.
For now, we define it heuristically (as we have not fully specified how to calculate total lost GDP as a function of disruptions yet). 
We offer a full definition further below, but this is useful to help fix ideas.
\begin{eqnarray}
\text{Power}_{ij}&=&\max_{\text{disruption of flows in/out of $i$'s technologies}}\frac{\text{\% lost GDP in Country }j}{\text{\% lost GDP in Country }i}\nonumber\\
&=&\max_{\text{disruption of flows in/out of $i$'s technologies}}\left[\frac{\text{lost GDP in Country }j}{\text{lost GDP in Country }i}\right]\bigg/\left[\frac{\text{GDP in Country }j}{\text{GDP in Country }i}\right]\nonumber
\end{eqnarray}

Percentage losses capture the power that larger countries with more diversified economies have over small ones and reflects the losses per capita in the aggressor and target economies. An alternative measure of power, that is appropriate in some circumstances, is the ratio of lost GDP. However, as can be seen in the equation above, the same disruptions maximize the ratio of absolute losses in GDP also maximize the  ratio of percentage losses in GDP. 

If $\text{Power}_{i j}<1$ then Country $i$ has to be willing to forego a larger (percentage loss) in GDP itself than it imposes on Country $j$. More generally, when $\text{Power}_{i j}$ is sufficiently low, no disruption by $i$ may be credible and hence $i$ may not wield any viable power over $j$. 

\subsubsection{An Illustrative Example}

Consider the initial equilibrium shown in Panel (a) of Figure \ref{fig:power2a}. There are two countries, Country $i$ and Country $j$. Country $i$ is on the right in Figure \ref{fig:power2a} and operates technologies $T_i=\{\tau2,\tau5\}$, while Country $j$ operates technologies $T_j=\{\tau1,\tau3,\tau4\}$. 
We normalize the price of the final good to be 1, and so the 
the equilibrium wage is $1/3$ (a total of 120 units of labor consume 40 units of output). GDP in Country $i$ is $70/3$ while it is $50/3$ in Country $j$, which summed give the total consumption. 

If Country $i$ wants to reduce the GDP of Country $j$, it can only do so by limiting the amount of the output of technology $\tau2$ that it exports to Country $j$. Reducing the flow $\tau2,\tau3$ from $10$ to $9$, as shown in Panel (b) of Figure \ref{fig:power2a}, reduces the amount that $\tau3$ can produce, which in turn reduces the amount that $\tau4$ can produce and the amount demanded from $\tau1$. This is shown in Panel (c) of Figure \ref{fig:power2a}. Thus, $1$ unit of labor becomes idle in Country $j$ whereas $5$ units of labor become idle in Country $i$. Hence GDP in Country $i$ falls by $1/70$th while GDP in Country $j$ falls by $1/10$th.

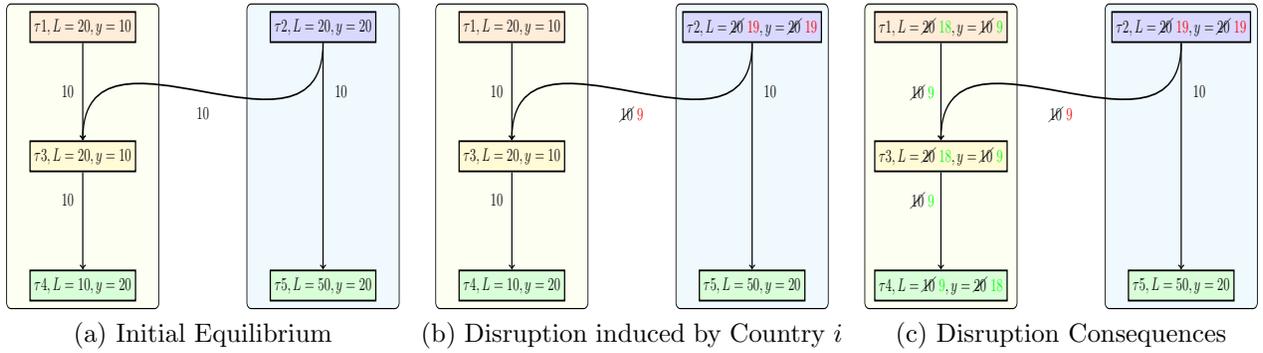
\begin{figure}[!ht]
    \subfloat[Initial Equilibrium]
    {
\centering
\resizebox{2.1in}{1.62in}{		
    \begin{tikzpicture}[
sqRED/.style={rectangle, draw=black!240, fill=red!15, very thick, minimum size=7mm},
		sqREDb/.style={rectangle, draw=black!240, fill=red!15, very thick, dashed, minimum size=7mm},
		sqBLUE/.style={rectangle, draw=black!240, fill=blue!15, very thick, minimum size=7mm},
		sqBLUEb/.style={rectangle, draw=black!240, fill=blue!15, very thick, dashed, minimum size=7mm},
		roundnode/.style={ellipse, draw=black!240, fill=green!15, very thick, dashed, minimum size=7mm},
		sqGREEN/.style={rectangle, draw=black!240, fill=green!15, very thick, minimum size=7mm},
		sqGREENb/.style={rectangle, draw=black!240, fill=green!15, very thick, dashed, minimum size=7mm},
		roundRED/.style={ellipse, draw=black!240, fill=red!15, very thick, dashed, minimum size=7mm},
		roundREDb/.style={ellipse, draw=black!240, fill=red!15, very thick, minimum size=7mm},
		roundYELL/.style={ellipse, draw=black!240, fill=yellow!20, very thick, dashed, minimum size=7mm},
		sqYELL/.style={rectangle, draw=black!240, fill=yellow!20, very thick, minimum size=7mm},
		roundORA/.style={ellipse, draw=black!240, fill=orange!15, very thick, dashed, minimum size=7mm},		
		sqORA/.style={rectangle, draw=black!240, fill=orange!15, very thick, minimum size=7mm},
		sqORAb/.style={rectangle, draw=black!240, fill=orange!15, very thick, dashed, minimum size=7mm},
		sqMAG/.style={rectangle, draw=black!240, fill=magenta!5, very thick, minimum size=7mm},
		squaredBLACK/.style={rectangle, draw=black!240, fill=white!15, very thick, minimum size=7mm},	
		roundBROWN/.style={ellipse, draw=black!240, fill=brown!15, very thick, dashed, minimum size=7mm},
		sqBROWN/.style={rectangle, draw=black!240, fill=brown!15, very thick, minimum size=7mm},		
        ]

\node[](HG1A) at (-1.75, 3.25) {};
\node[](HG1B) at (-6.25, -3.25) {};
\node [he, draw=black!240, fill=lime!5, fit=(HG1A) (HG1B)] {};

\node[](HG2A) at (1.75, 3.25) {};
\node[](HG2B) at (6.25, -3.25) {};
\node [he, draw=black!240, fill=cyan!5, fit=(HG2A) (HG2B)] {};

		\node[sqORA](A1) at (-4, 3) {$\tau1, L=20,  y=10$};
        \node[] at (-4.5, 1.5) { $10$};
		
        \node[sqBLUE](A2) at (4, 3) {$\tau2,  L=20,  y=20$};
        \node[] at (0, 1) { $10$};
		\node[] at (4.6, 1.5) { $10$};
			
		\node[sqYELL](C1) at (-4, 0) {$\tau3, L=20, y =10 $};
       \node[] at (-4.5, -1) { $10$};

		
		\node[sqGREEN](F1) at (-4, -3) {$\tau4,L=10, y=20$};
	
		\node[sqGREEN](F2) at (4, -3) {$\tau5, L = 50, y=20$};	
		
		\tikzset{thick edge/.style={-, black, fill=none, thick, text=black}}
		\tikzset{thick arc/.style={->, black, fill=black, thick, >=stealth, text=black}}
		
		\draw[line width=0.4mm, black,-> ] (A1) to[out=-90,in=90] (C1);
		\draw[line width=0.4mm, black,-> ] (A2) to[out=-90,in=90] (C1);
		
		\draw[line width=0.4mm, black,-> ] (C1) to[out=-90,in=90] (F1);
		\draw[line width=0.4mm, black,-> ] (A2) to[out=-90,in=90] (F2);

		
       
		\end{tikzpicture}
		}
       }
    \hfill
    \subfloat[Disruption induced by Country $i$ ]
    {
\centering
\resizebox{2.1in}{1.62in}{		
    \begin{tikzpicture}[
sqRED/.style={rectangle, draw=black!240, fill=red!15, very thick, minimum size=7mm},
		sqREDb/.style={rectangle, draw=black!240, fill=red!15, very thick, dashed, minimum size=7mm},
		sqBLUE/.style={rectangle, draw=black!240, fill=blue!15, very thick, minimum size=7mm},
		sqBLUEb/.style={rectangle, draw=black!240, fill=blue!15, very thick, dashed, minimum size=7mm},
		roundnode/.style={ellipse, draw=black!240, fill=green!15, very thick, dashed, minimum size=7mm},
		sqGREEN/.style={rectangle, draw=black!240, fill=green!15, very thick, minimum size=7mm},
		sqGREENb/.style={rectangle, draw=black!240, fill=green!15, very thick, dashed, minimum size=7mm},
		roundRED/.style={ellipse, draw=black!240, fill=red!15, very thick, dashed, minimum size=7mm},
		roundREDb/.style={ellipse, draw=black!240, fill=red!15, very thick, minimum size=7mm},
		roundYELL/.style={ellipse, draw=black!240, fill=yellow!20, very thick, dashed, minimum size=7mm},
		sqYELL/.style={rectangle, draw=black!240, fill=yellow!20, very thick, minimum size=7mm},
		roundORA/.style={ellipse, draw=black!240, fill=orange!15, very thick, dashed, minimum size=7mm},		
		sqORA/.style={rectangle, draw=black!240, fill=orange!15, very thick, minimum size=7mm},
		sqORAb/.style={rectangle, draw=black!240, fill=orange!15, very thick, dashed, minimum size=7mm},
		sqMAG/.style={rectangle, draw=black!240, fill=magenta!5, very thick, minimum size=7mm},
		squaredBLACK/.style={rectangle, draw=black!240, fill=white!15, very thick, minimum size=7mm},	
		roundBROWN/.style={ellipse, draw=black!240, fill=brown!15, very thick, dashed, minimum size=7mm},
		sqBROWN/.style={rectangle, draw=black!240, fill=brown!15, very thick, minimum size=7mm},		
        ]

\node[](HG1A) at (-1.75, 3.25) {};
\node[](HG1B) at (-6.25, -3.25) {};
\node [he, draw=black!240, fill=lime!5, fit=(HG1A) (HG1B)] {};

\node[](HG2A) at (1.75, 3.25) {};
\node[](HG2B) at (6.25, -3.25) {};
\node [he, draw=black!240, fill=cyan!5, fit=(HG2A) (HG2B)] {};

		\node[sqORA](A1) at (-4, 3) {$\tau1, L=20,  y=10$};
        \node[] at (-4.5, 1.5) { $10$};
		
        \node[sqBLUE](A2) at (4, 3) {$\tau2,  L=\cancel{20} \ \textcolor{red}{19},  y=\cancel{20} \ \textcolor{red}{19}$};
        \node[] at (0, 1) { $\cancel{10} \ \textcolor{red}{9}$};
		\node[] at (4.6, 1.5) { $10$};
			
		\node[sqYELL](C1) at (-4, 0) {$\tau3, L=20, y =10 $};
       \node[] at (-4.5, -1) { $10$};

		
		\node[sqGREEN](F1) at (-4, -3) {$\tau4,L=10, y=20$};
	
		\node[sqGREEN](F2) at (4, -3) {$\tau5, L = 50, y=20$};	
		
		\tikzset{thick edge/.style={-, black, fill=none, thick, text=black}}
		\tikzset{thick arc/.style={->, black, fill=black, thick, >=stealth, text=black}}
		
		\draw[line width=0.4mm, black,-> ] (A1) to[out=-90,in=90] (C1);
		\draw[line width=0.4mm, black,-> ] (A2) to[out=-90,in=90] (C1);
		
		\draw[line width=0.4mm, black,-> ] (C1) to[out=-90,in=90] (F1);
		\draw[line width=0.4mm, black,-> ] (A2) to[out=-90,in=90] (F2);

		
       
		\end{tikzpicture}
		}
        }
           \hfill
    \subfloat[Disruption Consequences]
    {
\centering
\resizebox{2.1in}{1.62in}{		
    \begin{tikzpicture}[
sqRED/.style={rectangle, draw=black!240, fill=red!15, very thick, minimum size=7mm},
		sqREDb/.style={rectangle, draw=black!240, fill=red!15, very thick, dashed, minimum size=7mm},
		sqBLUE/.style={rectangle, draw=black!240, fill=blue!15, very thick, minimum size=7mm},
		sqBLUEb/.style={rectangle, draw=black!240, fill=blue!15, very thick, dashed, minimum size=7mm},
		roundnode/.style={ellipse, draw=black!240, fill=green!15, very thick, dashed, minimum size=7mm},
		sqGREEN/.style={rectangle, draw=black!240, fill=green!15, very thick, minimum size=7mm},
		sqGREENb/.style={rectangle, draw=black!240, fill=green!15, very thick, dashed, minimum size=7mm},
		roundRED/.style={ellipse, draw=black!240, fill=red!15, very thick, dashed, minimum size=7mm},
		roundREDb/.style={ellipse, draw=black!240, fill=red!15, very thick, minimum size=7mm},
		roundYELL/.style={ellipse, draw=black!240, fill=yellow!20, very thick, dashed, minimum size=7mm},
		sqYELL/.style={rectangle, draw=black!240, fill=yellow!20, very thick, minimum size=7mm},
		roundORA/.style={ellipse, draw=black!240, fill=orange!15, very thick, dashed, minimum size=7mm},		
		sqORA/.style={rectangle, draw=black!240, fill=orange!15, very thick, minimum size=7mm},
		sqORAb/.style={rectangle, draw=black!240, fill=orange!15, very thick, dashed, minimum size=7mm},
		sqMAG/.style={rectangle, draw=black!240, fill=magenta!5, very thick, minimum size=7mm},
		squaredBLACK/.style={rectangle, draw=black!240, fill=white!15, very thick, minimum size=7mm},	
		roundBROWN/.style={ellipse, draw=black!240, fill=brown!15, very thick, dashed, minimum size=7mm},
		sqBROWN/.style={rectangle, draw=black!240, fill=brown!15, very thick, minimum size=7mm},		
        ]

\node[](HG1A) at (-1.75, 3.25) {};
\node[](HG1B) at (-6.25, -3.25) {};
\node [he, draw=black!240, fill=lime!5, fit=(HG1A) (HG1B)] {};

\node[](HG2A) at (1.75, 3.25) {};
\node[](HG2B) at (6.25, -3.25) {};
\node [he, draw=black!240, fill=cyan!5, fit=(HG2A) (HG2B)] {};

		\node[sqORA](A1) at (-4, 3) {$\tau1, L=\cancel{20} \ \textcolor{green}{18},  y=\cancel{10} \ \textcolor{green}{9}$};
        \node[] at (-4.6, 1.5) { $\cancel{10} \ \textcolor{green}{9}$};
		
        \node[sqBLUE](A2) at (4, 3) {$\tau2,  L=\cancel{20} \ \textcolor{red}{19},  y=\cancel{20} \ \textcolor{red}{19}$};
        \node[] at (0, 1) { $\cancel{10} \ \textcolor{red}{9}$};
		\node[] at (4.6, 1.5) { $10$};
			
		\node[sqYELL](C1) at (-4, 0) {$\tau3, L=\cancel{20} \ \textcolor{green}{18}, y =\cancel{10} \ \textcolor{green}{9} $};
       \node[] at (-4.6, -1) { $\cancel{10} \ \textcolor{green}{9}$};

		
		\node[sqGREEN](F1) at (-4, -3) {$\tau4,L=\cancel{10} \ \textcolor{green}{9}, y=\cancel{20} \ \textcolor{green}{18}$};
	
		\node[sqGREEN](F2) at (4, -3) {$\tau5, L = 50, y=20$};	
		
		\tikzset{thick edge/.style={-, black, fill=none, thick, text=black}}
		\tikzset{thick arc/.style={->, black, fill=black, thick, >=stealth, text=black}}
		
		\draw[line width=0.4mm, black,-> ] (A1) to[out=-90,in=90] (C1);
		\draw[line width=0.4mm, black,-> ] (A2) to[out=-90,in=90] (C1);
		
		\draw[line width=0.4mm, black,-> ] (C1) to[out=-90,in=90] (F1);
		\draw[line width=0.4mm, black,-> ] (A2) to[out=-90,in=90] (F2);

		
       
		\end{tikzpicture}
		}
        }
    \caption{\footnotesize{A targeted withholding of the output of technology $\tau2 $ by Country $i$ (on the right) leads to one unit of lost labor to Country $i$ and imposes a loss of 5 units of labor on Country $j$ (on the left). Country $i$ has a power of 5 over Country $j$ if we measure it in absolute terms, and $7=(5/50)/(1/70)$ if we measure it in percentage of GDP terms.}
    \label{fig:power2a}}   
  \end{figure}

This example illustrates several points. 

First, it is important to consider upstream disruptions, here to technology $\tau1$, in order to account for how the idle labor is distributed across countries.   This contrasts with our earlier analysis in which we were measuring overall (world) GDP, which can be captured via final goods reduction.  Here, tracking changes in labor at each technology captures how the overall loss of 2/40 units of final goods are distributed.  Country $i$ loses 1/120 of the total labor in the economy while Country $j$ loses 5/120 of the total labor.  This leads to the ratio of 5,  and correspondingly Country $i$ loses 1/3 of a unit of consumption, while Country $j$ loses 5/3 units of consumption in total.   

Second, Country $i$ exerts the biggest loss possible on Country $j$ per unit of own loss by directing its disruption of $\tau2$ towards technology $\tau3$ and keeping the amount shipped to $\tau5$ intact.  There is no other disruption available to Country $i$ that would do as well as the one we have considered in terms of power.  It could be scaled up, but would have the same ratio of losses, until it hits 10 units being withheld, and then further withholding would only hurt Country $i$.   As such the relative power index, $\text{Power}_{ij}=7$, and Country $i$ has a relatively high amount of leverage over Country $j$. 

Third, as we mentioned above, an alternative measure of power that it can be appropriate to use in some circumstances is the ratio of absolute GDP losses, and in this case the power of Country $i$ over Country $j$ would be $5$. 

Fourth, in this example, there is only one flow between the two countries and so if we want to examine the power of Country $j$ over Country $i$,  it could only disrupt things via the same flow.   In this particular example, the power of $j$ over $i$ is just the reciprocal 1/7 (or 1/5 in absolute terms). In richer examples, with multiple flows across borders (see below), the best disruption for one country will generally not be the same as the best disruption for the other.  

Fifth, in this example $\tau4$ and $\tau5$ produce the same final good. 
If instead they produced different goods, then the consumers will also be forced to purchase a worse combination of goods after the disruption (as production by $\tau4$ is reduced but production by $\tau5$ is not).\footnote{ Proportional rationing of the over-demanded final good will result in Country $i$'s consumers consuming $18(69/114)$ and $20(69/114)$ units of the two final goods, while Country $j$'s consumers will consume the remaining amounts. Thus the consumers in both countries will consume different amounts of the same bundle of final goods. As by assumptions preferences are homothetic, this is also what would also happen with flexible final good prices.} Finally, Country $i$ is relatively powerful in this example. The only way Country $j$ can disrupt production in Country $i$ is to reduce its demand for the output of technology $\tau2$. However, the impact of doing so is the same as shown in Panel (c) of Figure \ref{fig:power2} and so it only has power $1/5$ over Country $i$. So, to hurt Country $i$, Country $j$ has to hurt itself (considerably) more. The substantial power that Country $i$ has over Country $j$ is in part derived from its ability to insulate its own final good production from the shock it induces. Its disruption does not induce a further disrupted flow of goods back across borders affecting its own technologies.  

\subsection{A Formal Definition of Power}\label{sec:marginalpower}

We now provide a formal definition of the power Country $i$ has to disrupt Country $j$. To do this we must first specify how the full economy responds to a set of flow reductions by Country $i$. 

This involves three complications compared to the downstream propagation algorithm that we defined previously:
\begin{itemize}
\item  To account for all labor reductions, we also need to track how every directly or indirectly affected technology's disruption propagates upstream.
\item  A disruption can involve multiple levels of discretion by Country $i$ that have to be consistent with each other. 
\item  Even without cycles in the network, there can be (finite) cyclical feedback as shocks propagate both upstream and downstream.\footnote{For instance, if a good 0 used in the production of good 1 is held back, and reduces production of  good 2, but good 2 also uses the good 0 as an input, then reduced demand from good 2 further reduces the production of good 0.}
\end{itemize}

We attack these complications by defining a notion of a set of ``consistent disruptions,'' and showing it is without loss of generality to restrict attention to these. 
While this simplifies matters, the set of available sets of consistent disruptions can be large and complex. We show that under some basic conditions on the supply network, when looking for the optimal disruption that determines the power one country has over another, attention can be restricted without loss of generality to sets of minimal consistent disruptions induced by a disruption to a {\sl single} technology.

We consider supply networks with no (directed) cycles.   
Thus, as above there exists a well defined order $\succ$ and sequence of all technologies $\tau^1, \ldots, \tau^L$ 
such that $\tau^\ell \succ \tau^{\ell'} $ implies $\tau^\ell  $ is not downstream from $ \tau^{\ell'}$. We label the technologies from upstream to downstream so that $\tau^\ell \succ \tau^{\ell'} $ when $\ell'>\ell$.

Consider some equilibrium and let $x^*_{\tau\tau'}$ denote the equilibrium flow of output from technology $\tau$ to technology $\tau'$ and $y^*_\tau, L^*_\tau$ be the equilibrium output and labor usage associated with $\tau$. 

A {\sl consistent disruption} of an equilibrium $(x^*_{\tau\tau'}, y^*_\tau, L^*_\tau)_{\tau,\tau'}$ 
by Country $i$ via technology $\tau^i\in T_i$ 
that is reduced to $\lambda<1$ of its initial level\footnote{Note that any disruption of $\tau^i $ along any flows up or downstream reduces its output, and so this nests both cases.} is a specification of inputs, outputs, and labor for each technology, 
$x_{\tau\tau'}, y_\tau, L_\tau$ for every $\tau,\tau'$,
as follows.  
Let $\tau^i=\tau^\ell$ in our ordered sequence for some $\ell\in \{1, \ldots, L\}$.  
\begin{itemize}
\item First, specify new downstream flows $x^d_{\tau^\ell\tau'}$ such that
$x^d_{\tau^\ell\tau'}\leq x^*_{\tau^\ell\tau'}$ for all $\tau',$ and 
$\sum_{\tau'} x^d_{\tau^\ell\tau'} = \lambda y^*_{\tau^\ell}$. 
This specifies how $i$ allocates the disrupted production of technology $\tau^{\ell}$.
  
\item We then follow this iteratively for downstream technologies $\ell' $ (such that $\ell\succ \ell'$) starting with technology $\ell +1$. It may be that $\ell'$ has had multiple cuts on its inputs from technologies further upstream when it is considered, and so the amount that it needs to ration is the max shortfall across its inputs.   
That is,  
the current (temporary) $y^d_{\tau^{\ell'}} = \min_{j\in I(\tau^{\ell'})} \frac{\sum_{\tau':O(\tau')=j} x^d_{\tau'\tau^{\ell'}}}{\sum_{\tau':O(\tau')=j} x^{*}_{\tau'\tau^{\ell'}}}$
\begin{itemize}
\item and if $\tau^{\ell'}\in T_i$, then select any set $x^d_{\tau^{\ell'}\tau'}\leq x^*_{\tau^{\ell'}\tau'}$ for all $\tau',$ and such that 
$\sum_{\tau'} x^d_{\tau^{\ell'}\tau'} =  y^d_{\tau^{\ell'}}$. So $i$ controls the disrupted flows through its own technologies.
\item and if $\tau^{\ell'}\in T_j$ for $j\ne i$, then set  $x^d_{\tau^{\ell'}\tau'}\leq x^*_{\tau^{\ell'}\tau'}$ for all $\tau',$  such that 
$\frac{ x^d_{\tau^{\ell'}\tau'}}{x^*_{\tau^{\ell'}\tau'}} =  \frac{y^d_{\tau^{\ell'}}}{y^*_{\tau^{\ell'}}}$. Thus there is proportional rationing of other countries' outputs.
\end{itemize}
\item We do this until we hit $\ell'=L$, which is necessarily a final good.  
Note that this is well-ordered due to the lack of cycles.
Set $y^u_{\tau^{L}}=y^d_{\tau^{L}}$, to start the upstream process.
\item Next, we follow the consequences back upstream starting at technology $\tau^L$ and iterating until we get to technology $\tau^1$.\footnote{Note that upstream propagation has no further downstream consequences.}  The iterative step is as follows. 
\item We ration $\tau^{\ell'}$'s inputs to levels $x^u_{\tau^{\ell'}\tau'}$ to match the production  $y^u_{\tau^{\ell'}}$  
from the previous steps, as follows:
\begin{itemize}
\item if $\tau^{\ell'}\in T_i$, then select upstream choices  $x^u_{\tau'\tau^{\ell'}}$ such that $x^u_{\tau'\tau^{\ell'}}\leq x^d_{\tau'\tau^{\ell'}}$ for all $\tau'$ and
$\frac{\sum_{\tau':O(\tau')=j} x^u_{\tau'\tau^{\ell'}}}{\sum_{\tau':O(\tau')=j} x^*_{\tau'\tau^{\ell'}}} =  \frac{y^u_{\tau^{\ell'}}}{y^*_{\tau^{\ell'}}}$ for all $j\in I(\tau^{\ell'})$ 
\item and if $\tau^{\ell'}\in T_j$, then set  $x^u_{\tau'\tau^{\ell'}}$ such that $x^u_{\tau'\tau^{\ell'}}\leq x^d_{\tau'\tau^{\ell'}}$ for all $\tau'$ and if $x^d_{\tau'\tau^{\ell'}}>0$ then 
$\frac{ x^u_{\tau'\tau^{\ell'}}}{x^d_{\tau'\tau^{\ell'}}} $ is equal for all such $\tau'$ and 
$\frac{\sum_{\tau':O(\tau')=j} x^u_{\tau'\tau^{\ell'}}}{\sum_{\tau':O(\tau')=j} x^*_{\tau'\tau^{\ell'}}} =  \frac{y^u_{\tau^{\ell'}}}{y^*_{\tau^{\ell'}}}$ for all $j\in I(\tau^{\ell'})$.
\end{itemize}
\item By construction, these upstream disruptions eliminate flows of inputs goods that are unused in production, but cutting an upstream flow has no further downstream consequences (given the absence of directed cycles).   
\end{itemize}
We are left with a well defined (individual technology) 
disruption. 
We can now treat this $(x_{\tau\tau'}, y_\tau, L_\tau)$ specification as if it were a new ``equilibrium'' and then apply a new disruption to some other $\tau^{i'} \in T_i$. 

We let $\mathcal{D}_i$ denote the set of consistent disruptions available  to Country $i$ via sequences of such disruptions, i.e., those disruptions that can be obtained from any finite sequence of individual technology disruptions as described above.\footnote{In applying sequences of disruptions, the marginal proportional rationing of Country $j$ is done using the relative size of the original equilibrium flows.
We show in the appendix that the order of disruptions is not consequential as long as no flows have been cut to zero.}

For any 
$(x_{\tau\tau'}, y_\tau, L_\tau)_{\tau,\tau'}\in \mathcal{D}_i$ denote the corresponding GDP of Country $i$ (and similarly for $j$) by
$$GDP_i((x_{\tau\tau'}, y_\tau, L_\tau)_{\tau,\tau'})= w^*_i\sum_{\tau\in T_i} L_\tau ,$$ 
where $w^*_i$ is the original wage level (as by assumption this is not adjusted in the short run). Letting $GDP^*_i$ denote the initial equilibrium GDP of Country $i$,
\begin{eqnarray}
\text{Power}_{ij}&=&\max_{(x_{\tau\tau'}, y_\tau, L_\tau)_{\tau,\tau'}\in \mathcal{D}_i}\left[\frac{GDP^*_j-GDP_j((x_{\tau\tau'}, y_\tau, L_\tau)_{\tau,\tau'})}{GDP^*_i-GDP_i((x_{\tau\tau'}, y_\tau, L_\tau)_{\tau,\tau'})}\right]\bigg/\left[\frac{GDP^*_j}{GDP^*_i}\right]\nonumber\\
&=&\max_{(x_{\tau\tau'}, y_\tau, L_\tau)_{\tau,\tau'}\in \mathcal{D}_i}\left(\frac{\sum_{\tau \in T_j} L^*_{\tau}-L_{\tau}}{\sum_{\tau \in T_j}L^*_{\tau}}\right)\bigg/\left(\frac{\sum_{\tau\in T_i}L^*_{\tau}-L_{\tau}}{\sum_{\tau\in T_i}L^*_{\tau}}\right)\nonumber
\end{eqnarray}

The space of consistent disruptions $\mathcal{D}_i$ can be large. 
We first consider the case in which disruptions are such that no technology is completely shut down.  Once the new supply chains hit corners where some technology flows are completely disrupted, power calculations become more involved as we detail below.  In most applications, countries are not completely shutting down cross border flows, and then things are much easier to analyze. 
In particular, if none of $i$'s technological flows are completely disrupted then,
as we show, maximum power is achieved via a consistent individual technology disruption.
We call these partial disruptions.

\begin{proposition}\label{prop:disruption_frontier}
If the initial equilibrium supply network has no (directed) cycles and only partial disruptions are considered, then 
then there exists a technology $\tau^i\in T_i$ such that the the power of Country $i$ over Country $j$ is maximized by a consistent individual technology disruption $\tau^i$. 
Moreover, the disruption can be rescaled arbitrarily (subject to staying partial) while maintaining the same power.
\end{proposition}

Proposition \ref{prop:disruption_frontier} shows that when the supply network is acyclic
there is a disruption available to $i$ that hurts country $j$ the most per amount of disruption in country $i$, and it involves country $i$ initially disrupting a single technology in $T_i$. So, the problem of finding the set of disruptions available to $i$ that maximizes $i$'s power over $j$, as defined above, simplifies to searching over disruptions to a single technology in $T_i$.
This makes such disruptions relatively easy to find and implement.

The proof of Proposition \ref{prop:disruption_frontier} is in Appendix \ref{sec:proofs}. It works by showing that complicated disruptions by $i$ can be broken down into separate independent disruptions, such that each independent disruption is induced by an initial disruption to a single technology in $T_i$. 
The power associated with the complicated disruption is then a weighted average of the power associated with these separate disruptions, and so one of the separate disruptions must be weakly more powerful than the combination.

We remark that an algorithm looking for the technologies to disrupt can actually further restrict attention to the technologies 
that are closest to the border in a well-defined sense.  That is, one can concentrate on disruptions of  $\tau\in T_i$ such that 
either there is at least one downstream path from $\tau$ to a final good technology that involves no technologies in $T_i$ anywhere on the path, or 
 there is at least one upstream path from $\tau$ to an intermediate good technology that uses only labor in its production that involves no technologies in $T_i$ anywhere on the path.

When supply networks become more involved, the importance of
restricting attention to a smaller set of potential technologies to disrupt becomes more combinatorially important.  
We illustrate that with a more involved example in Appendix  \ref{sec:involvedexample}.

That more involved example also shows how one country can have much more power over another than vice versa.  
In doing so it helps identify the characteristics of the supply network that grant Country $i$ power over Country $j$. Country $i$ will have a substantial amount of power over Country $j$ when: (i) Country $j$'s output is relatively highly concentrated in one supply chain (measured by its proportion of wage income associated with that supply chain); (ii) Country $i$ has the ability to disrupt this supply chain; and (iii) Country $i$ has relatively diversified production (measured by its proportion of wage income associated with other supply chains). 

An example of a powerful disruption is one to a key good that is relatively monopolized by the disrupting country and not costly in labor itself, but is present in high-labor-usage technologies (up and/or downstream) that lie significantly within the target country.   

Note that it is possible for Countries $i$ and $j$ to both have substantial power over each other (potentially via completely separate parts of the supply network), and for the power relationship to be very asymmetric---so that one country has substantial power over the other but not vice versa.

\subsection{The Disruption Possibility Frontier}\label{sec:nonmarginal}

If countries consider large disruptions that completely cut ties across borders, then that can alter the marginal impact of other technologies as it changes the ways in which subsequent flows are rationed. 
To capture such extreme interventions, one can no longer limit attention 
to just individual technology disruptions and combinations and orderings of disruptions become important.  
We can still define power, and it has a single well-defined value, but the algorithm for finding it when allowing for complete disruptions is more involved.

We ask, what is the maximum percentage reduction in GDP in Country $j$ that Country $i$ can impose for an $x$\% loss in own GDP, varying from 0 to 100 percent own loss. 
This is defined by solving the following problem for all $\lambda\in[0,1]$. 
$\min_{\{(x_{\tau\tau'}, y_\tau, L_\tau)_{\tau,\tau'}\in \mathcal{D}_i\}} \frac{GDP_j\big((x_{\tau\tau'}, y_\tau, L_\tau)_{\tau,\tau'}\big)}{GDP^*_j}$ subject to $GDP_i\big((x_{\tau\tau'}, y_\tau, L_\tau)_{\tau,\tau'}\big)\geq (1-\lambda)GDP^*_i$.

Solving this problem for different values of $\lambda$ identifies the trade-off Country $i$ faces between disrupting their own output and disrupting the output of Country $j$. As long as Country $i$ prefers larger disruptions to Country $j$ and smaller disruptions to itself, it will choose disruptions that implement some point on this disruption possibility frontier. 

Figure \ref{fig:disruption_possibility_frontier} illustrates the disruption possibility frontier 
for the example from Figure \ref{fig:power_calculation}. The disruption possibility frontier is piecewise linear. $\text{Power}_{i j}$ is given by the slope of the disruption possibility frontier at $\lambda=0$; which corresponds to the optimal disruptions that fall short of completely disrupting any of $i$'s flows in the supply network.  

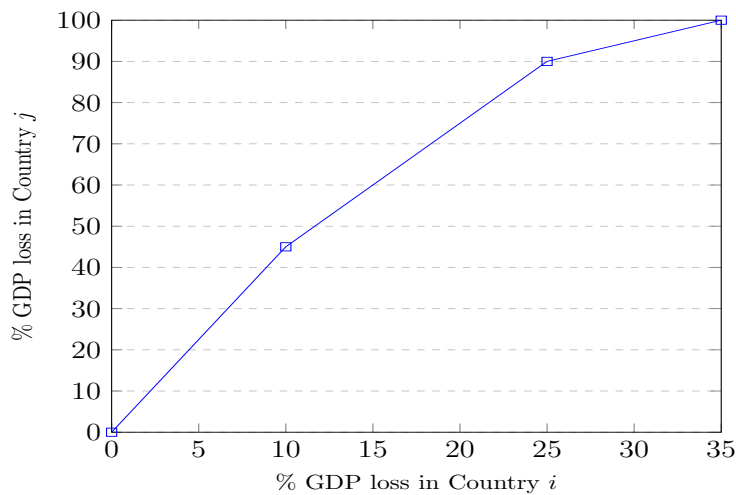
\begin{figure}[!ht]

\begin{center}
\resizebox{4in}{2.6in}{		
\begin{tikzpicture}
\begin{axis}[
    xlabel={\footnotesize{\% GDP loss in Country $i$}},
    ylabel={\footnotesize{\% GDP loss in Country $j$}},
    xmin=0, xmax=35,
    ymin=0, ymax=100,
    xtick={0,5,10,15,20,25,30,35},
    ytick={0,10,20,30,40,50,60,70,80,90,100},
    legend pos=north west,
    ymajorgrids=true,
    grid style=dashed,
]

\addplot[
    color=blue,
    mark=square,
    ]
    coordinates {
    (0,0)(10,45)(25,90)(35,100)
    };
    
\end{axis}
\end{tikzpicture}
}
    \caption{Disruption Possibility Frontier for Country $i$ disrupting Country $j$
    for the example from Figure \ref{fig:power_calculation}.}
\label{fig:disruption_possibility_frontier}
\end{center}
\end{figure}

While our power measure is a useful single statistic and an accurate measure of maximum power for disruptions that do not zero out any trade flows, for large disruptions that do zero out some trade flows, complex disruptions (that sequentially disrupt multiple technologies) need to be considered and the full disruption possibility frontier needs to be mapped.

 In Figure \ref{fig:power2}, we illustrate how the marginal impact of a disruption of one flow can actually increase as
another flow is zeroed out. This illustrates how the disruption possibility frontier need not be concave, and why when considering complete disruptions, the power analysis becomes combinatorially complex.

\begin{figure}[!ht]
    \subfloat[Initial Equilibrium]
    {
\centering
\resizebox{3.15in}{2.43in}{		
    \begin{tikzpicture}[
sqRED/.style={rectangle, draw=black!240, fill=red!15, very thick, minimum size=7mm},
		sqREDb/.style={rectangle, draw=black!240, fill=red!15, very thick, dashed, minimum size=7mm},
		sqBLUE/.style={rectangle, draw=black!240, fill=blue!15, very thick, minimum size=7mm},
		sqBLUEb/.style={rectangle, draw=black!240, fill=blue!15, very thick, dashed, minimum size=7mm},
		roundnode/.style={ellipse, draw=black!240, fill=green!15, very thick, dashed, minimum size=7mm},
		sqGREEN/.style={rectangle, draw=black!240, fill=green!15, very thick, minimum size=7mm},
		sqGREENb/.style={rectangle, draw=black!240, fill=green!15, very thick, dashed, minimum size=7mm},
		roundRED/.style={ellipse, draw=black!240, fill=red!15, very thick, dashed, minimum size=7mm},
		roundREDb/.style={ellipse, draw=black!240, fill=red!15, very thick, minimum size=7mm},
		roundYELL/.style={ellipse, draw=black!240, fill=yellow!20, very thick, dashed, minimum size=7mm},
		sqYELL/.style={rectangle, draw=black!240, fill=yellow!20, very thick, minimum size=7mm},
		roundORA/.style={ellipse, draw=black!240, fill=orange!15, very thick, dashed, minimum size=7mm},		
		sqORA/.style={rectangle, draw=black!240, fill=orange!15, very thick, minimum size=7mm},
		sqORAb/.style={rectangle, draw=black!240, fill=orange!15, very thick, dashed, minimum size=7mm},
		sqMAG/.style={rectangle, draw=black!240, fill=magenta!5, very thick, minimum size=7mm},
		squaredBLACK/.style={rectangle, draw=black!240, fill=white!15, very thick, minimum size=7mm},	
		roundBROWN/.style={ellipse, draw=black!240, fill=brown!15, very thick, dashed, minimum size=7mm},
		sqBROWN/.style={rectangle, draw=black!240, fill=brown!15, very thick, minimum size=7mm},		
        ]

\node[](HG1A) at (1.5, 3.25) {};
\node[](HG1B) at (-6, -3.25) {};
\node [he, draw=black!240, fill=lime!5, fit=(HG1A) (HG1B)] {};

\node[](HG2A) at (2.5, 3.25) {};
\node[](HG2B) at (6, -3.25) {};
\node [he, draw=black!240, fill=cyan!5, fit=(HG2A) (HG2B)] {};

		\node[sqORA](A1) at (-4, 3) {$\tau1, L=10$};
        \node[] at (-3.5, 1.5) { $10$};

        \node[sqBLUE](A2) at (0, 3) {$\tau2,  L=10$};
        \node[] at (0.5, 2) { $5$};
			
        \node[sqBLUE](A3) at (4, 3) {$\tau3,  L=10$};
        \node[] at (3.5, 1.5) { $5$};
			
		\node[sqYELL](C1) at (0, 0) {$\tau3, L=20$};
       \node[] at (-3.5, -1.5) { $10$};
       \node[] at (3.5, -1.5) { $10$};
	    \node[sqGREEN](F1) at (-4, -3) {$\tau4,L=10$};
		\node[sqGREEN](F2) at (4, -3) {$\tau5, L = 10$};	
		
		\tikzset{thick edge/.style={-, black, fill=none, thick, text=black}}
		\tikzset{thick arc/.style={->, black, fill=black, thick, >=stealth, text=black}}
		
		\draw[line width=0.4mm, black,-> ] (A1) to[out=-90,in=90] (C1);
		\draw[line width=0.4mm, black,-> ] (A2) to[out=-90,in=90] (C1);
       	\draw[line width=0.4mm, black,-> ] (A3) to[out=-90,in=90] (C1);
        \draw[line width=0.4mm, black,-> ] (C1) to[out=-90,in=90] (F1);
		\draw[line width=0.4mm, black,-> ] (C1) to[out=-90,in=90] (F2);
		\end{tikzpicture}
		}
        }
    \hfill
    \subfloat[Marginal Disruption 1]
    {
\centering
\resizebox{3.15in}{2.43in}{		
        \begin{tikzpicture}[
sqRED/.style={rectangle, draw=black!240, fill=red!15, very thick, minimum size=7mm},
		sqREDb/.style={rectangle, draw=black!240, fill=red!15, very thick, dashed, minimum size=7mm},
		sqBLUE/.style={rectangle, draw=black!240, fill=blue!15, very thick, minimum size=7mm},
		sqBLUEb/.style={rectangle, draw=black!240, fill=blue!15, very thick, dashed, minimum size=7mm},
		roundnode/.style={ellipse, draw=black!240, fill=green!15, very thick, dashed, minimum size=7mm},
		sqGREEN/.style={rectangle, draw=black!240, fill=green!15, very thick, minimum size=7mm},
		sqGREENb/.style={rectangle, draw=black!240, fill=green!15, very thick, dashed, minimum size=7mm},
		roundRED/.style={ellipse, draw=black!240, fill=red!15, very thick, dashed, minimum size=7mm},
		roundREDb/.style={ellipse, draw=black!240, fill=red!15, very thick, minimum size=7mm},
		roundYELL/.style={ellipse, draw=black!240, fill=yellow!20, very thick, dashed, minimum size=7mm},
		sqYELL/.style={rectangle, draw=black!240, fill=yellow!20, very thick, minimum size=7mm},
		roundORA/.style={ellipse, draw=black!240, fill=orange!15, very thick, dashed, minimum size=7mm},		
		sqORA/.style={rectangle, draw=black!240, fill=orange!15, very thick, minimum size=7mm},
		sqORAb/.style={rectangle, draw=black!240, fill=orange!15, very thick, dashed, minimum size=7mm},
		sqMAG/.style={rectangle, draw=black!240, fill=magenta!5, very thick, minimum size=7mm},
		squaredBLACK/.style={rectangle, draw=black!240, fill=white!15, very thick, minimum size=7mm},	
		roundBROWN/.style={ellipse, draw=black!240, fill=brown!15, very thick, dashed, minimum size=7mm},
		sqBROWN/.style={rectangle, draw=black!240, fill=brown!15, very thick, minimum size=7mm},		
        ]

\node[](HG1A) at (1.5, 3.25) {};
\node[](HG1B) at (-6, -3.25) {};
\node [he, draw=black!240, fill=lime!5, fit=(HG1A) (HG1B)] {};

\node[](HG2A) at (2.5, 3.25) {};
\node[](HG2B) at (6, -3.25) {};
\node [he, draw=black!240, fill=cyan!5, fit=(HG2A) (HG2B)] {};

		\node[sqORA](A1) at (-4, 3) {$\tau1, L=\cancel{10}$ \textcolor{red}{$5$}};
        
        \node[] at (-3.5, 1.5) { $\cancel{10}$ \textcolor{red}{$5$}};

        \node[sqBLUE](A2) at (0, 3) {$\tau2,  L=10$};
        \node[] at (0.5, 2) { $5$};
			
        \node[sqBLUE](A3) at (4, 3) {$\tau3,  L=\cancel{10}$ \textcolor{red}{$0$}};
        \node[] at (3.5, 1.5) {  $\cancel{5}$ \textcolor{red}{$0$}};
			
		\node[sqYELL](C1) at (0, 0) {$\tau3, L=\cancel{20}$ \textcolor{red}{$10$}};
       \node[] at (-3.5, -1.5) { $\cancel{10}$ \textcolor{red}{$5$}};
       \node[] at (3.5, -1.5) { $\cancel{10}$ \textcolor{red}{$5$}};
	    \node[sqGREEN](F1) at (-4, -3) {$\tau4,L=\cancel{10}$ \textcolor{red}{$5$}};
		\node[sqGREEN](F2) at (4, -3) {$\tau5, L=\cancel{10}$ \textcolor{red}{$5$}};	
		
		\tikzset{thick edge/.style={-, black, fill=none, thick, text=black}}
		\tikzset{thick arc/.style={->, black, fill=black, thick, >=stealth, text=black}}
		
		\draw[line width=0.4mm, black,-> ] (A1) to[out=-90,in=90] (C1);
		\draw[line width=0.4mm, black,-> ] (A2) to[out=-90,in=90] (C1);
       	\draw[line width=0.4mm, black,-> ] (A3) to[out=-90,in=90] (C1);
        \draw[line width=0.4mm, black,-> ] (C1) to[out=-90,in=90] (F1);
		\draw[line width=0.4mm, black,-> ] (C1) to[out=-90,in=90] (F2);
		\end{tikzpicture}
		}
        }
           \hfill
    \subfloat[Marginal Disruption 2]
    {
\centering
\resizebox{3.15in}{2.43in}{		
    \begin{tikzpicture}[
sqRED/.style={rectangle, draw=black!240, fill=red!15, very thick, minimum size=7mm},
		sqREDb/.style={rectangle, draw=black!240, fill=red!15, very thick, dashed, minimum size=7mm},
		sqBLUE/.style={rectangle, draw=black!240, fill=blue!15, very thick, minimum size=7mm},
		sqBLUEb/.style={rectangle, draw=black!240, fill=blue!15, very thick, dashed, minimum size=7mm},
		roundnode/.style={ellipse, draw=black!240, fill=green!15, very thick, dashed, minimum size=7mm},
		sqGREEN/.style={rectangle, draw=black!240, fill=green!15, very thick, minimum size=7mm},
		sqGREENb/.style={rectangle, draw=black!240, fill=green!15, very thick, dashed, minimum size=7mm},
		roundRED/.style={ellipse, draw=black!240, fill=red!15, very thick, dashed, minimum size=7mm},
		roundREDb/.style={ellipse, draw=black!240, fill=red!15, very thick, minimum size=7mm},
		roundYELL/.style={ellipse, draw=black!240, fill=yellow!20, very thick, dashed, minimum size=7mm},
		sqYELL/.style={rectangle, draw=black!240, fill=yellow!20, very thick, minimum size=7mm},
		roundORA/.style={ellipse, draw=black!240, fill=orange!15, very thick, dashed, minimum size=7mm},		
		sqORA/.style={rectangle, draw=black!240, fill=orange!15, very thick, minimum size=7mm},
		sqORAb/.style={rectangle, draw=black!240, fill=orange!15, very thick, dashed, minimum size=7mm},
		sqMAG/.style={rectangle, draw=black!240, fill=magenta!5, very thick, minimum size=7mm},
		squaredBLACK/.style={rectangle, draw=black!240, fill=white!15, very thick, minimum size=7mm},	
		roundBROWN/.style={ellipse, draw=black!240, fill=brown!15, very thick, dashed, minimum size=7mm},
		sqBROWN/.style={rectangle, draw=black!240, fill=brown!15, very thick, minimum size=7mm},		
        ]

\node[](HG1A) at (1.5, 3.25) {};
\node[](HG1B) at (-6, -3.25) {};
\node [he, draw=black!240, fill=lime!5, fit=(HG1A) (HG1B)] {};

\node[](HG2A) at (2.5, 3.25) {};
\node[](HG2B) at (6, -3.25) {};
\node [he, draw=black!240, fill=cyan!5, fit=(HG2A) (HG2B)] {};

		\node[sqORA](A1) at (-4, 3) {$\tau1, L=\cancel{10}$ \textcolor{red}{$5$}};
        \node[] at (-3.5, 1.5) { $\cancel{10}$ \textcolor{red}{$5$}};

        \node[sqBLUE](A2) at (0, 3) {$\tau2,  L=\cancel{10}$ \textcolor{red}{$5$}};
        \node[] at (0.5, 2) { $\cancel{5}$ \textcolor{red}{$2.5$}};
			
        \node[sqBLUE](A3) at (4, 3) {$\tau3,  L=\cancel{10}$ \textcolor{red}{$5$}};
        \node[] at (3.5, 1.5) { $\cancel{5}$ \textcolor{red}{$2.5$}};
			
		\node[sqYELL](C1) at (0, 0) {$\tau3, L=\cancel{20}$ \textcolor{red}{$10$}};
       \node[] at (-3.5, -1.5) { $10$};
       \node[] at (3.5, -1.5) { $\cancel{10}$ \textcolor{red}{$0$}};
	    \node[sqGREEN](F1) at (-4, -3) {$\tau4,L=10$};
		\node[sqGREEN](F2) at (4, -3) {$\tau5, L=\cancel{10}$ \textcolor{red}{$0$}};	
		
		\tikzset{thick edge/.style={-, black, fill=none, thick, text=black}}
		\tikzset{thick arc/.style={->, black, fill=black, thick, >=stealth, text=black}}
		
		\draw[line width=0.4mm, black,-> ] (A1) to[out=-90,in=90] (C1);
		\draw[line width=0.4mm, black,-> ] (A2) to[out=-90,in=90] (C1);
       	\draw[line width=0.4mm, black,-> ] (A3) to[out=-90,in=90] (C1);
        \draw[line width=0.4mm, black,-> ] (C1) to[out=-90,in=90] (F1);
		\draw[line width=0.4mm, black,-> ] (C1) to[out=-90,in=90] (F2);
		\end{tikzpicture}
        		}
        }
           \hfill
    \subfloat[Combined Disruption]
    {
\centering
\resizebox{3.15in}{2.43in}{		
    \begin{tikzpicture}[
sqRED/.style={rectangle, draw=black!240, fill=red!15, very thick, minimum size=7mm},
		sqREDb/.style={rectangle, draw=black!240, fill=red!15, very thick, dashed, minimum size=7mm},
		sqBLUE/.style={rectangle, draw=black!240, fill=blue!15, very thick, minimum size=7mm},
		sqBLUEb/.style={rectangle, draw=black!240, fill=blue!15, very thick, dashed, minimum size=7mm},
		roundnode/.style={ellipse, draw=black!240, fill=green!15, very thick, dashed, minimum size=7mm},
		sqGREEN/.style={rectangle, draw=black!240, fill=green!15, very thick, minimum size=7mm},
		sqGREENb/.style={rectangle, draw=black!240, fill=green!15, very thick, dashed, minimum size=7mm},
		roundRED/.style={ellipse, draw=black!240, fill=red!15, very thick, dashed, minimum size=7mm},
		roundREDb/.style={ellipse, draw=black!240, fill=red!15, very thick, minimum size=7mm},
		roundYELL/.style={ellipse, draw=black!240, fill=yellow!20, very thick, dashed, minimum size=7mm},
		sqYELL/.style={rectangle, draw=black!240, fill=yellow!20, very thick, minimum size=7mm},
		roundORA/.style={ellipse, draw=black!240, fill=orange!15, very thick, dashed, minimum size=7mm},		
		sqORA/.style={rectangle, draw=black!240, fill=orange!15, very thick, minimum size=7mm},
		sqORAb/.style={rectangle, draw=black!240, fill=orange!15, very thick, dashed, minimum size=7mm},
		sqMAG/.style={rectangle, draw=black!240, fill=magenta!5, very thick, minimum size=7mm},
		squaredBLACK/.style={rectangle, draw=black!240, fill=white!15, very thick, minimum size=7mm},	
		roundBROWN/.style={ellipse, draw=black!240, fill=brown!15, very thick, dashed, minimum size=7mm},
		sqBROWN/.style={rectangle, draw=black!240, fill=brown!15, very thick, minimum size=7mm},		
        ]

\node[](HG1A) at (1.5, 3.25) {};
\node[](HG1B) at (-6, -3.25) {};
\node [he, draw=black!240, fill=lime!5, fit=(HG1A) (HG1B)] {};

\node[](HG2A) at (2.5, 3.25) {};
\node[](HG2B) at (6, -3.25) {};
\node [he, draw=black!240, fill=cyan!5, fit=(HG2A) (HG2B)] {};

		\node[sqORA](A1) at (-4, 3) {$\tau1, L=\cancel{10}$ \textcolor{red}{$2.5$}};
        \node[] at (-3.5, 1.5) { $\cancel{10}$ \textcolor{red}{$2.5$}};

        \node[sqBLUE](A2) at (0, 3) {$\tau2,  L=\cancel{10}$ \textcolor{red}{$5$}};
        \node[] at (0.5, 2) { $\cancel{5}$ \textcolor{red}{$2.5$}};
			
        \node[sqBLUE](A3) at (4, 3) {$\tau3,  L=\cancel{10}$ \textcolor{red}{$0$}};
        \node[] at (3.5, 1.5) { $\cancel{5}$ \textcolor{red}{$0$}};
			
		\node[sqYELL](C1) at (0, 0) {$\tau3, L=\cancel{20}$ \textcolor{red}{$5$}};
       \node[] at (-3.5, -1.5) { $\cancel{10}$ \textcolor{red}{$5$}};
       \node[] at (3.5, -1.5) { $\cancel{10}$ \textcolor{red}{$0$}};
	    \node[sqGREEN](F1) at (-4, -3) {$\tau4,L=\cancel{10}$ \textcolor{red}{$5$}};
		\node[sqGREEN](F2) at (4, -3) {$\tau5, L=\cancel{10}$ \textcolor{red}{$0$}};	
		
		\tikzset{thick edge/.style={-, black, fill=none, thick, text=black}}
		\tikzset{thick arc/.style={->, black, fill=black, thick, >=stealth, text=black}}
		
		\draw[line width=0.4mm, black,-> ] (A1) to[out=-90,in=90] (C1);
		\draw[line width=0.4mm, black,-> ] (A2) to[out=-90,in=90] (C1);
       	\draw[line width=0.4mm, black,-> ] (A3) to[out=-90,in=90] (C1);
        \draw[line width=0.4mm, black,-> ] (C1) to[out=-90,in=90] (F1);
		\draw[line width=0.4mm, black,-> ] (C1) to[out=-90,in=90] (F2);
		\end{tikzpicture}
		}
        }
    \caption{\footnotesize{For the initial equilibirum shown in Panel (a), Panels (b) and (c) show the impact of the different marginal disruptions available to Country $i$ (on the right in each panel). The power of both these disruptions is the same and constant up to the point where flows zero out. Panel (d) shows the impact of both these aforementioned disruptions being sequentially implemented (in either order). The power of this combined disruption is greater than the power of either marginal disruption.}
    \label{fig:power2}}   
  \end{figure}
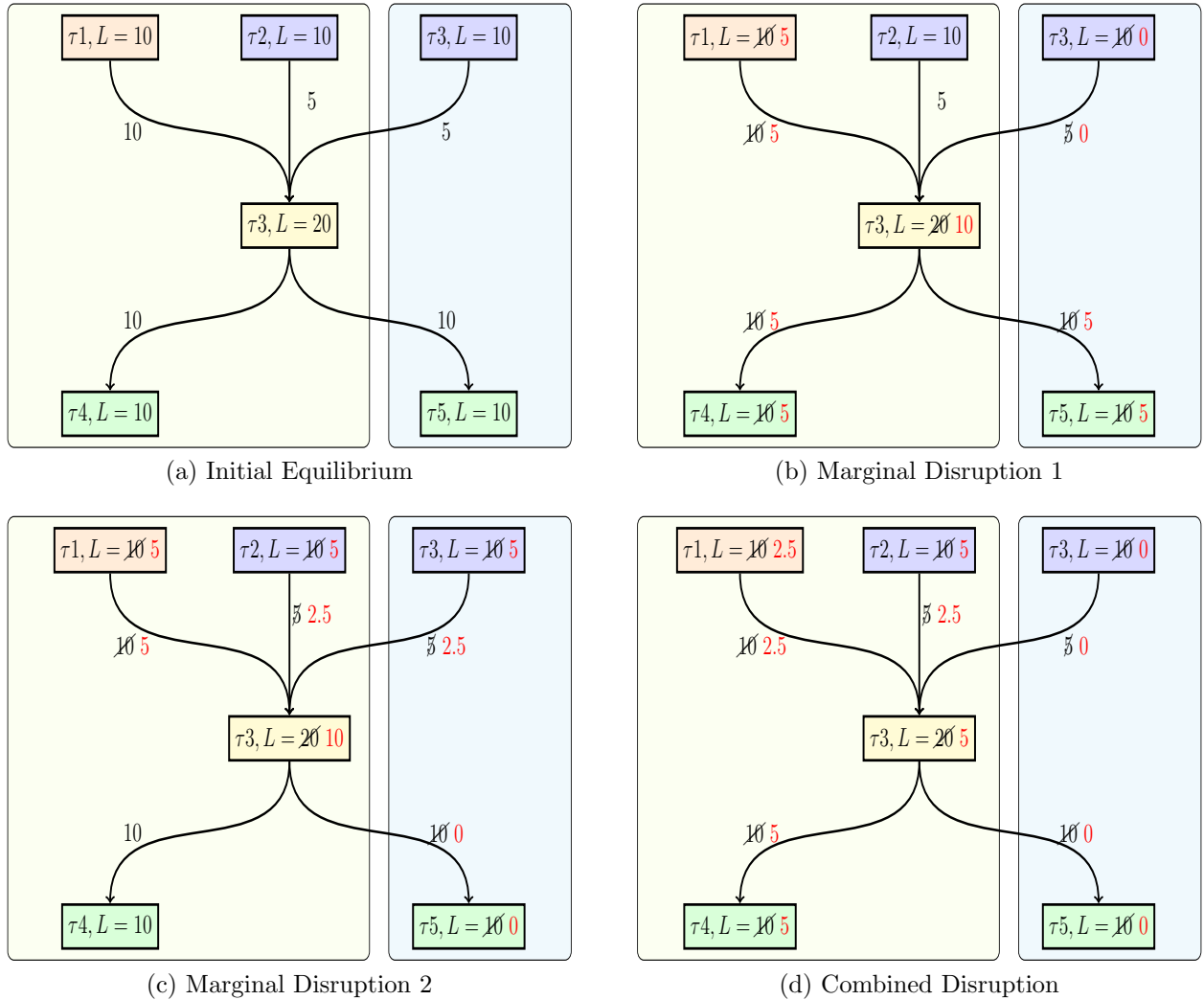

\subsection{Strategic Power When All Countries Optimize Flows}\label{sec:strategic_power}

The definitions above work with non-strategic decisions at all flows except those by the country that is doing the disruption.  This does not provide scope for other countries responding to mitigate the damage or retaliate. While in the short run other countries may not have time to respond, by the medium run a response should be expected. Moreover, if the disruption has been anticipated by the target country then it might be in a position to immediately respond.

It is straightforward to 
develop a definition of power that includes a strategic response. 
It is important to define since, as we show via examples in the appendix, having a country react strategically rather than proportionally can end up rerouting disruptions back to the country that is initiating the disruptions.  This can lead to substantial changes to the power calculation.  
While this is important to account for, it is analytically straightforward and so we provide the definitions and examples in Appendix \ref{sec:strategicpower}.

\section{Complexity, Fragility, and Globalization}

\subsection{Supply Chain Complexity and Increased Fragility}

Next we examine how the impact of a disruption depends on the complexity of a supply chain.
We show that shocks have more impact as production becomes more complex and relies on more intermediate goods, all else held equal.


Consider a setting in which the bound from Proposition \ref{prop:shocked_GDP} is obtained (e.g., the setting of Proposition \ref{prop:general_industry_shocks}), and let us consider a random shock. Let $S$ denote the average number of different (non-final) technologies used directly or indirectly to produce a final good---averaged across final goods. Thus, $S$ is a measure of the complexity of supply chains. Let $\Delta GDP$ denote the initial equilibrium $GDP$ minus the $GDP$ after any shocks. Finally, let $q$ denote the ratio of the average expenditures spent on a randomly picked (non-final) technology to the average expenditures on a randomly picked final good, and $m$ denote the expected number of final goods that lie downstream from a randomly picked (non-final) technology.

\begin{proposition}\label{prop:total_exposure}
\emph{[Supply Chain Disruption as a Function of Complexity]}
Consider an economy in which the bound from Proposition \ref{prop:shocked_GDP} binds, final goods prices are independent of the complexity of their supply chains.  Let disruptions to intermediate good technologies---in which a proportion $(1-\lambda)$ of output is lost---occur with probability $\pi$, independently of other disruptions.  Then for small $\pi$  (so that the probability of two or more shocks on the same chain is vanishingly small relative to the probability of a single shock) and small $1-\lambda$ (so that Hulten's Theorem applies):
\[
\text{Short-Run:  } \quad\mathbb{E}\left[ \frac{\Delta GDP}{GDP}\right] \approx  - (1-\lambda) \pi  S,
 \, \, \,
\text{Long-Run:  } \quad\mathbb{E}\left[ \frac{\Delta GDP}{GDP}\right] \approx  - (1-\lambda) \pi S \frac{q}{m}.
\]
\end{proposition}

Some intuition behind Proposition \ref{prop:total_exposure} is as follows. As supply chains become more complex, more intermediate goods are used (directly or indirectly) in the production of final goods, and so the probability of short-run of disruptions increases. In the long-run, complexity still matters because it still increases the probability of disruption.
However, this is now modulated by the size of the impact which is proportional to the average cost of a shocked intermediate good to the value of final goods it affects, as we would expect by Hulten's Theorem.

The proof of Proposition \ref{prop:total_exposure} is straightforward and so we simply outline it here.
First, note that $m=SF/M$.
Next, given that $ M \pi$ is expected number of disruptions, and each hits a fraction of $ m / F$ of the total set of goods,
it follows that $\pi M m/F  = \pi  S $ is the expected fraction of final goods shocked and the expected probability of disrupting a typical final good (appealing to the fact that when $\pi$ is small the probability that more than one technology is shocked at a time is vanishingly small compared to a single technology being shocked, and so we do not worry about multiple shocks to any supply chain).
In the short run, a disrupted chain reduces consumption of the final good by a fraction $(1-\lambda)$, and so disrupting a typical final good by $(1-\lambda)$ with a probability $\pi  S $ gives the short-run result.
For the long-run result, the expected number of disruptions is approximately $M\pi$ and the expenditures on the technology compared to $GDP$ is
$q/F$.  By Hulten's Theorem the expected marginal value is approximately $M\pi q/F = \pi S q/m$.

The contrast of short and long-run disruptions is again stark.
If supply chains are completely horizontal, so that inputs go directly into final goods, then  $q\approx 1/S$ (ignoring labor cost in assembling final goods, which are otherwise added to the denominator) and $m=1$, and so the long-run effect is approximately  $(1-\lambda) \pi$.  This is $1/S$ of the short-run effect.
If supply chains are completely vertical and $m=1$, and inputs use similar amounts of labor, then $q\approx 1/2$, and so the overall effect is approximately $(1-\lambda) \pi S/2$.
If there are parallel supply chains so that each input reaches a different final good, then the impact of the long run is like the horizontal case, while the short run impact is more compartmentalized.

This shows that there are systematic ways in which network shape affects both the short- and long-run impacts, and that they depend on different features of the network.
The long-run impact is more dependent upon depth vs breadth,  while the short-run impact is more dependent upon how many final goods given inputs are upstream from.

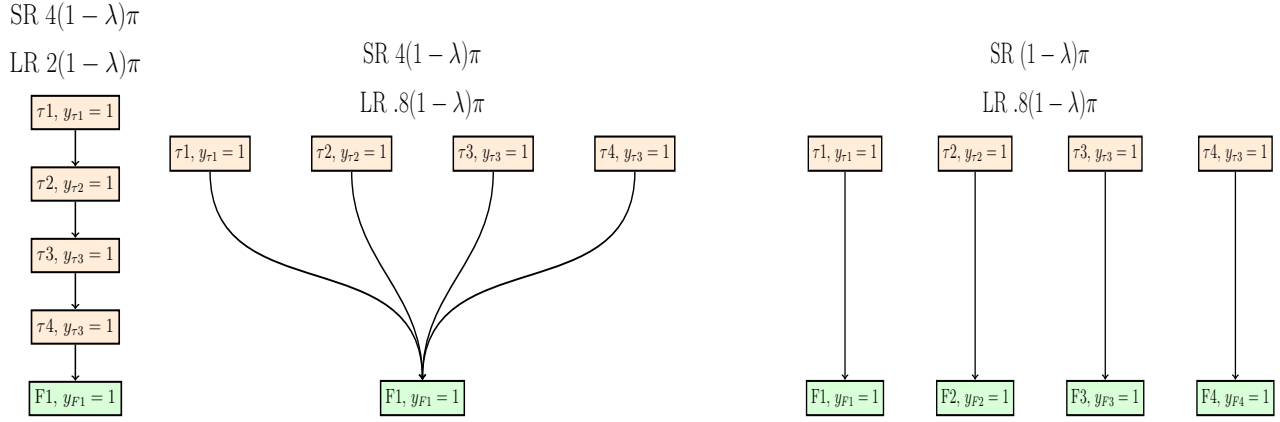
\begin{figure}[t!]
\resizebox{.8in}{2.2in}{		
\begin{tikzpicture}[
		sqRED/.style={rectangle, draw=black!240, fill=red!15, very thick, minimum size=7mm},
		sqBLUE/.style={rectangle, draw=black!240, fill=blue!15, very thick, minimum size=7mm},
		roundnode/.style={ellipse, draw=black!240, fill=green!15, very thick, dashed, minimum size=7mm},
		sqGREEN/.style={rectangle, draw=black!240, fill=green!15, very thick, minimum size=7mm},
		roundRED/.style={ellipse, draw=black!240, fill=red!15, very thick, dashed, minimum size=7mm},
		roundREDb/.style={ellipse, draw=black!240, fill=red!15, very thick, minimum size=7mm},
		roundYELL/.style={ellipse, draw=black!240, fill=yellow!20, very thick, dashed, minimum size=7mm},
		sqYELL/.style={rectangle, draw=black!240, fill=yellow!20, very thick, minimum size=7mm},
		roundORA/.style={ellipse, draw=black!240, fill=orange!15, very thick, dashed, minimum size=7mm},		
		sqORA/.style={rectangle, draw=black!240, fill=orange!15, very thick, minimum size=7mm},
		squaredBLACK/.style={rectangle, draw=black!240, fill=white!15, very thick, minimum size=7mm},	
		squaredBLACK2/.style={rectangle, draw=red!240, fill=white!15, very thick, dashed, minimum size=7mm},	
		roundBROWN/.style={ellipse, draw=black!240, fill=brown!15, very thick, dashed, minimum size=7mm},
		sqBROWN/.style={rectangle, draw=black!240, fill=brown!15, very thick, minimum size=7mm},		
        sqPURPLE/.style={rectangle, draw=black!240, fill=purple!15, very thick, minimum size=7mm},				
		]

		\node[sqORA](I1) at (-2, 3) { $\tau$1, $y_{\tau 1}=1$};
\node[] at (-2, 5) {\Large SR $4(1-\lambda)\pi$};	
\node[] at (-2, 4) {\Large LR $2(1-\lambda)\pi$};		


		\node[sqORA](I2) at (-2, 1.5) { $\tau$2, $y_{\tau 2}=1$};
	
        \node[sqORA](I3) at (-2, 0) { $\tau$3, $y_{\tau 3}=1$};

        \node[sqORA](I4) at (-2, -1.5) { $\tau$4, $y_{\tau 3}=1$};

		\node[sqGREEN](F1) at (-2, -3) { F1, $y_{F1}=1$};

		\tikzset{thick edge/.style={-, black, fill=none, thick, text=black}}
		\tikzset{thick arc/.style={->, black, fill=black, thick, >=stealth, text=black}}
		

		\draw[line width=0.4mm, black,-> ] (I1) to[out=-90,in=90] (I2);
		\draw[line width=0.4mm, black,-> ] (I2) to[out=-90,in=90] (I3);
		\draw[line width=0.4mm, black,-> ] (I3) to[out=-90,in=90] (I4);
		\draw[line width=0.4mm, black,-> ] (I4) to[out=-90,in=90] (F1);

		\end{tikzpicture}
		}
\hfill
\resizebox{2.7in}{2in}{			
\begin{tikzpicture}[
		sqRED/.style={rectangle, draw=black!240, fill=red!15, very thick, minimum size=7mm},
		sqBLUE/.style={rectangle, draw=black!240, fill=blue!15, very thick, minimum size=7mm},
		roundnode/.style={ellipse, draw=black!240, fill=green!15, very thick, dashed, minimum size=7mm},
		sqGREEN/.style={rectangle, draw=black!240, fill=green!15, very thick, minimum size=7mm},
		roundRED/.style={ellipse, draw=black!240, fill=red!15, very thick, dashed, minimum size=7mm},
		roundREDb/.style={ellipse, draw=black!240, fill=red!15, very thick, minimum size=7mm},
		roundYELL/.style={ellipse, draw=black!240, fill=yellow!20, very thick, dashed, minimum size=7mm},
		sqYELL/.style={rectangle, draw=black!240, fill=yellow!20, very thick, minimum size=7mm},
		roundORA/.style={ellipse, draw=black!240, fill=orange!15, very thick, dashed, minimum size=7mm},		
		sqORA/.style={rectangle, draw=black!240, fill=orange!15, very thick, minimum size=7mm},
		squaredBLACK/.style={rectangle, draw=black!240, fill=white!15, very thick, minimum size=7mm},	
		squaredBLACK2/.style={rectangle, draw=red!240, fill=white!15, very thick, dashed, minimum size=7mm},	
		roundBROWN/.style={ellipse, draw=black!240, fill=brown!15, very thick, dashed, minimum size=7mm},
		sqBROWN/.style={rectangle, draw=black!240, fill=brown!15, very thick, minimum size=7mm},		
        sqPURPLE/.style={rectangle, draw=black!240, fill=purple!15, very thick, minimum size=7mm},				
		]

		\node[sqORA](I1) at (-6, 0) { $\tau$1, $y_{\tau 1}=1$};
\node[] at (0, 2) {\Large SR $4(1-\lambda)\pi$};	
\node[] at (0, 1) {\Large LR $.8(1-\lambda)\pi$};			


		\node[sqORA](I2) at (-2, 0) { $\tau$2, $y_{\tau 2}=1$};
	
        \node[sqORA](I3) at (2, 0) { $\tau$3, $y_{\tau 3}=1$};

        \node[sqORA](I4) at (6, 0) { $\tau$4, $y_{\tau 3}=1$};

		\node[sqGREEN](F1) at (0, -5) { F1, $y_{F1}=1$};

		\tikzset{thick edge/.style={-, black, fill=none, thick, text=black}}
		\tikzset{thick arc/.style={->, black, fill=black, thick, >=stealth, text=black}}
		

		\draw[line width=0.4mm, black,-> ] (I1) to[out=-90,in=90] (F1);
		\draw[line width=0.4mm, black,-> ] (I2) to[out=-90,in=90] (F1);
		\draw[line width=0.4mm, black,-> ] (I3) to[out=-90,in=90] (F1);
		\draw[line width=0.4mm, black,-> ] (I4) to[out=-90,in=90] (F1);

		\end{tikzpicture}
		}	
\hfill
\resizebox{.5in}{2in}{			
\begin{tikzpicture}[
		sqRED/.style={rectangle, draw=black!240, fill=red!15, very thick, minimum size=7mm},
		sqBLUE/.style={rectangle, draw=black!240, fill=blue!15, very thick, minimum size=7mm},
		roundnode/.style={ellipse, draw=black!240, fill=green!15, very thick, dashed, minimum size=7mm},
		sqGREEN/.style={rectangle, draw=black!240, fill=green!15, very thick, minimum size=7mm},
		roundRED/.style={ellipse, draw=black!240, fill=red!15, very thick, dashed, minimum size=7mm},
		roundREDb/.style={ellipse, draw=black!240, fill=red!15, very thick, minimum size=7mm},
		roundYELL/.style={ellipse, draw=black!240, fill=yellow!20, very thick, dashed, minimum size=7mm},
		sqYELL/.style={rectangle, draw=black!240, fill=yellow!20, very thick, minimum size=7mm},
		roundORA/.style={ellipse, draw=black!240, fill=orange!15, very thick, dashed, minimum size=7mm},		
		sqORA/.style={rectangle, draw=black!240, fill=orange!15, very thick, minimum size=7mm},
		squaredBLACK/.style={rectangle, draw=black!240, fill=white!15, very thick, minimum size=7mm},	
		squaredBLACK2/.style={rectangle, draw=red!240, fill=white!15, very thick, dashed, minimum size=7mm},	
		roundBROWN/.style={ellipse, draw=black!240, fill=brown!15, very thick, dashed, minimum size=7mm},
		sqBROWN/.style={rectangle, draw=black!240, fill=brown!15, very thick, minimum size=7mm},		
        sqPURPLE/.style={rectangle, draw=black!240, fill=purple!15, very thick, minimum size=7mm},				
		]
\node[] at (0, 2) {};	
		\end{tikzpicture}
		}	
\hfill
\resizebox{2.5in}{2in}{			
\begin{tikzpicture}[
		sqRED/.style={rectangle, draw=black!240, fill=red!15, very thick, minimum size=7mm},
		sqBLUE/.style={rectangle, draw=black!240, fill=blue!15, very thick, minimum size=7mm},
		roundnode/.style={ellipse, draw=black!240, fill=green!15, very thick, dashed, minimum size=7mm},
		sqGREEN/.style={rectangle, draw=black!240, fill=green!15, very thick, minimum size=7mm},
		roundRED/.style={ellipse, draw=black!240, fill=red!15, very thick, dashed, minimum size=7mm},
		roundREDb/.style={ellipse, draw=black!240, fill=red!15, very thick, minimum size=7mm},
		roundYELL/.style={ellipse, draw=black!240, fill=yellow!20, very thick, dashed, minimum size=7mm},
		sqYELL/.style={rectangle, draw=black!240, fill=yellow!20, very thick, minimum size=7mm},
		roundORA/.style={ellipse, draw=black!240, fill=orange!15, very thick, dashed, minimum size=7mm},		
		sqORA/.style={rectangle, draw=black!240, fill=orange!15, very thick, minimum size=7mm},
		squaredBLACK/.style={rectangle, draw=black!240, fill=white!15, very thick, minimum size=7mm},	
		squaredBLACK2/.style={rectangle, draw=red!240, fill=white!15, very thick, dashed, minimum size=7mm},	
		roundBROWN/.style={ellipse, draw=black!240, fill=brown!15, very thick, dashed, minimum size=7mm},
		sqBROWN/.style={rectangle, draw=black!240, fill=brown!15, very thick, minimum size=7mm},		
        sqPURPLE/.style={rectangle, draw=black!240, fill=purple!15, very thick, minimum size=7mm},				
		]

		\node[sqORA](I1) at (-6, 0) { $\tau$1, $y_{\tau 1}=1$};
		


		\node[sqORA](I2) at (-2, 0) { $\tau$2, $y_{\tau 2}=1$};
	
        \node[sqORA](I3) at (2, 0) { $\tau$3, $y_{\tau 3}=1$};

        \node[sqORA](I4) at (6, 0) { $\tau$4, $y_{\tau 3}=1$};
\node[] at (0, 2) {\Large SR $(1-\lambda)\pi$};	
\node[] at (0, 1) {\Large LR $.8(1-\lambda)\pi$};		
		\node[sqGREEN](F1) at (-6, -5) { F1, $y_{F1}=1$};
		\node[sqGREEN](F2) at (-2, -5) { F2, $y_{F2}=1$};
		\node[sqGREEN](F3) at (2, -5) { F3, $y_{F3}=1$};
		\node[sqGREEN](F4) at (6, -5) { F4, $y_{F4}=1$};

		\tikzset{thick edge/.style={-, black, fill=none, thick, text=black}}
		\tikzset{thick arc/.style={->, black, fill=black, thick, >=stealth, text=black}}
		

		\draw[line width=0.4mm, black,-> ] (I1) to[out=-90,in=90] (F1);
		\draw[line width=0.4mm, black,-> ] (I2) to[out=-90,in=90] (F2);
		\draw[line width=0.4mm, black,-> ] (I3) to[out=-90,in=90] (F3);
		\draw[line width=0.4mm, black,-> ] (I4) to[out=-90,in=90] (F4);
	

		\end{tikzpicture}
		}
		  \caption{\footnotesize{Three different supply chain configurations and the corresponding impacts.  In each case the endowment of labor is 5 units and each input technology uses one unit of labor, with the remaining unit of labor used in the production of the final good(s).  In the vertical and horizontal supply chains the complexity is $S=4$ while it is
		  $S=1$ in the parallel case.  The number of final goods downstream of any input is $m=1$ in each case, and the corresponding $q$s are $.5, .2, .8$.   The corresponding price vectors are
		   $p_{vertical}=\left(\frac{1}{5},\frac{1}{5},\frac{2}{5},\frac{3}{5},\frac{4}{5},{1}\right)$,
		    $p_{horizontal}=\left(\frac{1}{5},\frac{1}{5},\frac{1}{5},\frac{1}{5},\frac{1}{5},{1}\right)$,
		    $p_{parallel}=\left(\frac{1}{5},\frac{1}{5},\frac{1}{5},\frac{1}{5},\frac{1}{5},\frac{1}{4},\frac{1}{4},\frac{1}{4},\frac{1}{4}\right)$.
		    }}
    \label{fig:horivert}
\end{figure}

The contrast as a function of the supply network shape is illustrated in Figure \ref{fig:horivert}.  We see that the short run is the same regardless of whether the supply chain is horizontal or vertical, since each input still disrupts the final good.  By contrast, when the supply chains are parallel with the same number of inputs but used for four different final goods, then the short-run disruption is lowered.
The long-run disruption shows very different patterns.  Not only are those disruptions much smaller, but here they differ between the horizontal and vertical,  but not between the horizontal and parallel, as vertical supply chains build up costs, while the input costs are the same across the horizontal and parallel.

\subsection{Globalization}

We next discuss what happens when globalization, modeled as reduced transportation costs,  leads to
increased specialization and the conditions hypothesized in Proposition \ref{prop:general_industry_shocks} apply.
The idea is simple: As transportation costs become sufficiently low, technological diversity decreases and each good becomes single-sourced by the least costly technology, and technology shocks become industry shocks.

In particular, suppose that the set of world technologies is such that if there were no transportation costs, then no two distinct technologies that produce the same good use exactly the same
amount of labor (in total, directly plus indirectly in terms of their most efficient equilibrium input mix).   This is effectively a generic property in terms of sets of technologies $T$'s, as it is a knife-edge case that two different technologies end up using exactly the same amount of labor.

In this case, once transportation costs drop sufficiently, any good that is produced in equilibrium ends up being sourced from the same technology.

The fact that this implies that the bound from Proposition \ref{prop:shocked_GDP} binds can be seen as follows. Every affected final good technology has a path from some shocked technology to it. Since goods are all single-sourced, each technology on the path between an affected final good technology and a shocked technology is fully reduced by $(1-\lambda)$. Intuitively, as all transportation costs decrease sufficiently we get specialization, which eliminates technological or sourcing diversity and turns technology-specific shocks into industry-wide shocks. Thus, by Proposition \ref{prop:general_industry_shocks}, the bound from Proposition \ref{prop:shocked_GDP} binds.

Another implication 
is that moving to a frictionless economy changes fragility.
More generally  consolidation of supply chains changes fragility. Consolidation of final good production tends to make shocks more severe when they hit, while consolidation of supply chains further upstream of a final goods producer tends to make disruptions to that final goods producer less likely.

\begin{proposition}\label{prop:consolidation}
Consider two economies, indexed $1,2$, that share the same technology sets $T_1,\ldots,T_n$ and have equilibria that satisfy the bound from Proposition \ref{prop:shocked_GDP} from the shock to any active technology. Consider some final good $f$ that is produced in equilibrium by both economies by some technology $\tau^f$, and suppose that
the equilibrium output of technology $\tau^f$ is higher in economy 2 than 1. If the set of technologies that lie on a directed path to $\tau^f$ is smaller in economy 2 (i.e.,  $\mathcal{G}^2(\tau^f)\subsetneq \mathcal{G}^1(\tau^f)$), and shocks are independent across technologies with the same proportional disruption, then the probability of a disruption to $\tau^f$ is lower in economy 2 than economy 1
\emph{in both the short and long run}, but the expected short-run size of that disruption conditional upon occurrence is higher in economy 2 than 1.\footnote{The expected size of the long-run disruption impact is ambiguous as it depends on the expected cost of the disrupted technology.}
\end{proposition}

Proposition \ref{prop:consolidation} is straightforward to prove, and so we omit its proof.\footnote{The main nuance is that there is a higher chance in economy 1 that multiple upstream technologies are shocked simultaneously, but given that the bound holds, then whether one or more technologies are shocked by the same amount, the consequences are identical.}
Nonetheless, Proposition \ref{prop:consolidation} is important to state.  It has as a corollary that if a lowering of trade costs leads to the consolidation of supply chains and---via improved efficiency---larger outputs of final goods from some single technology, then there are fewer technologies upstream from the final good that might fail, but greater consequences when any of them are shocked.

The trade off between the probability of a disruption and the size of the disruption is a fundamental point that applies to increasingly specialized supply chains. Of course, this holds all else constant.  As new technologies emerge due to globalization the production network evolves and supply chains can become more complex and cross more borders. This can expose them more to shocks such as political disruptions and issues that disrupt shipping. Moreover, the overall distribution of shocks can co-evolve with the supply network. This opens the possibility that one could see both an increase in the probability of a disruption {\sl and} the size of the disruption.

\section{Concluding Remarks}

We comment on the flexibility of the model and identify further explorations for which the model can serve as a foundation.

\noindent \textbf{Sanctions}
One potential further application of the model is to estimate the impact of sanctions.
Our approach would provide an estimate of the impact of targeted sanctions, and can be used as a foundation for selectively targeting goods and services that would most impact parts of the world economy and not others.
As is clear from our analysis, the short-run and long-run impacts of sanctions can differ dramatically, depending on the ability of target countries or industries to reallocate over time.\footnote{See \cite{Moll2022} for an analysis of the impact of Russian energy sanctions for the German economy.}

\noindent \textbf{Edge shocks}
We specified shocks as being incident on nodes (country-specific industries or technologies), but
our analysis extends directly to analyze edge shocks. Suppose for example that a good is produced
and output is shipped through the Red Sea and other output is shipped elsewhere, and we want to consider the possibility of a shock that disrupts shipping through the Red Sea. Then we can simply divide the original technology into two different nodes, one that ships via the Red Sea and another that ships elsewhere.   Shocking the first node is equivalent to shocking the Red Sea edge of the original technology, and we have seen that such shocks can cause relatively severe short term disruptions.\footnote{ For example, in an article title ``Tesla to Halt Production in Germany as Red Sea Conflict Hits Supply Chains,'' The Wall Street Journal reports that Tesla had to shut down its main European factory for two weeks following the Houthi attacks in the Red Sea (January 12, 2024).}

\noindent \textbf{Inventories and Endogenous Robustness}
In our benchmark of a short-run disruption, disrupted goods are fully missing from production.
Firms can maintain (costly) inventories of inputs to avoid issues with supply chain disruptions.   An existing inventory can buffer some of the impact depending on a disruption's magnitude and how long it lasts.
Firms might even maintain alternative production technologies, some which are inefficient, but can serve as backups in times of disruption.
This differs across industries and the perceived dangers faced from disruptions.

Although producers prefer to avoid disruptions, excess inventory costs (beyond minimum ones to run production processes) may not be compensated in the face of competitive pressures, and thus excess inventories tend to be low.  To the extent that fully contingent contracts are not written in advance for final consumption goods (and, indeed, most goods are sold on spot markets), market incompleteness favors firms with lower costs and thus pressure them to avoid excess inventory costs.  Firms that are more robust, gain by sales when others are disrupted, but to the extent that they cannot capitalize on those profits ex ante (due to incomplete contracts), the market may be inefficient.
This is an interesting topic for further consideration and, again, our results provide a foundation on which to build.

\noindent \textbf{Services}
While it might be natural to think of the model in terms of physical goods rather than services, the model includes services (as both intermediate and final `goods'). Inputs, including labor, are used to create outputs in the same way, and disruptions propagate through the supply network equivalently.  Given that labor is often less mobile than other (intermediate) goods, it might be that structure of networks around services vary depending on what types of service it is: e.g., coding which can be done remotely, versus healthcare which is done more locally.  The model provides the structure with which to estimate shock effects and how those differ across different supply network structures, and this application presents an interesting future agenda.

\noindent \textbf{Endogenous Power}
Our power calculations provide an interesting foundation for further analysis of endogenous supply chains.  Choosing from where to source different imports could impact the leverage that one country has over another.   For example,  \cite{de2019disentangling} provides evidence that Mexican  exports to the US use a higher fraction of US-sourced goods than Mexican exports to other countries.  This would lower the power the US would have by lowering Mexican imports, as it would feedback to US sourced inputs.  Such strategic supply chain construction is another interesting open question.

{
\linespread{1}\selectfont
\bibliographystyle{abbrv}
\bibliography{globalization}
}

\newpage
\appendix

{\large \center{\bf Supplementary Appendix for ``Supply Chain Disruptions, the Structure of Production
Networks, and the Impact of Globalization,'' by
Matthew L. Elliott and Matthew O. Jackson}}

\bigskip
\setcounter{section}{0}
\setcounter{page}{1}

\renewcommand{\thesection}{\Roman{section}}

\section{Omitted Proofs}\label{sec:proofs}

\noindent {\bf Proof of Proposition~\ref{prop:hulten}:}

We first prove the result for final goods, and then extend it to intermediate goods.

Euler's Homogeneous Function Theorem implies that for any $(c_1, \ldots,   c_F)\in \Re^F_+$:
 \begin{equation}
 \label{euler}
 U(c_1, \ldots,   c_F) =  \sum_f c_f \frac{\partial U}{\partial c_f}.
 \end{equation}

By Lemma \ref{lemma:representative_consumer} total world consumption choices are as if there is a representative consumer with preferences represented by the homogeneous of degree $1$ utility function $U(c_1, \ldots,   c_F)$ and wealth $I=\sum_n L_n p_{Ln}$.
The representative consumer's problem is choosing non-negative amounts of the final goods to consume to
$${\rm maximize \ \ } U(c_1, \ldots, c_F) {\rm \ \  subject \ \ to \ \ } \sum_f p_f c_f \leq I.$$
Thus, in equilibrium,
$$\frac{\partial U}{\partial c_f} = \lambda p_f$$
for all final goods consumed in positive amounts, where $\lambda>0$ is the Lagrange multiplier on the wealth constraint.
Thus, from (\ref{euler}) it follows that in equilibrium
\begin{equation}
U(c_1, \ldots,   c_F) =  \lambda \sum_f c_f p_f.
\label{eq:U}
\end{equation}

Moreover, in equilibrium, constraint (\ref{clearfinal}) holds. Allowing $\tau_f$ to change from one this implies that:\footnote{One can interpret $y_\tau$ as a number of units of operation of the technology, and then $\tau_f$ different from one scales the amount produced.}
\begin{equation}
c_{f} = \sum_{\tau \in T: O(\tau )=f} \tau_f y_{\tau}
 \label{clearfinal2}
\end{equation}
Substituting (\ref{clearfinal2}) into (\ref{eq:U}) yields
\begin{equation}
U(c_1, \ldots,   c_F) =  \lambda \sum_f p_f \sum_{\tau \in T: O(\tau )=f} \tau_f y_{\tau}.
\label{eq:U2}
\end{equation}
As this is an equilibrium expression for utility, the envelope theorem can be applied and hence, for a given production technology $\tau$ with $O(\tau)=f$,
\begin{equation}
\frac{\partial U}{\partial \tau_f} = \lambda p_f y_{\tau} \label{eq:partial_U}
\end{equation} Then from (\ref{eq:partial_U}) and (\ref{eq:U}),
$$
\frac{\partial \log(U)}{\partial \log(\tau_f)} = \left(\frac{\partial U}{\partial \tau_f}\right)\left(\frac{\tau_f}{U}\right) =  \frac{\lambda p_f y_{\tau}\tau_f}{U}\biggr\rvert_{\tau_f=1}= \frac{p_f y_{\tau}}{\sum_f c_f p_f}.
$$
Note that as $GDP = \sum_f p_f c_f,$ by (\ref{eq:U}) $U= \lambda GDP$, and so we also have
$$
\frac{\partial \log(U)}{\partial \log(\tau_f)} =  \frac{\partial [ \log(GDP)+\log (\lambda)]}{\partial \log(\tau_f)} = \frac{\partial  \log(GDP)}{\partial \log(\tau_f)}
 = \frac{p_f y_{\tau}}{GDP}.
$$
We now extend this to intermediate goods.

The zero profit conditions for each firm allow the revenues of a firm to be equated to its costs, which can be expressed in terms of its suppliers' revenues, which are also equated to costs, and so on. Repeating this process, if a technology $\widehat{\tau}$ produces an intermediate good that is used directly or indirectly in the production of a final good, the revenues generated by sales of this final good can be expressed in terms of the revenues generated by technology $\widehat{\tau}$, and the remaining direct and indirect labor costs associated with the production of the final good.

Consider a technology $\tau$ with $O(\tau)=f\in F$. By the zero profit condition for technology $\tau$

$$p_f y_{\tau}=\sum_{k\in I(\tau)}\sum_{\tau':O(\tau')=k} p_{\tau'}x_{\tau'\tau}+\sum_n p_n x_{n\tau}.$$

But, for an input technology $\tau'$ such that $O(\tau')=k\in I(\tau)$ and $x_{\tau'\tau}>0$, we also have, by the zero profit condition

$$p_{\tau'} x_{\tau'\tau}=\left(\frac{x_{\tau'\tau}}{y_{\tau'}}\right)\left(\sum_{k'\in I(\tau')}\sum_{\tau'':O(\tau'')=k'} p_{\tau''}x_{\tau''\tau'}+\sum_n p_n x_{n\tau'}\right).$$

Iteratively substituting in these expressions for all intermediate good technologies except the shocked one, $\widehat{\tau}$, the obtained expression converges to one with the following form:

\begin{equation}p_f y_{\tau} = p_{\widehat{\tau}} \widehat{x}_{\widehat{\tau} \tau} +  \sum_n p_n \widehat{x}_{n\tau},\label{eq:telescoped_inputs}\end{equation}

where $\widehat{x}_{\widehat{\tau} \tau}$ is the amount of good $O(\widehat{\tau})$ produced by technology $\widehat{\tau}$ that ultimately ends up being used (directly or indirectly) by final good technology $\tau$ and $\widehat{x}_{n\tau}$ is the amount of labor, other than that used by technology $\widehat{\tau}$, from country $n$ that ultimately ends up being used (directly or indirectly) by final good technology $\tau$.

Note that as there is no waste in equilibrium all production of an intermediate goods can be assigned to final goods and hence \begin{equation}\sum_{f\in F}\sum_{\tau:O(\tau)=f}\widehat{x}_{\widehat{\tau} \tau}=y_{\widehat{\tau}}.\label{eq:presevered_output}\end{equation}

Consider an intermediate good technology $\widehat{\tau}$ with $O(\widehat{\tau})=k$, and substitute \ref{eq:telescoped_inputs} into \ref{eq:U2}, setting $\tau_f=1$ (as we are not varying it) and adding in $\widehat{\tau}_{k}=1$ (as we will be varying it). This gives
\begin{equation}
U(c_1, \ldots,   c_F) =  \lambda \sum_f \sum_{\tau \in T: O(\tau )=f} \left(\widehat{\tau}_{k} p_{\widehat{\tau}} \widehat{x}_{\widehat{\tau} \tau} +  \sum_n p_n \widehat{x}_{n\tau}\right).
\nonumber
\end{equation} As this is an equilibrium expression for utility, the envelope theorem can again be applied and hence
\begin{equation}
\frac{\partial U}{\partial \widehat{\tau}_k} = \lambda p_{\widehat{\tau}} \sum_f \sum_{\tau \in T: O(\tau )=f} \widehat{x}_{\widehat{\tau} \tau}.\nonumber
\end{equation}
Substituting in \ref{eq:presevered_output} we get
\begin{equation}
\frac{\partial U}{\partial \widehat{\tau}_k} = \lambda p_{\widehat{\tau}} y_{\widehat{\tau}}.\nonumber
\end{equation}
Thus, $$
\frac{\partial \log(U)}{\partial \log(\tau_f)} =  \frac{\partial [ \log(GDP)+\log (\lambda)]}{\partial \log(\widehat{\tau}_k)} = \frac{\partial  \log(GDP)}{\partial \log(\widehat{\tau}_k)}
 = \frac{p_{\widehat{\tau}} y_{\widehat{\tau}}}{GDP},
$$
which is the claimed expression.\eproof

\subsection{Proof of Proposition \ref{prop:shock_prop_alg}}

\begin{proof}
We begin by showing the first half of the proposition: that the output of the algorithm, $\omega^{*}$, is unique and that $\omega^{*}_{ij} \leq \omega_{ij}$ $\forall i,j \in \mathcal{N}$. Moreover, we show that for any path $\mathcal{P}$ starting at an affected node $i$ we have that $\omega^{*}_{jk} < \omega_{jk}$ $\forall j,k \in \mathcal{P}$.

Start with the original equilibrium flow network $\omega$, and reduce the value of the outlinks for the shocked nodes to $\lambda$ their initial level. Let $\Phi$ be the vector consisting of all non-zero link-weights in this network. Now, consider the space  $S=\prod_{j = 1, \ldots , |\Phi|} [0,\phi_{j}]$. Define the partial ordering on this space such that $s\succeq \widehat{s}$ for $s,\widehat{s} \in S$ if it is weakly greater entry by entry (i.e., $s\succeq \widehat{s}$ if and only if $s_i \geq \widehat{s}_i$ for all $i$). Note that $(S,\succeq)$ is a complete lattice.

We represent the flow along each of the links in each iteration of the algorithm by some $\omega \in S$.  It is clear that each iteration of the algorithm provides a continuous mapping $\Gamma:S\rightarrow S$ with $\Gamma(\omega)\preceq \omega$. This implies that $\Gamma(\cdot)$ is an isotone function with respect to $(S,\succeq)$ and hence, by the Knaster-Tarski theorem, the set of fixed points of $\Gamma(\cdot)$ is a complete lattice under $\succeq$. There is thus a largest fixed point, with respect to the partial ordering $\succeq$, which we denote by $\omega^*$.

As the space $S$ is compact, and in each iteration $\sum_i \omega_i$ weakly decreases, the algorithm converges by the monotone convergence theorem.
Thus, the limit of the algorithm, $\omega^{*}$, is well-defined and unique and $\omega^{*}_{ij} \leq \omega_{ij}$ $\forall i,j \in \mathcal{N}$.

The fact that $\omega^{*}$ is a solution to the minimum disruption problem is argued as follows. We have shown that the shock propagation algorithm converges to a fixed point of $\Gamma$. By the construction of $\Gamma$ any fixed point of it must satisfy the constraints in the minimum disruption problem (if there is any violation then the algorithm keeps iterating). Moreover, any flows $\omega$ that are not a fixed point of $\Gamma$ must violate one of the constraints of the minimum disruption problem---otherwise the algorithm would terminate. This implies that a solution to the minimum disruption problem must be a fixed point of $\Gamma$, and specifically the fixed point that maximizes the value of final goods produced. So we just need to show that the algorithm converges to this fixed point of $\Gamma$. To show maximality with respect to network structure notice that each reduction in the flow imposed by the algorithm is inevitable, so the flows $\omega^{*}$ represent the maximum that can be produced final good by final good after the disruption.

Finally, note that for each country-industry node that ends up with a strictly lower value of a flow on an in-link relative to the initial equilibrium (i.e., $\omega^*_i<\omega_i$), its output, and hence the flow on each of its out-links, must be strictly lower than in ${\mathcal{G}}$. This proves that for any path $\mathcal{P}$ starting at an affected node $i$ we have that $\omega^{*}_{jk} < \omega_{jk}$ $\forall j,k \in \mathcal{P}$.
\end{proof}

\subsection{Proof of Proposition \ref{prop:shocked_GDP}}


\begin{proof}The upper bound follows from the algorithm directly, noting that if none of a firm's suppliers have production levels below $\lambda$ of their initial levels, then, given the structure of the production functions, the firm's output is ever below $\lambda$ of its initial level. Thus, given that all goods start with at least $\lambda$ of their initial levels, any limit point of the algorithm has the production of each good at $\lambda$ of its initial level or higher.
\end{proof}

\subsection{Proof of Proposition \ref{prop:cut}}

\begin{proof}
Suppose that the edges adjacent to $T^{Shocked}$ constitute a $( R(T^{Shocked}), F(T^{Shocked}))$-cut.
One way in which we can assign the equilibrium flows (before the shock) to supply chains is by assigning them to supply chains in which all input goods are single-sourced and where, as there are no cycles in the disrupted industries subnetwork, there are no cycles in each such supply chain. We then observe that (i) for a set of nodes $T^{Shocked}$ to be an $( R(T^{Shocked}), F(T^{Shocked}))$-cut every path in the network leading to every final good $k\in  F(T^{Shocked})$ must contain at least one node in $T^{Shocked}$; and (ii), that for any path that passes through a node in $T^{Shocked}$, output is reduced by $(1-\lambda)$.  Part (ii) holds because along each path all goods are single-sourced---thus the output of each node downstream of the shocked node must have its output reduced by exactly $(1-\lambda)$ of its initial level. Parts (i) and (ii) together imply that the output of paths in the network leading to all final goods in $ F(T^{Shocked})$ are reduced by a proportion $1-\lambda$ of their initial level, and the upper bound is obtained.
\end{proof}

\subsection{Proof of Proposition \ref{prop:disruption_frontier}}

\begin{proof}
Before presenting the proof, let us mention the logic.
Given that no technological flow is sent to 0, and all of $j$'s proportional rationings are done relative to the original equilibrium flows (regardless of the sequence of disruptions), it follows that the marginal impact of a disruption consistent with disruption of technology $\tau_i\in T_i$, does not depend on what other disruptions are conducted before or after it.   
Thus, the total value of a sequence of disruptions can be written as a weighted sum of the values of the powers of marginal disruptions.  It follows that the total value can be no more than the maximum power across the marginal disruptions, which comes from a single $\tau_i\in T_i$.

More formally, consider a disruption $(x_{\tau\tau'},y_{\tau},L_{\tau})_{\tau,\tau'}\in \mathcal{D}_i$ that is obtained via a sequence of $K$ individual technology disruptions by Country $i$ given an initial equilibrium $(x^*_{\tau\tau'},y^*_{\tau},L^*_{\tau})_{\tau,\tau'}$. 
Let $(x^{\tilde K}_{\tau\tau'},y^{\tilde K}_{\tau},L^{\tilde K}_{\tau})_{\tau,\tau'}$ be the outcome reached after ${\tilde K}\leq K$ of the individual technology shocks. Thus $(x^{\tilde K}_{\tau\tau'},y^{\tilde K}_{\tau},L^{\tilde K}_{\tau})_{\tau,\tau'}=(x_{\tau\tau'},y_{\tau},L_{\tau})_{\tau,\tau'}$. As all such outcomes are consistent (by construction) $$y^{\tilde K}_{\tau}=\sum_{\tau'} x^{\tilde K}_{\tau\tau'}=\sum_{\tau':O(\tau')=g} \frac{ x^{\tilde K}_{\tau'\tau}}{\tau_g},$$ for all $g\in I(\tau)$, all $\tau$ and all $\tilde K\leq K$. 

Thus, if we let $d^{k}_{\tau}$ be the output reduction of technology $\tau$, and let $d^{k}_{\tau\tau'}$ be the flow reduction from technology $\tau$ to $\tau'$, caused by the $k$th individual technology disruption then \begin{equation}y^{*}_{\tau}-\sum_{k=1}^{\tilde K} d^{k}_{\tau}=\sum_{\tau'} (x^*_{\tau\tau'}- \sum_{k=1}^{\tilde K} d^{k}_{\tau\tau'})=\sum_{\tau':O(\tau')=g} \frac{1}{\tau_g}\bigg(x^*_{\tau'\tau}-\sum_{k=1}^{\tilde K} d^{k}_{\tau'\tau}\bigg),\label{eq:disruption_sequence_consistency}\end{equation} for all $g\in I(\tau)$, all $\tau$ and all $\tilde K\leq K$.

This implies that for $\tilde K=1$ the disruption $(x^{1}_{\tau\tau'},y^{1}_{\tau},L^{1}_{\tau})_{\tau,\tau'}$ is a consistent individual technology disruption (as it must be because in the first step we start from the equilibrium outcome). Now consider $\tilde K=k>1$. As equation (\ref{eq:disruption_sequence_consistency}) holds for $\tilde K-1$, this implies that $$-d^{k}_{\tau}=-\sum_{\tau'} d^{k}_{\tau\tau'}=-\sum_{\tau':O(\tau')=g} \frac{d^{k}_{\tau'\tau}}{\tau_g},$$ for all $g\in I(\tau)$ and all $\tau$, so using consistency of the initial equilibrium $$y^*_{\tau}-d^{k}_{\tau}=\sum_{\tau'}x^*_{\tau\tau'}- d^{k}_{\tau\tau'}=\sum_{\tau':O(\tau')=g} \frac{1}{\tau_g}\left(x^*_{\tau\tau'}-d^{k}_{\tau'\tau}\right),$$ for all $g\in I(\tau)$ and all $\tau$. Thus each disruption in the sequence is by itself a consistent individual technology disruption with respect to the initial equilibrium.

As by assumption $x^*_{\tau\tau'}>0$ implies the final flow associated with the disruption $x^K_{\tau\tau'}>0$ and, absent any zero-outed flows, rationing that occurs at all of the $K$ steps in the sequence of disruptions is done relative to the initial equilibrium flows, the $k$th disruption is a valid and available individual technology disruption for all $k$, for the initial equilibrium.

Without loss of generality label these valid individual technology disruptions $\zeta_1,\dots, \zeta_K$. Each such consistent disruption in this set is associated with a change in employed labor in country $i$ and in country $j$. Let $\Delta L_i(\zeta_1)$ denote the change in employed labor in country $i$ due to disruption $\zeta_1$. The power of the disruption for country $i$ is then a strictly increasing function of $\sum_{k=1}^{K} \Delta L_j(\zeta_k)/\sum_{k=1}^{K} \Delta L_i(\zeta_k).$ Moreover, $$\frac{\sum_{k=1}^{K} \Delta L_j(\zeta_k)}{\sum_{k=1}^{K} \Delta L_i(\zeta_k)}\leq \max_{k}\frac{\Delta L_j(\zeta_k)}{\Delta L_i(\zeta_k)},$$ and so there exists a valid individual technology disruption to a technology in $T_i$ that is weakly more powerful.
\end{proof}


\section{Extended Example}\label{sec:extended_eg}

In this section we develop a richer example to better illustrate some of the key ideas. We begin with technological dependencies. Figure \ref{fig:tech} provides an example in which the technological dependencies exhibit a cycle.

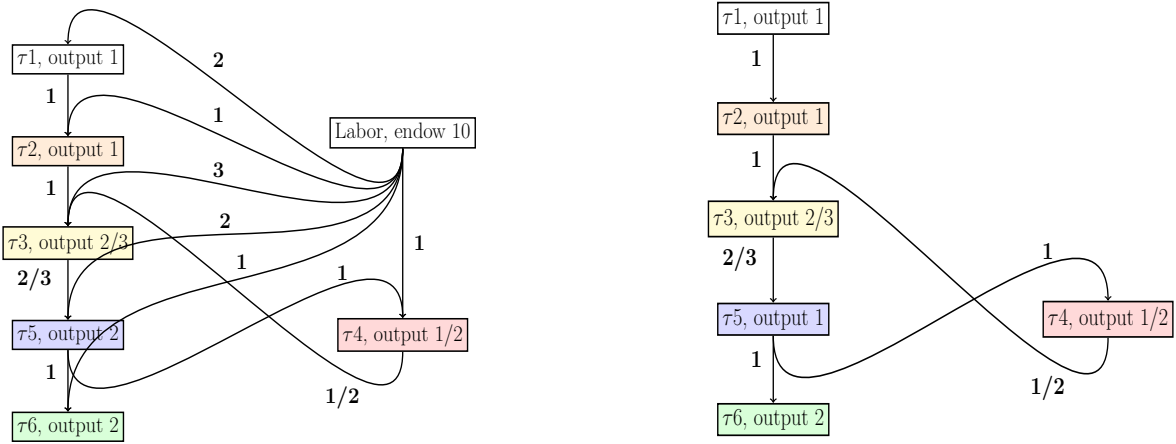
\begin{figure}[!ht]
    \subfloat[Technological dependencies\label{subfig:Tech}]
    {
\centering
\resizebox{2.9in}{3.2in}{		
\begin{tikzpicture}[
		sqRED/.style={rectangle, draw=black!240, fill=red!15, very thick, minimum size=7mm},
		sqBLUE/.style={rectangle, draw=black!240, fill=blue!15, very thick, minimum size=7mm},
		roundnode/.style={ellipse, draw=black!240, fill=green!15, very thick, dashed, minimum size=7mm},
		sqGREEN/.style={rectangle, draw=black!240, fill=green!15, very thick, minimum size=7mm},
		roundRED/.style={ellipse, draw=black!240, fill=red!15, very thick, dashed, minimum size=7mm},
		roundREDb/.style={ellipse, draw=black!240, fill=red!15, very thick, minimum size=7mm},
		roundYELL/.style={ellipse, draw=black!240, fill=yellow!20, very thick, dashed, minimum size=7mm},
		sqYELL/.style={rectangle, draw=black!240, fill=yellow!20, very thick, minimum size=7mm},
		roundORA/.style={ellipse, draw=black!240, fill=orange!15, very thick, dashed, minimum size=7mm},		
		sqORA/.style={rectangle, draw=black!240, fill=orange!15, very thick, minimum size=7mm},
		squaredBLACK/.style={rectangle, draw=black!240, fill=white!15, very thick, minimum size=7mm},	
		roundBROWN/.style={ellipse, draw=black!240, fill=brown!15, very thick, dashed, minimum size=7mm},
		sqBROWN/.style={rectangle, draw=black!240, fill=brown!15, very thick, minimum size=7mm},		
        sqPURPLE/.style={rectangle, draw=black!240, fill=purple!15, very thick, minimum size=7mm},				
		]
		\node[](S1) at (-3, -7) {};
        \node[](S2) at (3, 7) {};
        \node[](S3) at (13, 7) {};
        \node[](S4) at (13, -7) {};

		\node[squaredBLACK](R1) at (0, 5) {\Large $\tau1$, output 1};
        \node[] at (4.5, 5) {\Large $\bm{2}$};
        \node[] at (4.5, 3.5) {\Large $\bm{1}$};
        \node[] at (4.5, 2) {\Large $\bm{3}$};
        \node[] at (4.7, 0.6) {\Large $\bm{2}$};
        \node[] at (5.2, -.5) {\Large $\bm{1}$};

        \node[] at (-0.5, 4) {\Large $\bm{1}$};

		\node[sqORA, align=left](A1) at (0, 2.5) {\Large $\tau2$, output 1};
        \node[] at (-0.5, 1.5) {\Large $\bm{1}$};

		\node[sqYELL, align=left](C1) at (0, 0) {\Large $\tau3$, output 2/3};
        \node[] at (-1, -1) {\Large $\bm{2/3}$};

        \node[] at (8.2, -0.8) {\Large $\bm{1}$};

        \node[sqRED, align=left](D1) at (10, -2.5) {\Large $\tau4$, output 1/2};
        \node[] at (8.2, -4.2) {\Large $\bm{1/2}$};

		\node[sqBLUE, align=left](E1) at (0, -2.5) {\Large $\tau5$, output 2};
        \node[] at (-0.5, -3.5) {\Large $\bm{1}$};

		\node[sqGREEN, align=left](F1) at (0, -5) {\Large $\tau6$, output 2};
			
		
		\node[squaredBLACK, align=left](L) at (10, 3) {\Large Labor, endow 10};
		\node[] at (10.5, 0) {\Large $\bm{1}$};

		\tikzset{thick edge/.style={-, black, fill=none, thick, text=black}}
		\tikzset{thick arc/.style={->, black, fill=black, thick, >=stealth, text=black}}

		\draw[line width=0.4mm, black,-> ] (R1) to[out=-90,in=90] (A1);
		
        \draw[line width=0.4mm, black,-> ] (L) to[out=-90,in=90] (R1);
        \draw[line width=0.4mm, black,-> ] (L) to[out=-90,in=90] (A1);
        \draw[line width=0.4mm, black,-> ] (L) to[out=-90,in=90] (C1);
        \draw[line width=0.4mm, black,-> ] (L) to[out=-90,in=90] (D1);
        \draw[line width=0.4mm, black,-> ] (L) to[out=-90,in=90] (E1);
        \draw[line width=0.4mm, black,-> ] (L) to[out=-90,in=90] (F1);		
		
        \draw[line width=0.4mm, black,-> ] (A1) to[out=-90,in=90] (C1);
				
		\draw[line width=0.4mm, black,-> ] (D1) to[out=-90,in=90] (C1);

        \draw[line width=0.4mm, black,-> ] (E1) to[out=-90,in=90] (D1);
			
		\draw[line width=0.4mm, black,-> ] (C1) to[out=-90,in=90] (E1);

		\draw[line width=0.4mm, black,-> ] (E1) to[out=-90,in=90] (F1);
		

		\end{tikzpicture}
		}
	}
    \hfill
    \subfloat[Technological dependencies without labor pictured\label{subfig:tech_no_labor}]
    {
     \centering
\resizebox{2.9in}{3in}{		
\begin{tikzpicture}[
		sqRED/.style={rectangle, draw=black!240, fill=red!15, very thick, minimum size=7mm},
		sqBLUE/.style={rectangle, draw=black!240, fill=blue!15, very thick, minimum size=7mm},
		roundnode/.style={ellipse, draw=black!240, fill=green!15, very thick, dashed, minimum size=7mm},
		sqGREEN/.style={rectangle, draw=black!240, fill=green!15, very thick, minimum size=7mm},
		roundRED/.style={ellipse, draw=black!240, fill=red!15, very thick, dashed, minimum size=7mm},
		roundREDb/.style={ellipse, draw=black!240, fill=red!15, very thick, minimum size=7mm},
		roundYELL/.style={ellipse, draw=black!240, fill=yellow!20, very thick, dashed, minimum size=7mm},
		sqYELL/.style={rectangle, draw=black!240, fill=yellow!20, very thick, minimum size=7mm},
		roundORA/.style={ellipse, draw=black!240, fill=orange!15, very thick, dashed, minimum size=7mm},		
		sqORA/.style={rectangle, draw=black!240, fill=orange!15, very thick, minimum size=7mm},
		squaredBLACK/.style={rectangle, draw=black!240, fill=white!15, very thick, minimum size=7mm},	
		roundBROWN/.style={ellipse, draw=black!240, fill=brown!15, very thick, dashed, minimum size=7mm},
		sqBROWN/.style={rectangle, draw=black!240, fill=brown!15, very thick, minimum size=7mm},		
        sqPURPLE/.style={rectangle, draw=black!240, fill=purple!15, very thick, minimum size=7mm},				
		]
		\node[](S1) at (-3, -7) {};
        \node[](S2) at (2, 7) {};
        \node[](S3) at (13, 7) {};
        \node[](S4) at (13, -7) {};

		\node[squaredBLACK](R1) at (0, 5) {\Large $\tau1$, output 1};

        \node[] at (-0.5, 4) {\Large $\bm{1}$};

		\node[sqORA, align=left](A1) at (0, 2.5) {\Large $\tau2$, output 1};
        \node[] at (-0.5, 1.5) {\Large $\bm{1}$};

		\node[sqYELL, align=left](C1) at (0, 0) {\Large $\tau3$, output 2/3};
        \node[] at (-1, -1) {\Large $\bm{2/3}$};

        \node[] at (8.2, -0.8) {\Large $\bm{1}$};

        \node[sqRED, align=left](D1) at (10, -2.5) {\Large $\tau4$, output 1/2};
        \node[] at (8.2, -4.2) {\Large $\bm{1/2}$};

		\node[sqBLUE, align=left](E1) at (0, -2.5) {\Large $\tau5$, output 1};
        \node[] at (-0.5, -3.5) {\Large $\bm{1}$};

		\node[sqGREEN, align=left](F1) at (0, -5) {\Large $\tau6$, output 2};
			
		
		
		\tikzset{thick edge/.style={-, black, fill=none, thick, text=black}}
		\tikzset{thick arc/.style={->, black, fill=black, thick, >=stealth, text=black}}

		\draw[line width=0.4mm, black,-> ] (R1) to[out=-90,in=90] (A1);

        \draw[line width=0.4mm, black,-> ] (A1) to[out=-90,in=90] (C1);
				
		\draw[line width=0.4mm, black,-> ] (D1) to[out=-90,in=90] (C1);

        \draw[line width=0.4mm, black,-> ] (E1) to[out=-90,in=90] (D1);
			
		\draw[line width=0.4mm, black,-> ] (C1) to[out=-90,in=90] (E1);

		\draw[line width=0.4mm, black,-> ] (E1) to[out=-90,in=90] (F1);
		

		\end{tikzpicture}
		}
    }
    \caption{\footnotesize{An example of technological interdependencies and good flows: The weight on a directed link from an input into a technology represents the number of units of that input shipped. It is convenient to occasionally omit the dependence on labor to simplify the figures in what follows.  The final good technology is $\tau6$}}
    \label{fig:tech}
  \end{figure}
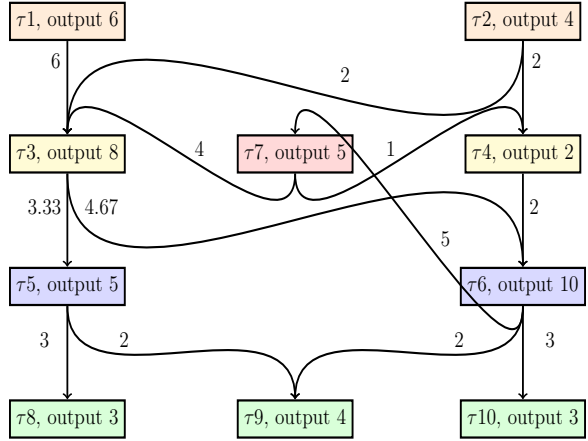
\begin{figure}[ht!]
{
\centering
\resizebox{4.9in}{2.7in}{	
\begin{tikzpicture}[
		sqRED/.style={rectangle, draw=black!240, fill=red!15, very thick, minimum size=7mm},
		sqBLUE/.style={rectangle, draw=black!240, fill=blue!15, very thick, minimum size=7mm},
		sqBLUEb/.style={rectangle, draw=black!240, fill=blue!15, very thick, dashed, minimum size=7mm},
        roundnode/.style={ellipse, draw=black!240, fill=green!15, very thick, dashed, minimum size=7mm},
		sqGREEN/.style={rectangle, draw=black!240, fill=green!15, very thick, minimum size=7mm},
		sqGREENb/.style={rectangle, draw=black!240, fill=green!15, very thick, dashed, minimum size=7mm},
        roundRED/.style={ellipse, draw=black!240, fill=red!15, very thick, dashed, minimum size=7mm},
		roundREDb/.style={ellipse, draw=black!240, fill=red!15, very thick, minimum size=7mm},
		roundYELL/.style={ellipse, draw=black!240, fill=yellow!20, very thick, dashed, minimum size=7mm},
		sqYELL/.style={rectangle, draw=black!240, fill=yellow!20, very thick, minimum size=7mm},
        sqYELLb/.style={rectangle, draw=black!240, fill=yellow!20, very thick, dashed, minimum size=7mm},
		roundORA/.style={ellipse, draw=black!240, fill=orange!15, very thick, dashed, minimum size=7mm},		
		sqORA/.style={rectangle, draw=black!240, fill=orange!15, very thick, minimum size=7mm},
        sqORAb/.style={rectangle, draw=black!240, fill=orange!15, very thick, dashed, minimum size=7mm},		
        squaredBLACK/.style={rectangle, draw=black!240, fill=white!15, very thick, minimum size=7mm},
        squaredBLACK2/.style={rectangle, draw=red!240, fill=white!15, very thick, dashed, minimum size=7mm},			
		roundBROWN/.style={ellipse, draw=black!240, fill=brown!15, very thick, dashed, minimum size=7mm},
		sqBROWN/.style={rectangle, draw=black!240, fill=brown!15, very thick, minimum size=7mm},		
        sqBROWNb/.style={rectangle, draw=black!240, fill=brown!15, very thick, dashed, minimum size=7mm},	
        sqPURPLE/.style={rectangle, draw=black!240, fill=purple!15, very thick, minimum size=7mm},				
		]
        \node[] at (8, 0) {};
        \node[] at (-12, 0) {};
		

		\node[sqORA](A1) at (-5, 2.5) { $\tau 1$, output  $6$};
        \node[] at (-5.25, 1.75) { $6$};

        \node[sqORA](A2) at (5, 2.5) { $\tau 2$, output $4$};
        \node[] at (1.1, 1.5) { $2$};
		\node[] at (5.3, 1.75) { $2$};
		
		\node[sqYELL](C1) at (-5, 0) { $\tau 3$, output $8$};
        \node[] at (-5.5, -1) { $3.33$};
		\node[] at (-4.25, -1) { $4.67$};

        \node[sqYELL](C2) at (5, 0) { $\tau 4$, output $2$};
		\node[] at (5.25, -1) { $2$};
				
		\node[sqRED](D1) at (0, 0) { $\tau 7$, output $5$};
        \node[] at (2.1, 0.1) { $1$};
		\node[] at (-2.1, 0.1) { $4$};

		\node[sqBLUE](E1) at (-5, -2.5) { $\tau 5$, output $5$};
        \node[] at (-5.5, -3.5) { $3$};
		\node[] at (-3.75, -3.5) { $2$};

        \node[sqBLUE](E2) at (5, -2.5) { $\tau 6$, output $10$};
        \node[] at (5.6, -3.5) { $3$};
		\node[] at (3.6, -3.5) { $2$};
        \node[] at (3.3, -1.6) { $5$};

		\node[sqGREEN](F1) at (-5, -5) {$\tau 8$, output $3$};
		\node[sqGREEN](F2) at (0, -5) {$\tau 9$, output $4$};
		\node[sqGREEN](F3) at (5, -5) {$\tau 10$, output $3$};
%
		
		\tikzset{thick edge/.style={-, black, fill=none, thick, text=black}}
		\tikzset{thick arc/.style={->, black, fill=black, thick, >=stealth, text=black}}
		

		
		
		\draw[line width=0.4mm, black,-> ] (A1) to[out=-90,in=90] (C1);
		\draw[line width=0.4mm, black,-> ] (A2) to[out=-90,in=90] (C1);
        \draw[line width=0.4mm, black,-> ] (A2) to[out=-90,in=90] (C2);

		\draw[line width=0.4mm, black,-> ] (D1) to[out=-90,in=90] (C1);
        \draw[line width=0.4mm, black,-> ] (D1) to[out=-90,in=90] (C2);

		\draw[line width=0.4mm, black,-> ] (C1) to[out=-90,in=90] (E1);
		\draw[line width=0.4mm, black,-> ] (C1) to[out=-90,in=90] (E2);
		\draw[line width=0.4mm, black,-> ] (C2) to[out=-90,in=90] (E2);

		
		\draw[line width=0.4mm, black,-> ] (E2) to[out=-90,in=90] (D1);

		\draw[line width=0.4mm, black,-> ] (E1) to[out=-90,in=90] (F1);
		\draw[line width=0.4mm, black,-> ] (E1) to[out=-90,in=90] (F2);
		\draw[line width=0.4mm, black,-> ] (E2) to[out=-90,in=90] (F3);
        \draw[line width=0.4mm, black,-> ] (E2) to[out=-90,in=90] (F2);

		\end{tikzpicture}
}}	
\caption{\footnotesize{An example of equilibrium flows between technologies / producers, labor not pictured. The weights on the directed edges correspond to the size of flows between technologies. Nodes of the same color produce the same good. Green technologies produce final goods.
\label{eqfig}}}
\end{figure}

A similar network can illustrate the equilibrium flow of goods between technologies. This is shown in Figure \ref{eqfig}, where multiple technologies are sometimes used to produce the same good, and in this example they do so using the same combinations of inputs.\footnote{ This case is non-generic in so far as as different countries will not have exactly the same technologies, with identical productivities, for producing a given good. However, more complicated examples with iceberg costs and heterogeneous technologies work in much the same way.}

\begin{figure}[ht!]
  \subfloat[10 percent shock to $\tau2$]
{
\centering
\resizebox{2.7in}{2.16in}{		
		\begin{tikzpicture}[
		sqRED/.style={rectangle, draw=black!240, fill=red!15, very thick, minimum size=7mm},
		sqBLUE/.style={rectangle, draw=black!240, fill=blue!15, very thick, minimum size=7mm},
		sqBLUEb/.style={rectangle, draw=black!240, fill=blue!15, very thick, dashed, minimum size=7mm},
        roundnode/.style={ellipse, draw=black!240, fill=green!15, very thick, dashed, minimum size=7mm},
		sqGREEN/.style={rectangle, draw=black!240, fill=green!15, very thick, minimum size=7mm},
		sqGREENb/.style={rectangle, draw=black!240, fill=green!15, very thick, dashed, minimum size=7mm},
        roundRED/.style={ellipse, draw=black!240, fill=red!15, very thick, dashed, minimum size=7mm},
		roundREDb/.style={ellipse, draw=black!240, fill=red!15, very thick, minimum size=7mm},
		roundYELL/.style={ellipse, draw=black!240, fill=yellow!20, very thick, dashed, minimum size=7mm},
		sqYELL/.style={rectangle, draw=black!240, fill=yellow!20, very thick, minimum size=7mm},
        sqYELLb/.style={rectangle, draw=black!240, fill=yellow!20, very thick, dashed, minimum size=7mm},
		roundORA/.style={ellipse, draw=black!240, fill=orange!15, very thick, dashed, minimum size=7mm},		
		sqORA/.style={rectangle, draw=black!240, fill=orange!15, very thick, minimum size=7mm},
        sqORAb/.style={rectangle, draw=red!240, fill=orange!15, very thick, dashed, minimum size=7mm},		
        squaredBLACK/.style={rectangle, draw=black!240, fill=white!15, very thick, minimum size=7mm},	
        squaredBLACK2/.style={rectangle, draw=red!240, fill=white!15, very thick, dashed, minimum size=7mm},		
		roundBROWN/.style={ellipse, draw=black!240, fill=brown!15, very thick, dashed, minimum size=7mm},
		sqBROWN/.style={rectangle, draw=black!240, fill=brown!15, very thick, minimum size=7mm},		
        sqBROWNb/.style={rectangle, draw=black!240, fill=brown!15, very thick, dashed, minimum size=7mm},	
        sqPURPLE/.style={rectangle, draw=black!240, fill=purple!15, very thick, minimum size=7mm},				
		]
        \node[] at (5, 0) {};
        \node[] at (-5, 0) {};
		

		\node[sqORA](A1) at (-4, 2.5) { $\tau 1$, $y=6$};
        \node[] at (-4.25, 1.75) { $6$};

        \node[sqORAb](A2) at (4, 2.5) {$\tau 2$, $y=\cancel{4}\;\textcolor{red}{3.6}$};
        \node[] at (0.1, 1.5) {$\cancel{2}$ \textcolor{red}{$1.8$}};
		\node[] at (4.45, 1.75) {$\cancel{2}$ \textcolor{red}{$1.8$}};

		\node[sqYELL](C1) at (-4, 0) { $\tau 3$, $y=8$};
        \node[] at (-4.5, -1) { $3.33$};
		\node[] at (-3.25, -1) { $4.67$};

        \node[sqYELL](C2) at (4, 0) { $\tau 4$, $y=2$};
		\node[] at (4.25, -1) { $2$};
				
		\node[sqRED](D1) at (0, 0) { $\tau 7$, $y=5$};
        \node[] at (1.7, 0.3) { $1$};
		\node[] at (-1.7, 0.3) { $4$};

		\node[sqBLUE](E1) at (-4, -2.5) { $\tau 5$, $y=5$};
        \node[] at (-4.5, -3.5) { $3$};
		\node[] at (-2.75, -3.5) { $2$};

        \node[sqBLUE](E2) at (4, -2.5) { $\tau 6$, $y=10$};
        \node[] at (4.6, -3.5) { $3$};
		\node[] at (2.6, -3.5) { $2$};
        \node[] at (1.6, -1.6) { $5$};

		\node[sqGREEN](F1) at (-4, -5) { $\tau 8$, $y=3$};
		\node[sqGREEN](F2) at (0, -5) { $\tau 9$, $y=4$};
		\node[sqGREEN](F3) at (4, -5) { $\tau 10$, $y=3$};

		\tikzset{thick edge/.style={-, black, fill=none, thick, text=black}}
		\tikzset{thick arc/.style={->, black, fill=black, thick, >=stealth, text=black}}

		\draw[line width=0.4mm, black,-> ] (A1) to[out=-90,in=90] (C1);
		\draw[line width=0.4mm, black,-> ] (A2) to[out=-90,in=90] (C1);
        \draw[line width=0.4mm, black,-> ] (A2) to[out=-90,in=90] (C2);

		\draw[line width=0.4mm, black,-> ] (D1) to[out=-90,in=90] (C1);
        \draw[line width=0.4mm, black,-> ] (D1) to[out=-90,in=90] (C2);

		\draw[line width=0.4mm, black,-> ] (C1) to[out=-90,in=90] (E1);
		\draw[line width=0.4mm, black,-> ] (C1) to[out=-90,in=90] (E2);
		\draw[line width=0.4mm, black,-> ] (C2) to[out=-90,in=90] (E2);

		
		\draw[line width=0.4mm, black,-> ] (E2) to[out=-90,in=90] (D1);

		\draw[line width=0.4mm, black,-> ] (E1) to[out=-90,in=90] (F1);
		\draw[line width=0.4mm, black,-> ] (E1) to[out=-90,in=90] (F2);
		\draw[line width=0.4mm, black,-> ] (E2) to[out=-90,in=90] (F3);
        \draw[line width=0.4mm, black,-> ] (E2) to[out=-90,in=90] (F2);

		\end{tikzpicture}
}
}
\hfill
  \subfloat[Shock propagates to $\tau3,\tau4$]
{
\centering
\resizebox{2.7in}{2.16in}{			
		\begin{tikzpicture}[
		sqRED/.style={rectangle, draw=black!240, fill=red!15, very thick, minimum size=7mm},
		sqBLUE/.style={rectangle, draw=black!240, fill=blue!15, very thick, minimum size=7mm},
		sqBLUEb/.style={rectangle, draw=black!240, fill=blue!15, very thick, dashed, minimum size=7mm},
        roundnode/.style={ellipse, draw=black!240, fill=green!15, very thick, dashed, minimum size=7mm},
		sqGREEN/.style={rectangle, draw=black!240, fill=green!15, very thick, minimum size=7mm},
		sqGREENb/.style={rectangle, draw=black!240, fill=green!15, very thick, dashed, minimum size=7mm},
        roundRED/.style={ellipse, draw=black!240, fill=red!15, very thick, dashed, minimum size=7mm},
		roundREDb/.style={ellipse, draw=black!240, fill=red!15, very thick, minimum size=7mm},
		roundYELL/.style={ellipse, draw=black!240, fill=yellow!20, very thick, dashed, minimum size=7mm},
		sqYELL/.style={rectangle, draw=black!240, fill=yellow!20, very thick, minimum size=7mm},
        sqYELLb/.style={rectangle, draw=red!240, fill=yellow!20, very thick, dashed, minimum size=7mm},
		roundORA/.style={ellipse, draw=black!240, fill=orange!15, very thick, dashed, minimum size=7mm},		
		sqORA/.style={rectangle, draw=black!240, fill=orange!15, very thick, minimum size=7mm},
        sqORAb/.style={rectangle, draw=black!240, fill=orange!15, very thick, dashed, minimum size=7mm},		
        squaredBLACK/.style={rectangle, draw=black!240, fill=white!15, very thick, minimum size=7mm},
        squaredBLACK2/.style={rectangle, draw=red!240, fill=white!15, very thick, dashed, minimum size=7mm},			
		roundBROWN/.style={ellipse, draw=black!240, fill=brown!15, very thick, dashed, minimum size=7mm},
		sqBROWN/.style={rectangle, draw=black!240, fill=brown!15, very thick, minimum size=7mm},		
        sqBROWNb/.style={rectangle, draw=black!240, fill=brown!15, very thick, dashed, minimum size=7mm},	
        sqPURPLE/.style={rectangle, draw=black!240, fill=purple!15, very thick, minimum size=7mm},				
		]
        \node[] at (5, 0) {};
        \node[] at (-5, 0) {};
		

		\node[sqORA](A1) at (-4, 2.5) { $\tau 1$, $y=6$};
        \node[] at (-4.25, 1.75) { $6$};

        \node[sqORA](A2) at (4, 2.5) { $\tau 2$, $y=3.6$};
        \node[] at (0.1, 1.5) { $1.8$};
		\node[] at (4.3, 1.75) { $1.8$};

		\node[sqYELLb](C1) at (-4, 0) {$\tau 3$, $y=\cancel{8}\;\textcolor{red}{7.8}$};
        \node[] at (-5.1, -1) {\cancel{$3.33$}\;\textcolor{red}{$3.25$}};
		\node[] at (-2.75, -1) {\cancel{$4.67$}\;\textcolor{red}{$4.55$}};

        \node[sqYELLb](C2) at (4, 0) {$\tau 4$, $y=\cancel{2}\;\textcolor{red}{1.8}$};
		\node[] at (4.45, -1) {\cancel{$2$}\;\textcolor{red}{$1.8$} };
				
		\node[sqRED](D1) at (0, 0) { $\tau 7$, $y=5$};
        \node[] at (1.7, 0.3) { $1$};
		\node[] at (-1.7, 0.3) { $4$};

		\node[sqBLUE](E1) at (-4, -2.5) { $\tau 5$, $y=5$};
        \node[] at (-4.5, -3.5) { $3$};
		\node[] at (-2.75, -3.5) { $2$};

        \node[sqBLUE](E2) at (4, -2.5) { $\tau 6$, $y=10$};
        \node[] at (4.6, -3.5) { $3$};
		\node[] at (2.6, -3.5) { $2$};
        \node[] at (1.6, -1.6) { $5$};

		\node[sqGREEN](F1) at (-4, -5) { $\tau 8$, $y=3$};
		\node[sqGREEN](F2) at (0, -5) { $\tau 9$, $y=4$};
		\node[sqGREEN](F3) at (4, -5) { $\tau 10$, $y=3$};

		\tikzset{thick edge/.style={-, black, fill=none, thick, text=black}}
		\tikzset{thick arc/.style={->, black, fill=black, thick, >=stealth, text=black}}

		\draw[line width=0.4mm, black,-> ] (A1) to[out=-90,in=90] (C1);
		\draw[line width=0.4mm, black,-> ] (A2) to[out=-90,in=90] (C1);
        \draw[line width=0.4mm, black,-> ] (A2) to[out=-90,in=90] (C2);

		\draw[line width=0.4mm, black,-> ] (D1) to[out=-90,in=90] (C1);
        \draw[line width=0.4mm, black,-> ] (D1) to[out=-90,in=90] (C2);

		\draw[line width=0.4mm, black,-> ] (C1) to[out=-90,in=90] (E1);
		\draw[line width=0.4mm, black,-> ] (C1) to[out=-90,in=90] (E2);
		\draw[line width=0.4mm, black,-> ] (C2) to[out=-90,in=90] (E2);

		
		\draw[line width=0.4mm, black,-> ] (E2) to[out=-90,in=90] (D1);

		\draw[line width=0.4mm, black,-> ] (E1) to[out=-90,in=90] (F1);
		\draw[line width=0.4mm, black,-> ] (E1) to[out=-90,in=90] (F2);
		\draw[line width=0.4mm, black,-> ] (E2) to[out=-90,in=90] (F3);
        \draw[line width=0.4mm, black,-> ] (E2) to[out=-90,in=90] (F2);

		\end{tikzpicture}
}	
}

  \subfloat[Shock continues to $\tau5,\tau6$]
{
\centering
\resizebox{2.7in}{2.16in}{			
		\begin{tikzpicture}[
		sqRED/.style={rectangle, draw=black!240, fill=red!15, very thick, minimum size=7mm},
		sqBLUE/.style={rectangle, draw=black!240, fill=blue!15, very thick, minimum size=7mm},
		sqBLUEb/.style={rectangle, draw=red!240, fill=blue!15, very thick, dashed, minimum size=7mm},
        roundnode/.style={ellipse, draw=black!240, fill=green!15, very thick, dashed, minimum size=7mm},
		sqGREEN/.style={rectangle, draw=black!240, fill=green!15, very thick, minimum size=7mm},
		sqGREENb/.style={rectangle, draw=black!240, fill=green!15, very thick, dashed, minimum size=7mm},
        roundRED/.style={ellipse, draw=black!240, fill=red!15, very thick, dashed, minimum size=7mm},
		roundREDb/.style={ellipse, draw=black!240, fill=red!15, very thick, minimum size=7mm},
		roundYELL/.style={ellipse, draw=black!240, fill=yellow!20, very thick, dashed, minimum size=7mm},
		sqYELL/.style={rectangle, draw=black!240, fill=yellow!20, very thick, minimum size=7mm},
        sqYELLb/.style={rectangle, draw=black!240, fill=yellow!20, very thick, dashed, minimum size=7mm},
		roundORA/.style={ellipse, draw=black!240, fill=orange!15, very thick, dashed, minimum size=7mm},		
		sqORA/.style={rectangle, draw=black!240, fill=orange!15, very thick, minimum size=7mm},
        sqORAb/.style={rectangle, draw=black!240, fill=orange!15, very thick, dashed, minimum size=7mm},		
        squaredBLACK/.style={rectangle, draw=black!240, fill=white!15, very thick, minimum size=7mm},
        squaredBLACK2/.style={rectangle, draw=red!240, fill=white!15, very thick, dashed, minimum size=7mm},			
		roundBROWN/.style={ellipse, draw=black!240, fill=brown!15, very thick, dashed, minimum size=7mm},
		sqBROWN/.style={rectangle, draw=black!240, fill=brown!15, very thick, minimum size=7mm},		
        sqBROWNb/.style={rectangle, draw=black!240, fill=brown!15, very thick, dashed, minimum size=7mm},	
        sqPURPLE/.style={rectangle, draw=black!240, fill=purple!15, very thick, minimum size=7mm},				
		]
        \node[] at (5, 0) {};
        \node[] at (-5, 0) {};
		

		\node[sqORA](A1) at (-4, 2.5) { $\tau 1$, $y=6$};
        \node[] at (-4.25, 1.75) { $6$};

        \node[sqORA](A2) at (4, 2.5) { $\tau 2$, $y=3.6$};
        \node[] at (0.1, 1.5) { $1.8$};
		\node[] at (4.3, 1.75) { $1.8$};

		\node[sqYELL](C1) at (-4, 0) { $\tau 3$, $y=7.8$};
        \node[] at (-4.5, -1) { $3.25$};
		\node[] at (-3.25, -1) { $4.55$};

        \node[sqYELL](C2) at (4, 0) { $\tau 4$, $y=1.8$};
		\node[] at (4.3, -1) { $1.8$ };
				
		\node[sqRED](D1) at (0, 0) { $\tau 7$, $y=5$};
        \node[] at (1.7, 0.3) { $1$};
		\node[] at (-1.7, 0.3) { $4$};

		\node[sqBLUEb](E1) at (-4, -2.5) {$\tau 5$, $y=\cancel{5}\;\textcolor{red}{4.88}$};
        \node[] at (-4.75, -3.5) {$\cancel{3}\;\textcolor{red}{2.93}$};
		\node[] at (-2.5, -3.5) {$\cancel{2}\;\textcolor{red}{1.95}$};

        \node[sqBLUEb](E2) at (4, -2.5) {$\tau 6$, $y=\cancel{10}\;\textcolor{red}{9.53}$};
        \node[] at (4.8, -3.5) {$\cancel{3}\;\textcolor{red}{2.86}$};
		\node[] at (2.4, -3.5) {$\cancel{2}\;\textcolor{red}{1.91}$};
        \node[] at (1.4, -1.6) {$\cancel{5}\;\textcolor{red}{4.76}$};

		\node[sqGREEN](F1) at (-4, -5) { $\tau 8$, $y=3$};
		\node[sqGREEN](F2) at (0, -5) { $\tau 9$, $y=4$};
		\node[sqGREEN](F3) at (4, -5) { $\tau 10$, $y=3$};

		\tikzset{thick edge/.style={-, black, fill=none, thick, text=black}}
		\tikzset{thick arc/.style={->, black, fill=black, thick, >=stealth, text=black}}

		\draw[line width=0.4mm, black,-> ] (A1) to[out=-90,in=90] (C1);
		\draw[line width=0.4mm, black,-> ] (A2) to[out=-90,in=90] (C1);
        \draw[line width=0.4mm, black,-> ] (A2) to[out=-90,in=90] (C2);

		\draw[line width=0.4mm, black,-> ] (D1) to[out=-90,in=90] (C1);
        \draw[line width=0.4mm, black,-> ] (D1) to[out=-90,in=90] (C2);

		\draw[line width=0.4mm, black,-> ] (C1) to[out=-90,in=90] (E1);
		\draw[line width=0.4mm, black,-> ] (C1) to[out=-90,in=90] (E2);
		\draw[line width=0.4mm, black,-> ] (C2) to[out=-90,in=90] (E2);

		
		\draw[line width=0.4mm, black,-> ] (E2) to[out=-90,in=90] (D1);

		\draw[line width=0.4mm, black,-> ] (E1) to[out=-90,in=90] (F1);
		\draw[line width=0.4mm, black,-> ] (E1) to[out=-90,in=90] (F2);
		\draw[line width=0.4mm, black,-> ] (E2) to[out=-90,in=90] (F3);
        \draw[line width=0.4mm, black,-> ] (E2) to[out=-90,in=90] (F2);

		\end{tikzpicture}
}	
}
\hfill
  \subfloat[Shock hits $\tau7,\tau8,\tau9,\tau10$]
{
\centering
\resizebox{2.7in}{2.16in}{			
		\begin{tikzpicture}[
		sqRED/.style={rectangle, draw=black!240, fill=red!15, very thick, minimum size=7mm},
		sqREDb/.style={rectangle, draw=red!240, fill=red!15, very thick, dashed, minimum size=7mm},
		sqBLUE/.style={rectangle, draw=black!240, fill=blue!15, very thick, minimum size=7mm},
		sqBLUEb/.style={rectangle, draw=black!240, fill=blue!15, very thick, dashed, minimum size=7mm},
        roundnode/.style={ellipse, draw=black!240, fill=green!15, very thick, dashed, minimum size=7mm},
		sqGREEN/.style={rectangle, draw=black!240, fill=green!15, very thick, minimum size=7mm},
		sqGREENb/.style={rectangle, draw=red!240, fill=green!15, very thick, dashed, minimum size=7mm},
        roundRED/.style={ellipse, draw=black!240, fill=red!15, very thick, dashed, minimum size=7mm},
		roundREDb/.style={ellipse, draw=red!240, fill=red!15, very thick, dashed, minimum size=7mm},
		roundYELL/.style={ellipse, draw=black!240, fill=yellow!20, very thick, dashed, minimum size=7mm},
		sqYELL/.style={rectangle, draw=black!240, fill=yellow!20, very thick, minimum size=7mm},
        sqYELLb/.style={rectangle, draw=black!240, fill=yellow!20, very thick, dashed, minimum size=7mm},
		roundORA/.style={ellipse, draw=black!240, fill=orange!15, very thick, dashed, minimum size=7mm},		
		sqORA/.style={rectangle, draw=black!240, fill=orange!15, very thick, minimum size=7mm},
        sqORAb/.style={rectangle, draw=black!240, fill=orange!15, very thick, dashed, minimum size=7mm},		
        squaredBLACK/.style={rectangle, draw=black!240, fill=white!15, very thick, minimum size=7mm},
        squaredBLACK2/.style={rectangle, draw=red!240, fill=white!15, very thick, dashed, minimum size=7mm},			
		roundBROWN/.style={ellipse, draw=black!240, fill=brown!15, very thick, dashed, minimum size=7mm},
		sqBROWN/.style={rectangle, draw=black!240, fill=brown!15, very thick, minimum size=7mm},		
        sqBROWNb/.style={rectangle, draw=black!240, fill=brown!15, very thick, dashed, minimum size=7mm},	
        sqPURPLE/.style={rectangle, draw=black!240, fill=purple!15, very thick, minimum size=7mm},				
		]
		
        \node[] at (5, 0) {};
        \node[] at (-5, 0) {};


		\node[sqORA](A1) at (-4, 2.5) { $\tau 1$, $y=6$};
        \node[] at (-4.25, 1.75) { $6$};

        \node[sqORA](A2) at (4, 2.5) { $\tau 2$, $y=3.6$};
        \node[] at (0.1, 1.5) { $1.8$};
		\node[] at (4.3, 1.75) { $1.8$};

		\node[sqYELL](C1) at (-4, 0) { $\tau 3$, $y=7.8$};
        \node[] at (-4.5, -1) { $3.25$};
		\node[] at (-3.25, -1) { $4.55$};

        \node[sqYELL](C2) at (4, 0) { $\tau 4$, $y=1.8$};
		\node[] at (4.3, -1) { $1.8$};

		\node[sqREDb](D1) at (0, 0) {$\tau 7$, $y=\cancel{5}\;\textcolor{red}{4.76}$};
        \node[] at (2.1, 0.55) {$\cancel{1}\;\textcolor{red}{0.95}$};
		\node[] at (-2.1, 0.55) {$\cancel{4}\;\textcolor{red}{3.81}$};

		\node[sqBLUE](E1) at (-4, -2.5) { $\tau 5$, $y=4.88$};
        \node[] at (-4.75, -3.5) { $2.93$};
		\node[] at (-2.5, -3.5) { $1.95$};

        \node[sqBLUE](E2) at (4, -2.5) { $\tau 6$, $y=9.53$};
        \node[] at (4.8, -3.5) { $2.86$};
		\node[] at (2.4, -3.5) { $1.91$};
        \node[] at (1.6, -1.6) { $4.76$};

		\node[sqGREENb](F1) at (-4, -5) {$\tau 8$, $y=\cancel{3}\;\textcolor{red}{2.93}$};
		\node[sqGREENb](F2) at (0, -5) {$\tau 9$, $y=\cancel{4}\;\textcolor{red}{3.86}$};
		\node[sqGREENb](F3) at (4, -5) {$\tau 10$, $y=\cancel{3}\;\textcolor{red}{2.86}$};

		\tikzset{thick edge/.style={-, black, fill=none, thick, text=black}}
		\tikzset{thick arc/.style={->, black, fill=black, thick, >=stealth, text=black}}

		\draw[line width=0.4mm, black,-> ] (A1) to[out=-90,in=90] (C1);
		\draw[line width=0.4mm, black,-> ] (A2) to[out=-90,in=90] (C1);
        \draw[line width=0.4mm, black,-> ] (A2) to[out=-90,in=90] (C2);

		\draw[line width=0.4mm, black,-> ] (D1) to[out=-90,in=90] (C1);
        \draw[line width=0.4mm, black,-> ] (D1) to[out=-90,in=90] (C2);

		\draw[line width=0.4mm, black,-> ] (C1) to[out=-90,in=90] (E1);
		\draw[line width=0.4mm, black,-> ] (C1) to[out=-90,in=90] (E2);
		\draw[line width=0.4mm, black,-> ] (C2) to[out=-90,in=90] (E2);

		
		\draw[line width=0.4mm, black,-> ] (E2) to[out=-90,in=90] (D1);

		\draw[line width=0.4mm, black,-> ] (E1) to[out=-90,in=90] (F1);
		\draw[line width=0.4mm, black,-> ] (E1) to[out=-90,in=90] (F2);
		\draw[line width=0.4mm, black,-> ] (E2) to[out=-90,in=90] (F3);
        \draw[line width=0.4mm, black,-> ] (E2) to[out=-90,in=90] (F2);	
		\end{tikzpicture}
}}

\subfloat[Shock further disrupts $\tau 3$ but not $\tau 4$]
{
\centering
\resizebox{2.7in}{2.16in}{		
		\begin{tikzpicture}[
		sqRED/.style={rectangle, draw=black!240, fill=red!15, very thick, minimum size=7mm},
		sqREDb/.style={rectangle, draw=red!240, fill=red!15, very thick, dashed, minimum size=7mm},
		sqBLUE/.style={rectangle, draw=black!240, fill=blue!15, very thick, minimum size=7mm},
		sqBLUEb/.style={rectangle, draw=black!240, fill=blue!15, very thick, dashed, minimum size=7mm},
        roundnode/.style={ellipse, draw=black!240, fill=green!15, very thick, dashed, minimum size=7mm},
		sqGREEN/.style={rectangle, draw=black!240, fill=green!15, very thick, minimum size=7mm},
		sqGREENb/.style={rectangle, draw=red!240, fill=green!15, very thick, dashed, minimum size=7mm},
        roundRED/.style={ellipse, draw=black!240, fill=red!15, very thick, dashed, minimum size=7mm},
		roundREDb/.style={ellipse, draw=red!240, fill=red!15, very thick, dashed, minimum size=7mm},
		roundYELL/.style={ellipse, draw=black!240, fill=yellow!20, very thick, dashed, minimum size=7mm},
		sqYELL/.style={rectangle, draw=black!240, fill=yellow!20, very thick, minimum size=7mm},
        sqYELLb/.style={rectangle, draw=red!240, fill=yellow!20, very thick, dashed, minimum size=7mm},
		roundORA/.style={ellipse, draw=black!240, fill=orange!15, very thick, dashed, minimum size=7mm},		
		sqORA/.style={rectangle, draw=black!240, fill=orange!15, very thick, minimum size=7mm},
        sqORAb/.style={rectangle, draw=black!240, fill=orange!15, very thick, dashed, minimum size=7mm},		
        squaredBLACK/.style={rectangle, draw=black!240, fill=white!15, very thick, minimum size=7mm},
        squaredBLACK2/.style={rectangle, draw=red!240, fill=white!15, very thick, dashed, minimum size=7mm},			
		roundBROWN/.style={ellipse, draw=black!240, fill=brown!15, very thick, dashed, minimum size=7mm},
		sqBROWN/.style={rectangle, draw=black!240, fill=brown!15, very thick, minimum size=7mm},		
        sqBROWNb/.style={rectangle, draw=black!240, fill=brown!15, very thick, dashed, minimum size=7mm},	
        sqPURPLE/.style={rectangle, draw=black!240, fill=purple!15, very thick, minimum size=7mm},				
		]
		
        \node[] at (5, 0) {};
        \node[] at (-5, 0) {};


		\node[sqORA](A1) at (-4, 2.5) { $\tau 1$, $y=6$};
        \node[] at (-4.25, 1.75) { $6$};

        \node[sqORA](A2) at (4, 2.5) { $\tau 2$, $y=3.6$};
        \node[] at (0.1, 1.5) { $1.8$};
		\node[] at (4.3, 1.75) { $1.8$};

		\node[sqYELLb](C1) at (-4, 0) { $\tau 3$, $y=\cancel{7.8}\;\textcolor{red}{7.62}$};
        \node[] at (-5.00, -1) { $\cancel{3.25}\;\textcolor{red}{3.17}$};
		\node[] at (-3.05, -1) { $\cancel{4.55}\;\textcolor{red}{4.45}$};

        \node[sqYELL](C2) at (4, 0) { $\tau 4$, $y=1.8$};
		\node[] at (4.3, -1) { $1.8$};

		\node[sqRED](D1) at (0, 0) {$\tau 7$, $y=4.76$};
        \node[] at (2.1, 0.55) {$0.95$};
		\node[] at (-2.1, 0.55) {$3.81$};

		\node[sqBLUE](E1) at (-4, -2.5) { $\tau 5$, $y=4.88$};
        \node[] at (-4.75, -3.5) { $2.93$};
		\node[] at (-2.5, -3.5) { $1.95$};

        \node[sqBLUE](E2) at (4, -2.5) { $\tau 6$, $y=9.53$};
        \node[] at (4.8, -3.5) { $2.86$};
		\node[] at (2.4, -3.5) { $1.91$};
        \node[] at (1.6, -1.6) { $4.76$};

		\node[sqGREEN](F1) at (-4, -5) {$\tau 8$, $y=2.93$};
		\node[sqGREEN](F2) at (0, -5) {$\tau 9$, $y=3.86$};
		\node[sqGREEN](F3) at (4, -5) {$\tau 10$, $y=2.86$};

		\tikzset{thick edge/.style={-, black, fill=none, thick, text=black}}
		\tikzset{thick arc/.style={->, black, fill=black, thick, >=stealth, text=black}}

		\draw[line width=0.4mm, black,-> ] (A1) to[out=-90,in=90] (C1);
		\draw[line width=0.4mm, black,-> ] (A2) to[out=-90,in=90] (C1);
        \draw[line width=0.4mm, black,-> ] (A2) to[out=-90,in=90] (C2);

		\draw[line width=0.4mm, black,-> ] (D1) to[out=-90,in=90] (C1);
        \draw[line width=0.4mm, black,-> ] (D1) to[out=-90,in=90] (C2);

		\draw[line width=0.4mm, black,-> ] (C1) to[out=-90,in=90] (E1);
		\draw[line width=0.4mm, black,-> ] (C1) to[out=-90,in=90] (E2);
		\draw[line width=0.4mm, black,-> ] (C2) to[out=-90,in=90] (E2);

		
		\draw[line width=0.4mm, black,-> ] (E2) to[out=-90,in=90] (D1);

		\draw[line width=0.4mm, black,-> ] (E1) to[out=-90,in=90] (F1);
		\draw[line width=0.4mm, black,-> ] (E1) to[out=-90,in=90] (F2);
		\draw[line width=0.4mm, black,-> ] (E2) to[out=-90,in=90] (F3);
        \draw[line width=0.4mm, black,-> ] (E2) to[out=-90,in=90] (F2);	
		\end{tikzpicture}
}}\hfill
\subfloat[Algorithm converges to these flows]
{
\centering
\resizebox{2.7in}{2.16in}{			
		\begin{tikzpicture}[
		sqRED/.style={rectangle, draw=black!240, fill=red!15, very thick, minimum size=7mm},
		sqREDb/.style={rectangle, draw=red!240, fill=red!15, very thick, dashed, minimum size=7mm},
		sqBLUE/.style={rectangle, draw=black!240, fill=blue!15, very thick, minimum size=7mm},
		sqBLUEb/.style={rectangle, draw=black!240, fill=blue!15, very thick, dashed, minimum size=7mm},
        roundnode/.style={ellipse, draw=black!240, fill=green!15, very thick, dashed, minimum size=7mm},
		sqGREEN/.style={rectangle, draw=black!240, fill=green!15, very thick, minimum size=7mm},
		sqGREENb/.style={rectangle, draw=red!240, fill=green!15, very thick, dashed, minimum size=7mm},
        roundRED/.style={ellipse, draw=black!240, fill=red!15, very thick, dashed, minimum size=7mm},
		roundREDb/.style={ellipse, draw=red!240, fill=red!15, very thick, dashed, minimum size=7mm},
		roundYELL/.style={ellipse, draw=black!240, fill=yellow!20, very thick, dashed, minimum size=7mm},
		sqYELL/.style={rectangle, draw=black!240, fill=yellow!20, very thick, minimum size=7mm},
        sqYELLb/.style={rectangle, draw=black!240, fill=yellow!20, very thick, dashed, minimum size=7mm},
		roundORA/.style={ellipse, draw=black!240, fill=orange!15, very thick, dashed, minimum size=7mm},		
		sqORA/.style={rectangle, draw=black!240, fill=orange!15, very thick, minimum size=7mm},
        sqORAb/.style={rectangle, draw=black!240, fill=orange!15, very thick, dashed, minimum size=7mm},		
        squaredBLACK/.style={rectangle, draw=black!240, fill=white!15, very thick, minimum size=7mm},
        squaredBLACK2/.style={rectangle, draw=red!240, fill=white!15, very thick, dashed, minimum size=7mm},			
		roundBROWN/.style={ellipse, draw=black!240, fill=brown!15, very thick, dashed, minimum size=7mm},
		sqBROWN/.style={rectangle, draw=black!240, fill=brown!15, very thick, minimum size=7mm},		
        sqBROWNb/.style={rectangle, draw=black!240, fill=brown!15, very thick, dashed, minimum size=7mm},	
        sqPURPLE/.style={rectangle, draw=black!240, fill=purple!15, very thick, minimum size=7mm},				
		]
		
        \node[] at (5, 0) {};
        \node[] at (-5, 0) {};

		\node[sqORA](A1) at (-4, 2.5) { $\tau 1$, $y=6$};
        \node[] at (-4.25, 1.75) { $6$};

        \node[sqORA](A2) at (4, 2.5) { $\tau 2$, $y=3.6$};
        \node[] at (0.1, 1.5) { $1.8$};
		\node[] at (4.3, 1.75) { $1.8$};

		\node[sqYELL](C1) at (-4, 0) { $\tau 3$, $y=7.2$};
        \node[] at (-4.50, -1) { $3$};
		\node[] at (-3.25, -1) { $4.2$};

        \node[sqYELL](C2) at (4, 0) { $\tau 4$, $y=1.8$};
		\node[] at (4.3, -1) { $1.8$};

		\node[sqRED](D1) at (0, 0) {$\tau 7$, $y=4.5$};
        \node[] at (2.1, 0.55) {$0.9$};
		\node[] at (-2.1, 0.55) {$3.6$};

		\node[sqBLUE](E1) at (-4, -2.5) { $\tau 5$, $y=4.5$};
        \node[] at (-4.75, -3.5) { $2.7$};
		\node[] at (-2.5, -3.5) { $1.8$};

        \node[sqBLUE](E2) at (4, -2.5) { $\tau 6$, $y=9$};
        \node[] at (4.8, -3.5) { $2.7$};
		\node[] at (2.4, -3.5) { $1.8$};
        \node[] at (1.6, -1.6) { $4.5$};

		\node[sqGREEN](F1) at (-4, -5) {$\tau 8$, $y=2.7$};
		\node[sqGREEN](F2) at (0, -5) {$\tau 9$, $y=3.6$};
		\node[sqGREEN](F3) at (4, -5) {$\tau 10$, $y=2.7$};

		\tikzset{thick edge/.style={-, black, fill=none, thick, text=black}}
		\tikzset{thick arc/.style={->, black, fill=black, thick, >=stealth, text=black}}

		\draw[line width=0.4mm, black,-> ] (A1) to[out=-90,in=90] (C1);
		\draw[line width=0.4mm, black,-> ] (A2) to[out=-90,in=90] (C1);
        \draw[line width=0.4mm, black,-> ] (A2) to[out=-90,in=90] (C2);

		\draw[line width=0.4mm, black,-> ] (D1) to[out=-90,in=90] (C1);
        \draw[line width=0.4mm, black,-> ] (D1) to[out=-90,in=90] (C2);

		\draw[line width=0.4mm, black,-> ] (C1) to[out=-90,in=90] (E1);
		\draw[line width=0.4mm, black,-> ] (C1) to[out=-90,in=90] (E2);
		\draw[line width=0.4mm, black,-> ] (C2) to[out=-90,in=90] (E2);

		
		\draw[line width=0.4mm, black,-> ] (E2) to[out=-90,in=90] (D1);

		\draw[line width=0.4mm, black,-> ] (E1) to[out=-90,in=90] (F1);
		\draw[line width=0.4mm, black,-> ] (E1) to[out=-90,in=90] (F2);
		\draw[line width=0.4mm, black,-> ] (E2) to[out=-90,in=90] (F3);
        \draw[line width=0.4mm, black,-> ] (E2) to[out=-90,in=90] (F2);	
		\end{tikzpicture}
}}\caption{\footnotesize{An illustration of the Shock Propagation Algorithm on the example from Figure \ref{eqfig}\label{shockalg}. Nodes that are the same color produce the same good. }}
\end{figure}
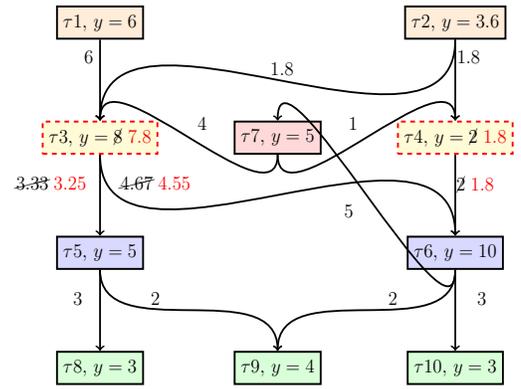
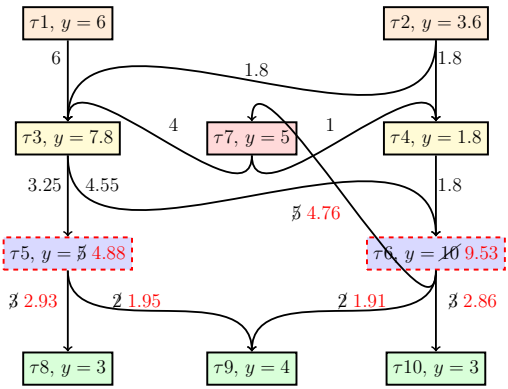
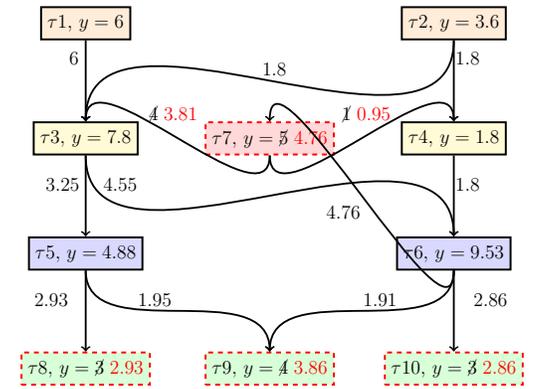
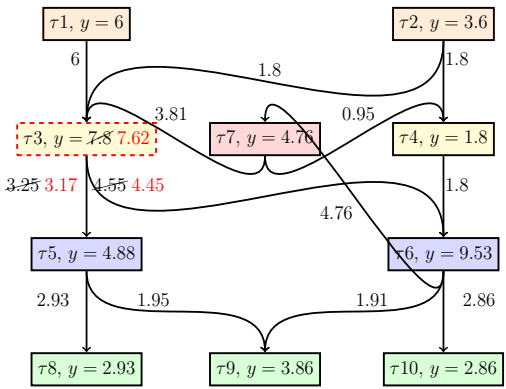
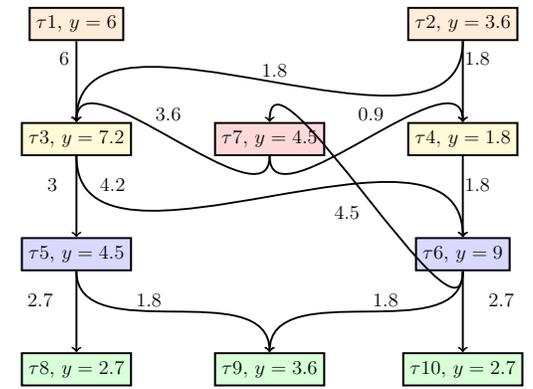

The Shock Propagation Algorithm is illustrated in Figure \ref{shockalg} in which there is
a 10 percent shock to technology $\tau2$. It traces the impact of the shock and updates the output of a node each time the supply of one of its inputs is reduced.
In this example there are several cycles, and note that $\tau6$ and $\tau7$ are on cycles involving $\tau3$ and $\tau4$.  The feedback via $\tau4$ stops after a few steps in the algorithm since the diminished flow directly from $\tau2$ ends up being the binding one.  However, the feedback via $\tau3$ continues infinitely, and eventually converges to the final levels pictured in panel (f).

Extending this example helps illustrate the concept of a disrupted industries sub-network as used in Proposition \ref{prop:cut}. This is shown in Figure \ref{fig:disruption}.

\begin{figure}[!ht]
    \subfloat[Equilibrium flows ${\mathcal{G}}$\label{subfig:disruption-0}]
    {
\centering
\resizebox{3.7in}{2.62in}{		
    \begin{tikzpicture}[
		sqRED/.style={rectangle, draw=black!240, fill=red!15, very thick, minimum size=7mm},
		sqREDb/.style={rectangle, draw=black!240, fill=red!15, very thick, dashed, minimum size=7mm},
		sqBLUE/.style={rectangle, draw=black!240, fill=blue!15, very thick, minimum size=7mm},
		sqBLUEb/.style={rectangle, draw=black!240, fill=blue!15, very thick, dashed, minimum size=7mm},
		roundnode/.style={ellipse, draw=black!240, fill=green!15, very thick, dashed, minimum size=7mm},
		sqGREEN/.style={rectangle, draw=black!240, fill=green!15, very thick, minimum size=7mm},
		roundRED/.style={ellipse, draw=black!240, fill=red!15, very thick, dashed, minimum size=7mm},
		roundREDb/.style={ellipse, draw=black!240, fill=red!15, very thick, minimum size=7mm},
		roundYELL/.style={ellipse, draw=black!240, fill=yellow!20, very thick, dashed, minimum size=7mm},
		sqYELL/.style={rectangle, draw=black!240, fill=yellow!20, very thick, minimum size=7mm},
		roundORA/.style={ellipse, draw=black!240, fill=orange!15, very thick, dashed, minimum size=7mm},		
		sqORA/.style={rectangle, draw=black!240, fill=orange!15, very thick, minimum size=7mm},
		sqORAb/.style={rectangle, draw=black!240, fill=orange!15, very thick, dashed, minimum size=7mm},
		sqMAG/.style={rectangle, draw=black!240, fill=magenta!5, very thick, minimum size=7mm},
		squaredBLACK/.style={rectangle, draw=black!240, fill=white!15, very thick, minimum size=7mm},	
		roundBROWN/.style={ellipse, draw=black!240, fill=brown!15, very thick, dashed, minimum size=7mm},
		sqBROWN/.style={rectangle, draw=black!240, fill=brown!15, very thick, minimum size=7mm},		
		]
        \node[] at (11, 0) {};		
        \node[] at (-6, 0) {};

		\node[squaredBLACK](R2) at (7, 5) { $\tau0$, $y=10$};
        \node[] at (7.25, 4) { $10$};
        \node[sqBROWN](Z1) at (7, 0) { $\tau 11$, $y=5$};
		\node[] at (7.25, -1) { $5$};

		\node[sqGREEN](X1) at (7, -5) { $\tau 12$, $y=5$};


		\node[sqORA](A1) at (-4, 2.5) { $\tau 1$, $y=6$};
        \node[] at (-4.25, 1.75) { $6$};

        \node[sqORA](A2) at (4, 2.5) { $\tau 2$, $y=6$};
        \node[] at (0, 1.5) { $2$};
		\node[] at (4.15, 1.15) { $2$};
		\node[] at (4.5, 1.75) { $2$};

		\node[sqYELL](C1) at (-4, 0) { $\tau 3$, $y=8$};
        \node[] at (-4.5, -1) { $3.33$};
		\node[] at (-3.25, -1) { $4.67$};

        \node[sqYELL](C2) at (4, 0) { $\tau 4$, $y=2$};
		\node[] at (4.25, -1) { $2$};

		\node[sqRED](D1) at (0, 0) { $\tau 7$, $y=5$};
        \node[] at (1.7, 0.1) { $1$};
		\node[] at (-1.7, 0.1) { $4$};

		\node[sqBLUE](E1) at (-4, -2.5) { $\tau 5$, $y=5$};
        \node[] at (-4.5, -3.5) { $3$};
		\node[] at (-3, -3.5) { $2$};

        \node[sqBLUE](E2) at (4, -2.5) { $\tau 6$, $y=10$};
        \node[] at (4.5, -3.5) { $3$};
		\node[] at (2.75, -3.5) { $2$};
        \node[] at (1.9, -1.6) { $5$};

		\node[sqGREEN](F1) at (-4, -5) { $\tau 8$, $y=3$};
		\node[sqGREEN](F2) at (0, -5) { $\tau 9$, $y=4$};
		\node[sqGREEN](F3) at (4, -5) { $\tau 10$, $y=3$};

		\tikzset{thick edge/.style={-, black, fill=none, thick, text=black}}
		\tikzset{thick arc/.style={->, black, fill=black, thick, >=stealth, text=black}}


		\draw[line width=0.4mm, black,-> ] (A1) to[out=-90,in=90] (C1);
		\draw[line width=0.4mm, black,-> ] (A2) to[out=-90,in=90] (C1);
        \draw[line width=0.4mm, black,-> ] (A2) to[out=-90,in=90] (C2);

		\draw[line width=0.4mm, black,-> ] (D1) to[out=-90,in=90] (C1);
        \draw[line width=0.4mm, black,-> ] (D1) to[out=-90,in=90] (C2);

		\draw[line width=0.4mm, black,-> ] (C1) to[out=-90,in=90] (E1);
		\draw[line width=0.4mm, black,-> ] (C1) to[out=-90,in=90] (E2);
		\draw[line width=0.4mm, black,-> ] (C2) to[out=-90,in=90] (E2);

		\draw[line width=0.4mm, black,-> ] (E2) to[out=-90,in=90, looseness = 1.7] (D1);

		\draw[line width=0.4mm, black,-> ] (E1) to[out=-90,in=90] (F1);
		\draw[line width=0.4mm, black,-> ] (E1) to[out=-90,in=90] (F2);
		\draw[line width=0.4mm, black,-> ] (E2) to[out=-90,in=90] (F3);
        \draw[line width=0.4mm, black,-> ] (E2) to[out=-90,in=90] (F2);

        \draw[line width=0.4mm, black,-> ] (R2) to[out=-90,in=90] (Z1);
        \draw[line width=0.4mm, black,-> ] (Z1) to[out=-90,in=90] (X1);
        \draw[line width=0.4mm, black,-> ] (A2) to[out=-90,in=90] (Z1);

		\end{tikzpicture}
		}
	}
    \hfill
    \subfloat[Disrupted industries sub-network $\widehat{{\mathcal{G}}}(\tau11)$\label{subfig:disruption-1}]
    {
     \centering
\resizebox{2.8in}{2.62in}{		
    \begin{tikzpicture}[
		sqRED/.style={rectangle, draw=black!240, fill=red!15, very thick, minimum size=7mm},
		sqREDb/.style={rectangle, draw=black!240, fill=red!15, very thick, dashed, minimum size=7mm},
		sqBLUE/.style={rectangle, draw=black!240, fill=blue!15, very thick, minimum size=7mm},
		sqBLUEb/.style={rectangle, draw=black!240, fill=blue!15, very thick, dashed, minimum size=7mm},
		roundnode/.style={ellipse, draw=black!240, fill=green!15, very thick, dashed, minimum size=7mm},
		sqGREEN/.style={rectangle, draw=black!240, fill=green!15, very thick, minimum size=7mm},
		roundRED/.style={ellipse, draw=black!240, fill=red!15, very thick, dashed, minimum size=7mm},
		roundREDb/.style={ellipse, draw=black!240, fill=red!15, very thick, minimum size=7mm},
		roundYELL/.style={ellipse, draw=black!240, fill=yellow!20, very thick, dashed, minimum size=7mm},
		sqYELL/.style={rectangle, draw=black!240, fill=yellow!20, very thick, minimum size=7mm},
		roundORA/.style={ellipse, draw=black!240, fill=orange!15, very thick, dashed, minimum size=7mm},		
		sqORA/.style={rectangle, draw=black!240, fill=orange!15, very thick, minimum size=7mm},
		sqORAb/.style={rectangle, draw=black!240, fill=orange!15, very thick, dashed, minimum size=7mm},
		sqMAG/.style={rectangle, draw=black!240, fill=magenta!5, very thick, minimum size=7mm},
		squaredBLACK/.style={rectangle, draw=black!240, fill=white!15, very thick, minimum size=7mm},	
		roundBROWN/.style={ellipse, draw=black!240, fill=brown!15, very thick, dashed, minimum size=7mm},
		sqBROWN/.style={rectangle, draw=black!240, fill=brown!15, very thick, minimum size=7mm},		
		sqBROWNb/.style={rectangle, draw=black!240, fill=brown!15, very thick, dashed, minimum size=7mm},		
		]
        \node[] at (12, 0) {};		
        \node[] at (0, -6) {};		

		\node[squaredBLACK](R2) at (7, 5) {$\tau0$, $y=10$};
        \node[] at (7.25, 4) { $10$};
        \node[sqBROWNb](Z1) at (7, 0) { $\tau 11$, $y=5$};
		\node[] at (7.25, -1) { $5$};

		\node[sqGREEN](X1) at (7, -5) { $\tau 12$, $y=5$};


        \node[sqORA](A2) at (4, 2.5) { $\tau 2$, $y=6$};

		\node[] at (4.5, 1.75) { $2$};
			
		\tikzset{thick edge/.style={-, black, fill=none, thick, text=black}}
		\tikzset{thick arc/.style={->, black, fill=black, thick, >=stealth, text=black}}


        \draw[line width=0.4mm, black,-> ] (R2) to[out=-90,in=90] (Z1);
        \draw[line width=0.4mm, black,-> ] (Z1) to[out=-90,in=90] (X1);
        \draw[line width=0.4mm, black,-> ] (A2) to[out=-90,in=90] (Z1);
		
		\end{tikzpicture}
		}
    }
    \caption{\footnotesize{Disrupted industries sub-networks: For the equilibrium flows in panel (a), panel (b) shows the disrupted industries subnetwork for $T^{Shocked}=\{\tau11\}$. The disrupted industries sub-network for $T^{Shocked}=\{\tau2\}$ is the network shown in panel (a), while the disrupted industries sub-network for $T^{Shocked}=\{\tau1\}$ is the same as the equilibrium flows network shown in Figure \ref{eqfig}.}}
    \label{fig:disruption}
  \end{figure}
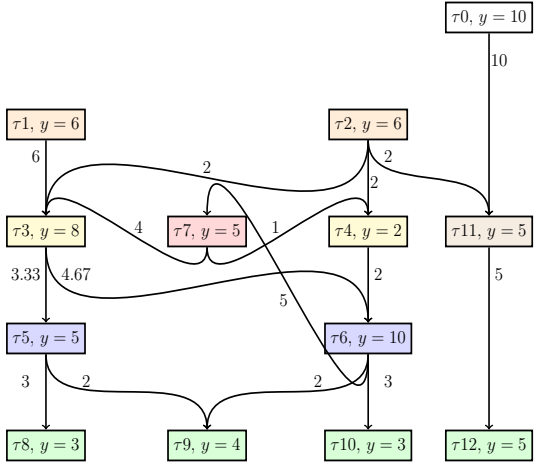
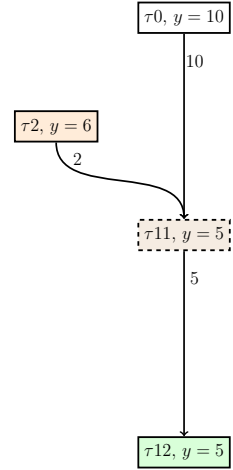

Finally, we use this example, but with the cycle removed, to illustrate the disruption centrality calculation. \begin{figure}[!ht]
    \subfloat[Disruptions from $\tau2$ to $\tau8,\tau9,\tau10$]
    {
\centering
\resizebox{3.7in}{2.62in}{		
    \begin{tikzpicture}[
sqRED/.style={rectangle, draw=black!240, fill=red!15, very thick, minimum size=7mm},
		sqREDb/.style={rectangle, draw=black!240, fill=red!15, very thick, dashed, minimum size=7mm},
		sqBLUE/.style={rectangle, draw=black!240, fill=blue!15, very thick, minimum size=7mm},
		sqBLUEb/.style={rectangle, draw=black!240, fill=blue!15, very thick, dashed, minimum size=7mm},
		roundnode/.style={ellipse, draw=black!240, fill=green!15, very thick, dashed, minimum size=7mm},
		sqGREEN/.style={rectangle, draw=black!240, fill=green!15, very thick, minimum size=7mm},
		sqGREENb/.style={rectangle, draw=black!240, fill=green!15, very thick, dashed, minimum size=7mm},
		roundRED/.style={ellipse, draw=black!240, fill=red!15, very thick, dashed, minimum size=7mm},
		roundREDb/.style={ellipse, draw=black!240, fill=red!15, very thick, minimum size=7mm},
		roundYELL/.style={ellipse, draw=black!240, fill=yellow!20, very thick, dashed, minimum size=7mm},
		sqYELL/.style={rectangle, draw=black!240, fill=yellow!20, very thick, minimum size=7mm},
		roundORA/.style={ellipse, draw=black!240, fill=orange!15, very thick, dashed, minimum size=7mm},		
		sqORA/.style={rectangle, draw=black!240, fill=orange!15, very thick, minimum size=7mm},
		sqORAb/.style={rectangle, draw=black!240, fill=orange!15, very thick, dashed, minimum size=7mm},
		sqMAG/.style={rectangle, draw=black!240, fill=magenta!5, very thick, minimum size=7mm},
		squaredBLACK/.style={rectangle, draw=black!240, fill=white!15, very thick, minimum size=7mm},	
		roundBROWN/.style={ellipse, draw=black!240, fill=brown!15, very thick, dashed, minimum size=7mm},
		sqBROWN/.style={rectangle, draw=black!240, fill=brown!15, very thick, minimum size=7mm},		
		]
        \node[] at (8.5, 0) {};		
        \node[] at (-8.5, 0) {};

 \node[] at (7, 5) {};

		\node[] at (0, 5) {};

		\node[sqORA](A1) at (-4, 2.5) {$\tau1, d=0$};
        \node[] at (-4.5, 1.75) { $0.75$};

        \node[sqORAb](A2) at (4, 2.5) {$\tau2, d=1$};
        \node[] at (0, 1.5) { $0.25$};
		\node[] at (4.15, 1.15) { $1$};
			
		\node[sqYELL](C1) at (-4, 0) {$\tau3,d=0.25$};
        \node[] at (-4.5, -1) { $1$};
		\node[] at (-3.25, -1) { $0.7$};

        \node[sqYELL](C2) at (4, 0) {$\tau4,d=1$};
		\node[] at (4.5, -1) { $0.3$};

		\node[sqBLUE](E1) at (-4, -2.5) {$\tau6,d=0.25$ };
        \node[] at (-4.5, -3.5) { $1$};
		\node[] at (-2.75, -3.5) { $0.5$};

        \node[sqBLUE](E2) at (4, -2.5) {$\tau7,d=0.475$};
        \node[] at (4.5, -3.5) { $1$};
		\node[] at (2.75, -3.5) { $0.5$};
		
		\node[sqGREENb](F1) at (-4, -5) {$\tau8,d=0.25$};
		\node[sqGREENb](F2) at (0, -5) {$\tau9,d=0.3625$};
		\node[sqGREENb](F3) at (4, -5) {$\tau10,d=0.475$};	
		
		\tikzset{thick edge/.style={-, black, fill=none, thick, text=black}}
		\tikzset{thick arc/.style={->, black, fill=black, thick, >=stealth, text=black}}

		\draw[line width=0.4mm, black,-> ] (A1) to[out=-90,in=90] (C1);
		\draw[line width=0.4mm, black,-> ] (A2) to[out=-90,in=90] (C1);
        \draw[line width=0.4mm, black,-> ] (A2) to[out=-90,in=90] (C2);

		\draw[line width=0.4mm, black,-> ] (C1) to[out=-90,in=90] (E1);
		\draw[line width=0.4mm, black,-> ] (C1) to[out=-90,in=90] (E2);
		\draw[line width=0.4mm, black,-> ] (C2) to[out=-90,in=90] (E2);

		\draw[line width=0.4mm, black,-> ] (E1) to[out=-90,in=90] (F1);
		\draw[line width=0.4mm, black,-> ] (E1) to[out=-90,in=90] (F2);
		\draw[line width=0.4mm, black,-> ] (E2) to[out=-90,in=90] (F3);
        \draw[line width=0.4mm, black,-> ] (E2) to[out=-90,in=90] (F2);
		
		\end{tikzpicture}
		}
	}
    \hfill
    \subfloat[Disruptions from $\tau2$ downstream to $\tau12$]
    {
     \centering
\resizebox{2.8in}{2.1in}{		
    \begin{tikzpicture}[
				sqRED/.style={rectangle, draw=black!240, fill=red!15, very thick, minimum size=7mm},
		sqREDb/.style={rectangle, draw=black!240, fill=red!15, very thick, dashed, minimum size=7mm},
		sqBLUE/.style={rectangle, draw=black!240, fill=blue!15, very thick, minimum size=7mm},
sqGREENb/.style={rectangle, draw=black!240, fill=green!15, very thick, dashed, minimum size=7mm},
		sqBLUEb/.style={rectangle, draw=black!240, fill=blue!15, very thick, dashed, minimum size=7mm},
		roundnode/.style={ellipse, draw=black!240, fill=green!15, very thick, dashed, minimum size=7mm},
		sqGREEN/.style={rectangle, draw=black!240, fill=green!15, very thick, minimum size=7mm},
		roundRED/.style={ellipse, draw=black!240, fill=red!15, very thick, dashed, minimum size=7mm},
		roundREDb/.style={ellipse, draw=black!240, fill=red!15, very thick, minimum size=7mm},
		roundYELL/.style={ellipse, draw=black!240, fill=yellow!20, very thick, dashed, minimum size=7mm},
		sqYELL/.style={rectangle, draw=black!240, fill=yellow!20, very thick, minimum size=7mm},
		roundORA/.style={ellipse, draw=black!240, fill=orange!15, very thick, dashed, minimum size=7mm},		
		sqORA/.style={rectangle, draw=black!240, fill=orange!15, very thick, minimum size=7mm},
		sqORAb/.style={rectangle, draw=black!240, fill=orange!15, very thick, dashed, minimum size=7mm},
		sqMAG/.style={rectangle, draw=black!240, fill=magenta!5, very thick, minimum size=7mm},
		sqMAGb/.style={rectangle, draw=black!240, fill=magenta!5, very thick, dashed, minimum size=7mm},
		squaredBLACK/.style={rectangle, draw=black!240, fill=white!15, very thick, minimum size=7mm},	
		roundBROWN/.style={ellipse, draw=black!240, fill=brown!15, very thick, dashed, minimum size=7mm},
		sqBROWN/.style={rectangle, draw=black!240, fill=brown!15, very thick, minimum size=7mm},		
		]
        \node[] at (0, 3) {};
        \node[] at (13, 3) {};
	
		\node[squaredBLACK](R2) at (7, 4) {$\tau0,d=0$};
        \node[] at (7.25, 3) { $1$};
        \node[sqBROWN](Z1) at (7, 0) {$\tau11,d=1$};
		\node[] at (7.25, -1) { $1$};

		\node[sqGREENb](X1) at (7, -3) {$\tau12,d=1$};

        \node[sqORAb](A2) at (4, 2.5) {$\tau2,d=1$};

		\node[] at (4.5, 1.75) { $1$};

		\tikzset{thick edge/.style={-, black, fill=none, thick, text=black}}
		\tikzset{thick arc/.style={->, black, fill=black, thick, >=stealth, text=black}}

        \draw[line width=0.4mm, black,-> ] (R2) to[out=-90,in=90] (Z1);
        \draw[line width=0.4mm, black,-> ] (Z1) to[out=-90,in=90] (X1);
        \draw[line width=0.4mm, black,-> ] (A2) to[out=-90,in=90] (Z1);
	
		\end{tikzpicture}
		}
    }
    \caption{\footnotesize{Consider the equilibrium flows shown in Figure \ref{fig:disruption}(a). Technology $\tau2 $ affects parts of the network,
     that leading to $\tau{8}, \tau{9}$ and $\tau{10}$, as well as that produced by $\tau{12}$. Panel (a) shows the disruptions for goods between $\tau2$ and $\tau{8}, \tau{9}$ and $\tau{10}$, as well as for good $\tau1$ as this is needed for the recursive calculations. Panel (b) shows the disruptions for goods between $\tau2$ and $\tau12$ as well as for technology $\tau0$. In both panels the weight on an edge from $\tau'$ to $\tau''$ represents $S(\tau',\tau'')$.}}
    \label{fig:disruption_centrality}
  \end{figure}
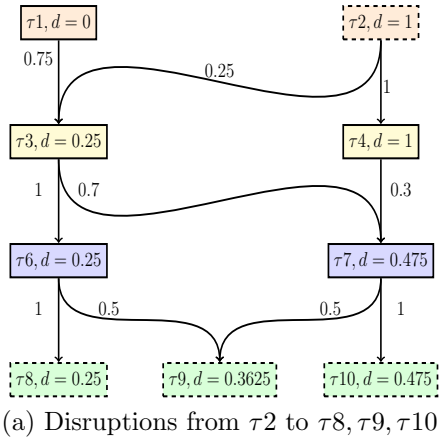
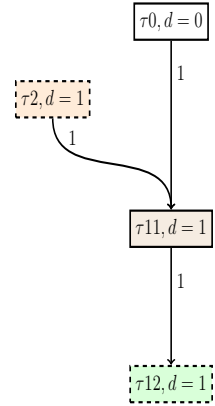

Figure \ref{fig:disruption_centrality} provides an example of how the $d_{\tau,f}(\tau')$ terms are calculated for the equilibrium flows shown in Figure \ref{fig:disruption}(a) with the disrupted technology being $\tau=\tau2$. 
As an example, if the equilibrium prices of the final goods are all 1 (which depend on the labor inputs of the different technologies), then the disruption centrality of technology $\tau2$ for these equilibrium flows is
\begin{equation}
D(\tau2)=\frac{10}{15} \bigg(0.25(0.3)+0.3625(0.4)+0.475(0.3)\bigg)+\frac{5}{15}=0.575\nonumber
\end{equation}

\section{The Medium Run}\label{sec:mediumrun}


While in the short run existing contracts prevent prices from adjusting, as embedded within our proportional rationing assumption, once prices can adjust existing production can reroute which can decrease the impact of the shock, and end up between
the short run and long run impacts.\footnote{E.g., see \cite{iyoha2024exports}.} We now consider this medium run situation, in which prices can adjust and 
shortages can be rationed efficiently, but in which production technologies are fixed, as is labor, which prevents any firms from increasing their output beyond pre-shock levels.

As the medium run can vary and end up anywhere between the short and long run outcomes, depending on circumstances, we provide examples and illustrate the forces, as a full characterization is simply an equilibrium statement that does not provide much further insight.

Flexible prices can correlate downstream disruptions in a way that minimizes the overall disruption. Consider, for example, 
the short run impact of a shock to producer $\tau1$ (without flexible prices) shown in panel (a) of Figure \ref{fig:flexible}.

Here a disruption to $\tau1$ affects the production of the final goods $\tau4$ through two potential channels. First $\tau1$ is used directly as an input, and secondly it is used to produce $\tau3$ which is an input for $\tau4$.  In this example $\tau1$ and $\tau2$ produce the same good,  and so by redirecting some of the output of $\tau1$ to $\tau3$ there is less disruption of that good which is critical to the production of
$\tau4$, while there is a substitute for $\tau1$ in the direct input for $\tau4$.

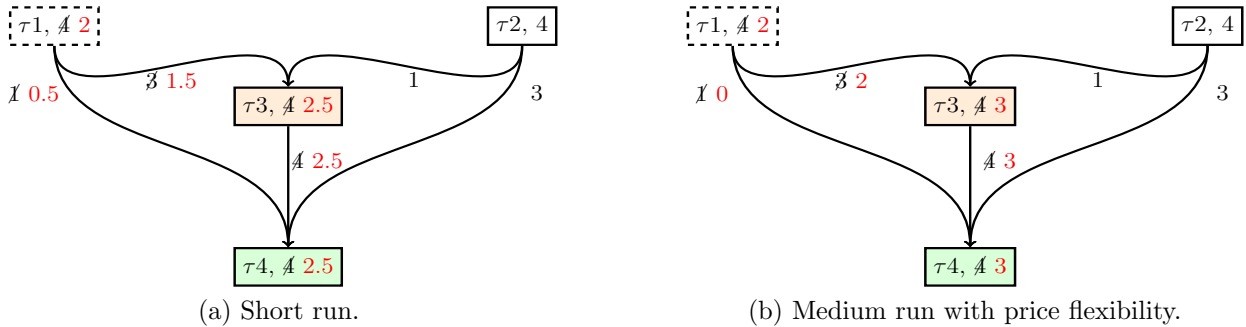
\begin{figure}[!ht]
    \subfloat[Short run.]
    {
\centering
\resizebox{3in}{1.5in}{		
\begin{tikzpicture}[
		sqRED/.style={rectangle, draw=black!240, fill=red!15, very thick, minimum size=7mm},
		sqBLUE/.style={rectangle, draw=black!240, fill=blue!15, very thick, minimum size=7mm},
		roundnode/.style={ellipse, draw=black!240, fill=green!15, very thick, dashed, minimum size=7mm},
		sqGREEN/.style={rectangle, draw=black!240, fill=green!15, very thick, minimum size=7mm},
		roundRED/.style={ellipse, draw=black!240, fill=red!15, very thick, dashed, minimum size=7mm},
		roundREDb/.style={ellipse, draw=black!240, fill=red!15, very thick, minimum size=7mm},
		roundYELL/.style={ellipse, draw=black!240, fill=yellow!20, very thick, dashed, minimum size=7mm},
		sqYELL/.style={rectangle, draw=black!240, fill=yellow!20, very thick, minimum size=7mm},
		roundORA/.style={ellipse, draw=black!240, fill=orange!15, very thick, dashed, minimum size=7mm},		
		sqORA/.style={rectangle, draw=black!240, fill=orange!15, very thick, minimum size=7mm},
		squaredBLACK/.style={rectangle, draw=black!240, fill=white!15, very thick, minimum size=7mm},	
		squaredBLACK2/.style={rectangle, draw=black!240, fill=white!15, very thick, dashed, minimum size=7mm},	
		roundBROWN/.style={ellipse, draw=black!240, fill=brown!15, very thick, dashed, minimum size=7mm},
		sqBROWN/.style={rectangle, draw=black!240, fill=brown!15, very thick, minimum size=7mm},		
        sqPURPLE/.style={rectangle, draw=black!240, fill=purple!15, very thick, minimum size=7mm},				
		]


		\node[squaredBLACK2](R1) at (-4, 2.5) { $\tau1$, $\cancel{4}$ \textcolor{red}{$2$}};
		\node[squaredBLACK](R2) at (4, 2.5) { $\tau2$, $4$};

        \node[] at (2.15, 1.5) { $1$};
        \node[] at (-2, 1.5) { $\cancel{3}$ \textcolor{red}{$1.5$}};

		\node[sqORA](I1) at (0, 1) {$\tau3$, $\cancel{4}$ \textcolor{red}{$2.5$}};
        \node[] at (0.5, 0) { $\cancel{4}$ \textcolor{red}{$2.5$}};
		\node[] at (-4.35, 1.25) { $\cancel{1}$ \textcolor{red}{$0.5$}};

		\node[] at (4.25, 1.25) { $3$};

		\node[sqGREEN](F1) at (0, -2) { $\tau4$, $\cancel{4}$ \textcolor{red}{$2.5$}};
	
		\tikzset{thick edge/.style={-, black, fill=none, thick, text=black}}
		\tikzset{thick arc/.style={->, black, fill=black, thick, >=stealth, text=black}}

		\draw[line width=0.4mm, black,-> ] (R1) to[out=-90,in=90] (I1);
		\draw[line width=0.4mm, black,-> ] (R2) to[out=-90,in=90] (I1);

		\draw[line width=0.4mm, black,-> ] (I1) to[out=-90,in=90] (F1);

		\draw[line width=0.4mm, black,-> ] (R1) to[out=-90,in=90] (F1);
		\draw[line width=0.4mm, black,-> ] (R2) to[out=-90,in=90] (F1);

		\end{tikzpicture}
}
	}
    \hfill
    \subfloat[Medium run with price flexibility.]
    {
     \centering
\resizebox{3in}{1.5in}{		

\begin{tikzpicture}[
		sqRED/.style={rectangle, draw=black!240, fill=red!15, very thick, minimum size=7mm},
		sqBLUE/.style={rectangle, draw=black!240, fill=blue!15, very thick, minimum size=7mm},
		roundnode/.style={ellipse, draw=black!240, fill=green!15, very thick, dashed, minimum size=7mm},
		sqGREEN/.style={rectangle, draw=black!240, fill=green!15, very thick, minimum size=7mm},
		roundRED/.style={ellipse, draw=black!240, fill=red!15, very thick, dashed, minimum size=7mm},
		roundREDb/.style={ellipse, draw=black!240, fill=red!15, very thick, minimum size=7mm},
		roundYELL/.style={ellipse, draw=black!240, fill=yellow!20, very thick, dashed, minimum size=7mm},
		sqYELL/.style={rectangle, draw=black!240, fill=yellow!20, very thick, minimum size=7mm},
		roundORA/.style={ellipse, draw=black!240, fill=orange!15, very thick, dashed, minimum size=7mm},		
		sqORA/.style={rectangle, draw=black!240, fill=orange!15, very thick, minimum size=7mm},
		squaredBLACK/.style={rectangle, draw=black!240, fill=white!15, very thick, minimum size=7mm},	
		squaredBLACK2/.style={rectangle, draw=black!240, fill=white!15, very thick, dashed, minimum size=7mm},	
		roundBROWN/.style={ellipse, draw=black!240, fill=brown!15, very thick, dashed, minimum size=7mm},
		sqBROWN/.style={rectangle, draw=black!240, fill=brown!15, very thick, minimum size=7mm},		
        sqPURPLE/.style={rectangle, draw=black!240, fill=purple!15, very thick, minimum size=7mm},				
		]


		\node[squaredBLACK2](R1) at (-4, 2.5) { $\tau1$, $\cancel{4}$ \textcolor{red}{$2$}};
		\node[squaredBLACK](R2) at (4, 2.5) { $\tau2$, $4$};

        \node[] at (2.15, 1.5) { $1$};
        \node[] at (-2, 1.5) { $\cancel{3}$ \textcolor{red}{$2$}};

		\node[sqORA](I1) at (0, 1) { $\tau3$, $\cancel{4}$ \textcolor{red}{$3$}};
        \node[] at (0.5, 0) { $\cancel{4}$ \textcolor{red}{$3$}};
		\node[] at (-4.35, 1.25) { $\cancel{1}$ \textcolor{red}{$0$}};

		\node[] at (4.25, 1.25) { $3$};

		\node[sqGREEN](F1) at (0, -2) { $\tau4$, $\cancel{4}$ \textcolor{red}{$3$}};

		\tikzset{thick edge/.style={-, black, fill=none, thick, text=black}}
		\tikzset{thick arc/.style={->, black, fill=black, thick, >=stealth, text=black}}

		\draw[line width=0.4mm, black,-> ] (R1) to[out=-90,in=90] (I1);
		\draw[line width=0.4mm, black,-> ] (R2) to[out=-90,in=90] (I1);
		);

		\draw[line width=0.4mm, black,-> ] (I1) to[out=-90,in=90] (F1);

		\draw[line width=0.4mm, black,-> ] (R1) to[out=-90,in=90] (F1);
		\draw[line width=0.4mm, black,-> ] (R2) to[out=-90,in=90] (F1);

		\end{tikzpicture}
		}
    }
    \caption{\footnotesize{How price adjustments can partly mitigate shocks. Here $\tau1$ and $\tau2$ are producing the same goods, and so can be substituted for each other.}}
    \label{fig:flexible}
  \end{figure}

A second example is illustrated in Figure \ref{fig:chips_medium_run}. Suppose labor is priced at $1$, one unit of labor is needed to make one unit of the intermediate good $R1$, $1$ unit of $R1$ and no units of labor are needed to make final good $F1$, while $1/9$th of a unit of $R1$ and $8/9$ths of a unit of labor are needed to make a unit of final good $F2$. Thus both $F1$ and $F2$ are priced at $1$. Suppose there are $10$ units of labor, and at these prices, labor demands $9$ units of good $F2$ and $1$ unit of $F1$. Thus the overall production by $F2$ is much more valuable than the overall production by $F1$: $F2$'s output contributes $9$ to a GDP of $10$, while $F1$'s output contributes $1$.

Consider now a 10\% total factor productivity shock to $R1$. In the short run, proportional rationing leads the output of both final goods to be reduced by 10\%, and GDP is correspondingly reduced by 10\% (achieving the short run upper bound). However, if we relax the proportional rationing constraint and allow intermediate good to flow to its most valuable use, then production of good $F2$ will not be affected at all, while output of $F1$ will decrease by 20\%. However, because production of good $F2$ is substantially more valuable, the reduction in GDP is now just 2\%.

For a real world example where such a reallocation would have been valuable, consider the semiconductor / computer chip crisis. In the short run the production of very valuable downstream goods like cars was disrupted, along with many other consumer goods. With flexible prices, in the medium run, the disruptions should become concentrated on less valuable consumer goods (like, for example, cheap toys). The extent to and speed with which this happened is an empirical question of interest.

Another illustrative example along these lines is the rolling electricity blackouts in California in the summer of 2020. This was applied without discrimination (save some emergency services and legal restrictions) and so fits our proportional rationing assumption. However, in anticipation of further future blackouts, more complicated contracts have emerged that result in priorities in rationing and corresponding differences in prices, with electricity being allocated to its highest value uses when there are shortages. Again, it is an empirical question regarding the extent to which efficient contracts have been put in place.

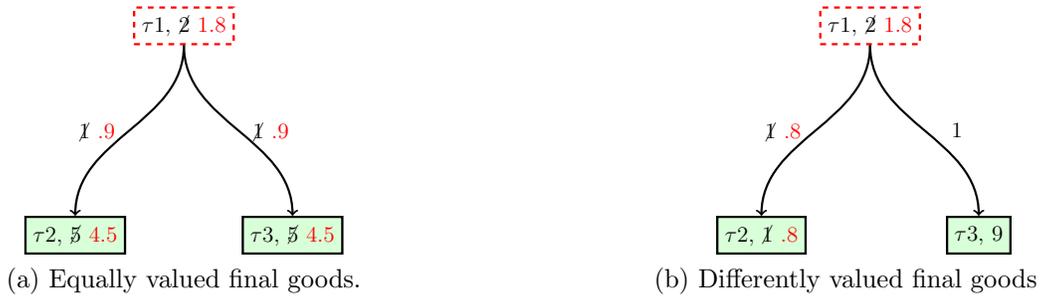
\begin{figure}[!ht]
    \subfloat[Equally valued final goods.]
    {
\centering
\resizebox{3in}{1.3in}{		
\begin{tikzpicture}[
		sqRED/.style={rectangle, draw=black!240, fill=red!15, very thick, minimum size=7mm},
		sqBLUE/.style={rectangle, draw=black!240, fill=blue!15, very thick, minimum size=7mm},
		roundnode/.style={ellipse, draw=black!240, fill=green!15, very thick, dashed, minimum size=7mm},
		sqGREEN/.style={rectangle, draw=black!240, fill=green!15, very thick, minimum size=7mm},
		roundRED/.style={ellipse, draw=black!240, fill=red!15, very thick, dashed, minimum size=7mm},
		roundREDb/.style={ellipse, draw=black!240, fill=red!15, very thick, minimum size=7mm},
		roundYELL/.style={ellipse, draw=black!240, fill=yellow!20, very thick, dashed, minimum size=7mm},
		sqYELL/.style={rectangle, draw=black!240, fill=yellow!20, very thick, minimum size=7mm},
		roundORA/.style={ellipse, draw=black!240, fill=orange!15, very thick, dashed, minimum size=7mm},		
		sqORA/.style={rectangle, draw=black!240, fill=orange!15, very thick, minimum size=7mm},
		squaredBLACK/.style={rectangle, draw=black!240, fill=white!15, very thick, minimum size=7mm},	
		squaredBLACK2/.style={rectangle, draw=red!240, fill=white!15, very thick, dashed, minimum size=7mm},	
		roundBROWN/.style={ellipse, draw=black!240, fill=brown!15, very thick, dashed, minimum size=7mm},
		sqBROWN/.style={rectangle, draw=black!240, fill=brown!15, very thick, minimum size=7mm},		
        sqPURPLE/.style={rectangle, draw=black!240, fill=purple!15, very thick, minimum size=7mm},				
		]

        \node[] at (5, 0) {  };
        \node[] at (-5, 0) {  };

		\node[squaredBLACK2](R1) at (0, 4) { $\tau1$, $\cancel{2}$ \textcolor{red}{$1.8$}};

        \node[] at (1.6, 2) {  $\cancel{1}$ \textcolor{red}{$.9$}};
        \node[] at (-1.6, 2) { $\cancel{1}$ \textcolor{red}{$.9$}};

		\node[sqGREEN](F1) at (-2, 0) { $\tau2$, $\cancel{5}$ \textcolor{red}{$4.5$}};
		\node[sqGREEN](F2) at (2, 0) { $\tau3$, $\cancel{5}$ \textcolor{red}{$4.5$}};


		\tikzset{thick edge/.style={-, black, fill=none, thick, text=black}}
		\tikzset{thick arc/.style={->, black, fill=black, thick, >=stealth, text=black}}

		\draw[line width=0.4mm, black,-> ] (R1) to[out=-90,in=90] (F2);
		\draw[line width=0.4mm, black,-> ] (R1) to[out=-90,in=90] (F1);


%

		\end{tikzpicture}
}
	}
    \hfill
    \subfloat[Differently valued final goods]
    {
     \centering
\hspace{0.2in}
\resizebox{3in}{1.3in}{		

\begin{tikzpicture}[
		sqRED/.style={rectangle, draw=black!240, fill=red!15, very thick, minimum size=7mm},
		sqBLUE/.style={rectangle, draw=black!240, fill=blue!15, very thick, minimum size=7mm},
		roundnode/.style={ellipse, draw=black!240, fill=green!15, very thick, dashed, minimum size=7mm},
		sqGREEN/.style={rectangle, draw=black!240, fill=green!15, very thick, minimum size=7mm},
		roundRED/.style={ellipse, draw=black!240, fill=red!15, very thick, dashed, minimum size=7mm},
		roundREDb/.style={ellipse, draw=black!240, fill=red!15, very thick, minimum size=7mm},
		roundYELL/.style={ellipse, draw=black!240, fill=yellow!20, very thick, dashed, minimum size=7mm},
		sqYELL/.style={rectangle, draw=black!240, fill=yellow!20, very thick, minimum size=7mm},
		roundORA/.style={ellipse, draw=black!240, fill=orange!15, very thick, dashed, minimum size=7mm},		
		sqORA/.style={rectangle, draw=black!240, fill=orange!15, very thick, minimum size=7mm},
		squaredBLACK/.style={rectangle, draw=black!240, fill=white!15, very thick, minimum size=7mm},	
		squaredBLACK2/.style={rectangle, draw=red!240, fill=white!15, very thick, dashed, minimum size=7mm},	
		roundBROWN/.style={ellipse, draw=black!240, fill=brown!15, very thick, dashed, minimum size=7mm},
		sqBROWN/.style={rectangle, draw=black!240, fill=brown!15, very thick, minimum size=7mm},		
        sqPURPLE/.style={rectangle, draw=black!240, fill=purple!15, very thick, minimum size=7mm},				
		]

        \node[] at (5, 0) {  };
        \node[] at (-5, 0) {  };

		\node[squaredBLACK2](R1) at (0, 4) { $\tau1$, $\cancel{2}$ \textcolor{red}{$1.8$}};

        \node[] at (1.6, 2) {  ${1}$ };
        \node[] at (-1.6, 2) { $\cancel{1}$ \textcolor{red}{$.8$}};

		\node[sqGREEN](F1) at (-2, 0) { $\tau2$, $\cancel{1}$ \textcolor{red}{$.8$}};
		\node[sqGREEN](F2) at (2, 0) { $\tau3$, ${9}$ };


		\tikzset{thick edge/.style={-, black, fill=none, thick, text=black}}
		\tikzset{thick arc/.style={->, black, fill=black, thick, >=stealth, text=black}}

		\draw[line width=0.4mm, black,-> ] (R1) to[out=-90,in=90] (F2);
		\draw[line width=0.4mm, black,-> ] (R1) to[out=-90,in=90] (F1);
	

		\end{tikzpicture}
		}
    }
    \caption{{\footnotesize The benefits of pricing flexibility. If the downstream goods have equal values the short run and medium run outcomes are the same and 10\% of output is lost in the medium run (Panel (a)). However, the downstream goods have different values, then the adjustment in the medium run is helpful as it can direct the disruption down the least valuable path. Now in the medium run only a 2\% of output is lost, versus 10\% in the short run that instead propagates equally down both paths (Panel (b)).}}
    \label{fig:chips_medium_run}
  \end{figure}

The output loss in the medium run (with flexible prices) is weakly less than the output loss in the short run, because we are relaxing the proportional rationing constraint (and effectively, reoptimizing). However, there are situations in which there is no difference between the medium and short run. For example, if the conditions of Proposition \ref{prop:general_industry_shocks} hold, such that there is no technological diversity and there are industry specific shocks and each shocked technology is used---directly or indirectly---in the production of only one final good, then there is no scope for reducing the impact of the shock in the medium run.   There is no scope for correlating the downstream instances of the disruption nor any scope for assigning the disruptions disproportionately to supply chains of lower-value final goods. On the other hand, it is possible to construct examples in which price flexibility makes a big difference. Specifically, the ratio of lost output in the short run to lost output in the medium run with price flexibility is unbounded (as we show in Supplementary Appendix \ref{sec:Potential_price_flexibility}).


In this section we have considered the medium run relaxation of the short run impact of a shock that allows prices to adjust. 
There are further variations that live between
the medium and long run.  For example, one can relax the no increased output constraints for some technologies. In a somewhat different setting, this is the problem studied in \cite{carvalhobottleneck2024}.
They define exogenous capacity constraints on outputs and good flows, and allow production of used technologies to be scaled up subject to these constraints. This can result in some firms, both upstream and downstream of a disruption, increasing their outputs and it is no longer possible to trace through the impact of a shock with a simple and intuitive algorithm like the shock propagation algorithm we consider here.

\section{Potential impact of price flexibility}\label{sec:Potential_price_flexibility}

In other cases, the ability for prices to adjust can make a big difference. One way to measure this difference is to consider the ratio of lost output when prices can adjust, to lost output when prices cannot adjust (and the proportional rationing constraint must be satisfied). Given an initial economy $\mathcal{E}$ and associated pre-shock equilibrium, and letting $y^{ip}_{\tau}$ be a solution to the minimum disruption problem (with inflexible prices), we define the \emph{loss to price rigidity (LPR)} as $$LPR(\mathcal{E},T^{Shocked},\lambda)=\frac{\sum_{\tau:O(\tau)\in F} (y_{\tau}^*-y^{ip}_{\tau}) }{\sum_{\tau:O(\tau)\in F}(y_{\tau}^*- y^{fp}_{\tau})}\geq 1.$$

\begin{proposition}\label{prop:unbounded_LPR}
The potential loss to price rigidity is unbounded: there exists a sequence of economies $\{\mathcal{E}_t\}_t$ such that $\lim_{t\rightarrow \infty}LPR(\mathcal{E}_t,T^{Shocked},\lambda)=\infty$.
\end{proposition}

We prove Proposition \ref{prop:unbounded_LPR} by example. Consider an economy $\mathcal{E}_t$ with initial equilibrium flows as depicted in Figure \ref{fig:LPR_Example} for $t=3$. For this example it is helpful to define intermediate good levels. There are two levels of intermediate goods. Level 1 intermediate goods use only labor as an input. Level 2 intermediate good use labor and intermediate level 1 goods as inputs. Finally, there are final goods which use intermediate level 2 goods and labor as inputs.

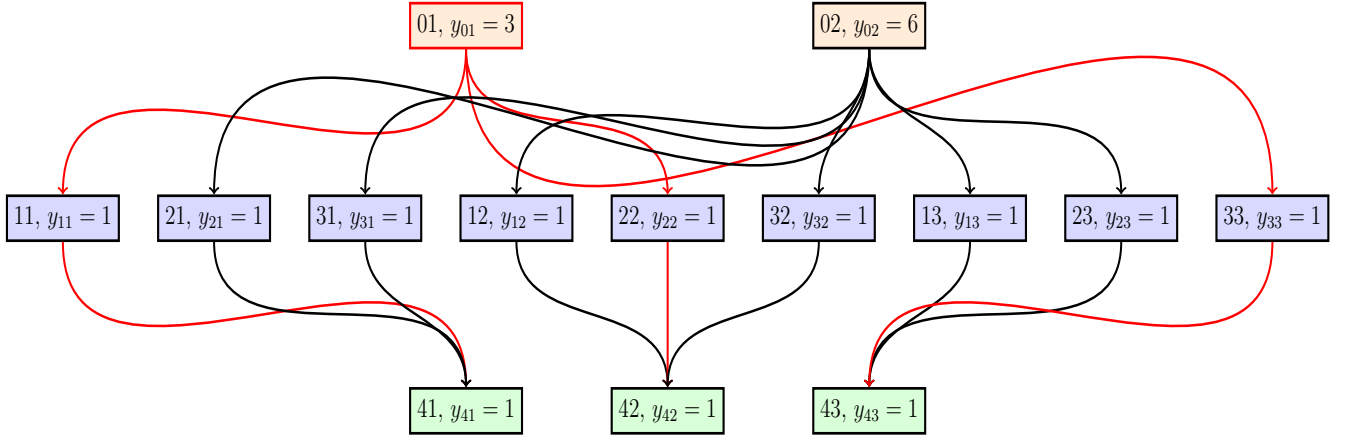
\begin{figure}[t!]
	\begin{center}
\resizebox{7in}{3.5in}{		
\begin{tikzpicture}[
		sqRED/.style={rectangle, draw=red!240, fill=orange!15, very thick, minimum size=7mm},
		sqBLUE/.style={rectangle, draw=black!240, fill=blue!15, very thick, minimum size=7mm},
		sqGREEN/.style={rectangle, draw=black!240, fill=green!15, very thick, minimum size=7mm},
		sqORA/.style={rectangle, draw=black!240, fill=orange!15, very thick, minimum size=7mm},
		]

		\node[sqRED](01) at (-4, 2) { 01, $y_{01}=3$};

        \node[sqORA](02) at (4, 2) { 02, $y_{02}=6$};
 		
        \node[sqBLUE](11) at (-12, -1) { 11, $y_{11}=1$};
        		
        \node[sqBLUE](21) at (-9, -1) { 21, $y_{21}=1$};
        		
        \node[sqBLUE](31) at (-6, -1) { 31, $y_{31}=1$};
        		
        \node[sqBLUE](12) at (-3, -1) { 12, $y_{12}=1$};
        		
        \node[sqBLUE](22) at (0, -1) { 22, $y_{22}=1$};

        \node[sqBLUE](32) at (3, -1) { 32, $y_{32}=1$};
        		
        \node[sqBLUE](13) at (6, -1) { 13, $y_{13}=1$};
        		
        \node[sqBLUE](23) at (9, -1) { 23, $y_{23}=1$};
        		
        \node[sqBLUE](33) at (12, -1) { 33, $y_{33}=1$};

        \node[sqGREEN](41) at (-4, -4) { 41, $y_{41}=1$};
		\node[sqGREEN](42) at (0, -4) { 42, $y_{42}=1$};
		\node[sqGREEN](43) at (4, -4) { 43, $y_{43}=1$};

		\tikzset{thick edge/.style={-, black, fill=none, thick, text=black}}
		\tikzset{thick arc/.style={->, black, fill=black, thick, >=stealth, text=black}}
		
	    \draw[line width=0.4mm, red,-> ] (01) to[out=-90,in=90] (11);
		\draw[line width=0.4mm, red,-> ] (01) to[out=-90,in=90] (22);
		\draw[line width=0.4mm, red,-> ] (01) to[out=-90,in=90] (33);

	    \draw[line width=0.4mm, black,-> ] (02) to[out=-90,in=90] (12);
		\draw[line width=0.4mm, black,-> ] (02) to[out=-90,in=90] (13);
		\draw[line width=0.4mm, black,-> ] (02) to[out=-90,in=90] (21);
	    \draw[line width=0.4mm, black,-> ] (02) to[out=-90,in=90] (23);
		\draw[line width=0.4mm, black,-> ] (02) to[out=-90,in=90] (31);
		\draw[line width=0.4mm, black,-> ] (02) to[out=-90,in=90] (32);

        \draw[line width=0.4mm, red,-> ] (11) to[out=-90,in=90] (41);
        \draw[line width=0.4mm, black,-> ] (21) to[out=-90,in=90] (41);
        \draw[line width=0.4mm, black,-> ] (31) to[out=-90,in=90] (41);

        \draw[line width=0.4mm, black,-> ] (12) to[out=-90,in=90] (42);
        \draw[line width=0.4mm, red,-> ] (22) to[out=-90,in=90] (42);
        \draw[line width=0.4mm, black,-> ] (32) to[out=-90,in=90] (42);

        \draw[line width=0.4mm, black,-> ] (13) to[out=-90,in=90] (43);
        \draw[line width=0.4mm, black,-> ] (23) to[out=-90,in=90] (43);
        \draw[line width=0.4mm, red,-> ] (33) to[out=-90,in=90] (43);
		\end{tikzpicture}
		}
		\caption{\label{fig:LPR_Example}{\footnotesize Initial equilibrium of economy $\mathcal{E}_3$, with red lines indicating flows that will be disrupted with inflexible pricing following a shock to technology $01$.}}					
	\end{center}
\end{figure}

Intermediate level 1 consists of a single type of good, good $0$, produced by two countries $1$ and $2$. Both countries require 1 unit of labor to produce one unit of output. Intermediate level 2 consists of $t$ goods. All $t$ goods are produced in all $t$ countries. Each good intermediate level $1$ good $i$ requires one unit of good $0$ and one unit of labor. In equilibrium, all goods are produced in all countries and for the production of good $i$ by country $i$, and for $i=1,\ldots,t$, input good $0$ is sourced from country $1$, while otherwise input good $0$ is sourced from country $2$. Finally, each country also produces the unique final good. Production of this final good in all countries requires one unit of all intermediate goods $1, \ldots, t$ and one unit of labor to produce one unit of output.

We assume that no transportation costs for any goods except labor. Country $1$, is endowed with $2t+1$ units of labor, country 2 is endowed with $t^2+1$ units of labor and the other countries are endowed with $t+1$ units of labor. In equilibrium: Country $1$ produces $t$ units of good $0$, one unit of goods $1,\ldots,t$ and one unit of the final good, Country $2$ produces $t(t-1)$ units of good $0$, one unit of goods $1,\ldots,t$ and one unit of the final good. The other countries produce one unit of goods $1,\ldots,t$ and one unit of the final good. There is then an equilibrium in which all final goods are produced using only domestically produced inputs, and all labor works in its home country.

\begin{figure}[t!]
	\begin{center}
\resizebox{7in}{3.5in}{		
\begin{tikzpicture}[
		sqRED/.style={rectangle, draw=red!240, fill=orange!15, very thick, minimum size=7mm},
		sqBLUE/.style={rectangle, draw=black!240, fill=blue!15, very thick, minimum size=7mm},
		sqGREEN/.style={rectangle, draw=black!240, fill=green!15, very thick, minimum size=7mm},
		sqORA/.style={rectangle, draw=black!240, fill=orange!15, very thick, minimum size=7mm},
		]

		\node[sqRED](01) at (-4, 2) { 01, $y_{01}=3$};

        \node[sqORA](02) at (4, 2) { 02, $y_{02}=6$};
 		
        \node[sqBLUE](11) at (-12, -1) { 11, $y_{11}=1$};
        		
        \node[sqBLUE](21) at (-9, -1) { 21, $y_{21}=1$};
        		
        \node[sqBLUE](31) at (-6, -1) { 31, $y_{31}=1$};
        		
        \node[sqBLUE](12) at (-3, -1) { 12, $y_{12}=1$};
        		
        \node[sqBLUE](22) at (0, -1) { 22, $y_{22}=1$};

        \node[sqBLUE](32) at (3, -1) { 32, $y_{32}=1$};
        		
        \node[sqBLUE](13) at (6, -1) { 13, $y_{13}=1$};
        		
        \node[sqBLUE](23) at (9, -1) { 23, $y_{23}=1$};
        		
        \node[sqBLUE](33) at (12, -1) { 33, $y_{33}=1$};

        \node[sqGREEN](41) at (-4, -4) { 41, $y_{41}=1$};
		\node[sqGREEN](42) at (0, -4) { 42, $y_{42}=1$};
		\node[sqGREEN](43) at (4, -4) { 43, $y_{43}=1$};

		\tikzset{thick edge/.style={-, black, fill=none, thick, text=black}}
		\tikzset{thick arc/.style={->, black, fill=black, thick, >=stealth, text=black}}
		

	    \draw[line width=0.4mm, black,-> ] (02) to[out=-90,in=90] (12);
		\draw[line width=0.4mm, black,-> ] (02) to[out=-90,in=90] (13);
		\draw[line width=0.4mm, black,-> ] (02) to[out=-90,in=90] (21);
	    \draw[line width=0.4mm, black,-> ] (02) to[out=-90,in=90] (23);
		\draw[line width=0.4mm, black,-> ] (02) to[out=-90,in=90] (31);
		\draw[line width=0.4mm, black,-> ] (02) to[out=-90,in=90] (32);

        \draw[line width=0.4mm, black,-> ] (21) to[out=-90,in=90] (42);
        \draw[line width=0.4mm, black,-> ] (31) to[out=-90,in=90] (43);

        \draw[line width=0.4mm, black,-> ] (12) to[out=-90,in=90] (42);
        \draw[line width=0.4mm, black,-> ] (32) to[out=-90,in=90] (42);

        \draw[line width=0.4mm, black,-> ] (13) to[out=-90,in=90] (43);
        \draw[line width=0.4mm, black,-> ] (23) to[out=-90,in=90] (43);
		\end{tikzpicture}
		}
		\caption{\label{fig:LPR_Example2}\footnotesize{ Rerouted flows with flexible pricing following a shock to technology $01$.}}					
	\end{center}
\end{figure}
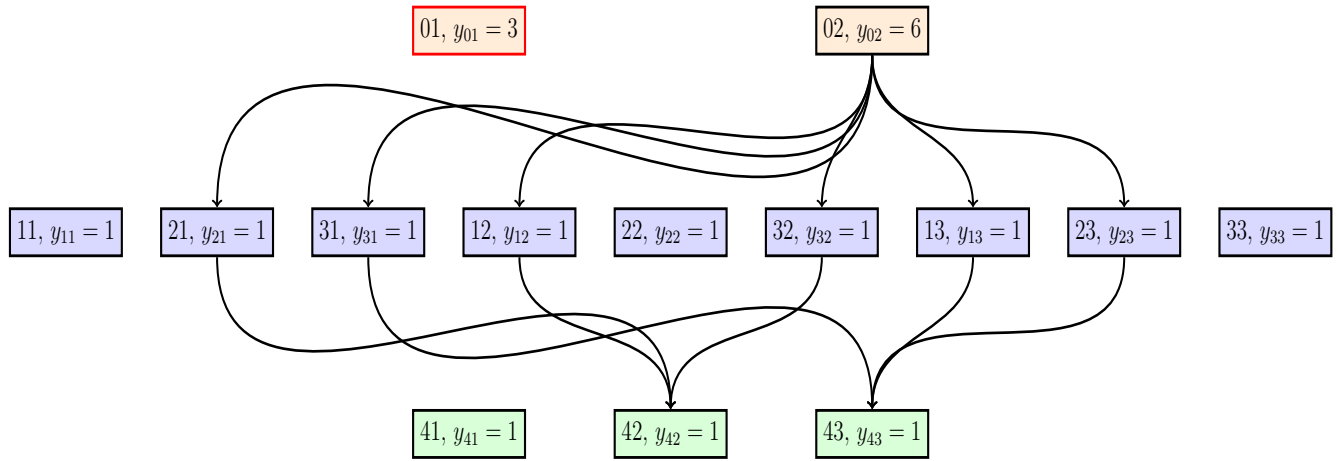

We consider a shock to the output of good $0$ production in country $1$ that reduces its output to $(1-\lambda)t$ units. With inflexible prices (such that there is proportional rationing), by Proposition \ref{prop:cut}, this reduces the aggregate output of final good production to $(1-\lambda)$ of its initial level (i.e., it achieves the bound from Proposition \ref{prop:shocked_GDP}). Intuitively, output of good $i$ in country $i$ is reduced to $(1-\lambda)$ of its initial level, because it sources good $0$ from country $1$, and as all of good $i$ used by country $i$ in the production of its final good is sourced domestically, output of its final good production also reduced to $(1-\lambda)$ of its initial level.

In contrast, with flexible prices, it is possible to maintain final good production in countries $2,\ldots,t$. As shown in Figure \ref{fig:LPR_Example2} for $t=3$, this can be done by rerouting the output of input good $i$ in country $1$ to country $i$, for $i=2,\dots, t$, for use in country $i$'s final good production. Thus, letting $\tau_{01}$ be the technology used to produce good $0$ in country $1$, $LPR(\mathcal{E}_t,\tau_{01},\lambda)=\frac{\lambda}{\lambda/t},$ and so $$\lim_{t\rightarrow \infty}LPR(\mathcal{E}_t,\tau_{01},\lambda)=\lim_{t\rightarrow \infty}t\left(\frac{\lambda}{\lambda}\right)=\infty.$$

However, if instead both level $0$ intermediate good producers has been shocked, then by Proposition \ref{prop:general_industry_shocks} we would have $LPR(\mathcal{E}_t,\{\tau_{01},\tau_{02}\},\lambda)=\lambda/\lambda=1,$ for all $t$.

\section{Power}

This appendix contains additional material to support the section on power in the main paper (Section \ref{sec:power}).

\subsection{A More Involved Example of Power}\label{sec:involvedexample}

Consider five countries labeled $1-5$, three raw materials (intermediate goods that only use labor in their production) $R1$ to $R3$, six intermediate goods $I1$ to $I6$ and four final goods $F1$ to $F4$. The initial equilibrium flows are shown in Figure \ref{fig:power_calculation}. Labor is immobile and we index each active technology by the good it produces and the country in which it is located. Countries 3-5 each operate a single technology, but both Country 1 and Country 2 operate multiple technologies. 
Units of output are defined such that one unit of labor is required to produce one unit of output (which can be done when each good is produced by one technology (possibly in different countries, but with identical recipes). 

Country 1 and Country 2 are both endowed with $100$ units of labor, which are fully deployed across their active technologies (illustrated by the respective shaded regions in Figure \ref{fig:power_calculation}). The GDP of Country 1 is therefore $100 w_1$, where $w_1$ is wage in country $1$, and correspondingly, the GDP of Country 2 is $100 w_2$.

\begin{figure}[!ht]
       \subfloat[Initial Equilibrium]
    {
\centering
\resizebox{6.5in}{5in}{		
    \begin{tikzpicture}[ 
        he/.style={draw, thick, rounded corners},
		sqRED/.style={rectangle, draw=black!240, fill=red!25, very thick, minimum size=7mm},
		sqBLUE/.style={rectangle, draw=black!240, fill=blue!15, very thick, minimum size=7mm},
		sqPURPLE/.style={rectangle, draw=black!240, fill=purple!15, very thick, minimum size=7mm},
        sqGREEN/.style={rectangle, draw=black!240, fill=green!15, very thick, minimum size=7mm},
		roundRED/.style={ellipse, draw=black!240, fill=red!15, very thick, dashed, minimum size=7mm},
		roundREDb/.style={ellipse, draw=black!240, fill=red!15, very thick, minimum size=7mm},
		roundYELL/.style={ellipse, draw=black!240, fill=yellow!20, very thick, dashed, minimum size=7mm},
		sqYELL/.style={rectangle, draw=black!240, fill=yellow!25, very thick, minimum size=7mm},
		roundORA/.style={ellipse, draw=black!240, fill=orange!20, very thick, dashed, minimum size=7mm},		
		sqORA/.style={rectangle, draw=black!240, fill=orange!15, very thick, minimum size=7mm},
		squaredBLACK/.style={rectangle, draw=black!240, fill=white!15, very thick, minimum size=7mm},
squaredBLACK2/.style={rectangle, draw=red!240, fill=white!15, very thick, dashed, minimum size=7mm},		
		roundBROWN/.style={ellipse, draw=black!240, fill=brown!15, very thick, dashed, minimum size=7mm},
		sqBROWN/.style={rectangle, draw=black!240, fill=brown!25, very thick, minimum size=7mm},		
        sqPURPLE/.style={rectangle, draw=black!240, fill=purple!15, very thick, minimum size=7mm},				
        sqGRAY/.style={rectangle, draw=black!240, fill=gray!25, very thick, minimum size=7mm},	 
        sqGRAY2/.style={rectangle, draw=black!240, fill=gray!15, very thick, minimum size=7mm},	 
        sqWHITE/.style={rectangle, draw=black!240, fill=white!10, very thick, minimum size=7mm},	 
        sqBLACK/.style={rectangle, draw=black!240, fill=black!25, very thick, minimum size=7mm},	 sqPINK/.style={rectangle, draw=black!240, fill=pink!25, very thick, minimum size=7mm},	sqTEAL/.style={rectangle, draw=black!240, fill=teal!25, very thick, minimum size=7mm},				
		]

\node[](HG1A) at (6.75, 5.25) {};
\node[](HG1B) at (-0.3, -10.25) {};
\node [he, fill=cyan!5, fit=(HG1A) (HG1B)] {};

\node[](HG2A) at (-5.85, 1.75) {};
\node[](HG2B) at (-1.1, -10.25) {};
\node [he, fill=lime!5, fit=(HG2A) (HG2B)] {};

		\node[sqTEAL](R1-3) at (-2.25, 5) {$\tau^{R1-3}, 4$};	
        \node[sqBLUE](R2-1) at (2.5, 5) {$\tau^{R2-1}, 11$};
        \node[sqGREEN](R3-5) at (8.1, 5) {$\tau^{R3-4}, 11$};

        \node[] at (3, 2.5) { $11$};
        
        \node[] at (-0.75, 3.5) { $2$};
		\node[] at (-3.75, 3.5) { $2$};
        \node[] at (7.5, 1.5) { $11$};

		\node[sqORA](I3-1) at (2.5, 0) {$\tau^{I3-1},11$};

        \node[] at (0, -2.5) { $3$};
        \node[] at (2, -2.5) { $2$};
        \node[] at (3.25, -2.5) { $6$};

		\node[sqYELL](I1-2) at (-5, 1.5) {$\tau^{I1-2}, 20$};
        \node[sqPURPLE](I2-1) at (0.5, 1.5) {$\tau^{I2-1},20$};
        \node[sqBROWN](I4-1) at (6, -2.5) {$\tau^{I4-1},4$};
        
        \node[] at (-7.75, -6) { $10$};
		\node[] at (-4, -5.8) { $10$};
        \node[] at (-4.75, -4) { $10$};
		\node[] at (-4, -4) { $10$};
        \node[] at (6.25, -6.25) { $4$};
        
        \node[sqRED](I5-2) at (-2, -5.5) {$\tau^{I5-2},15$};
        \node[sqRED](I5-1) at (0.5, -5.5) {$\tau^{I5-1},10$};
        \node[sqPINK](I6-1) at (3.5, -5.5) {$\tau^{I6-1},6$};

        \node[] at (-1.75, -6.5) { $15$};
        \node[] at (0.75, -6.5) { $10$};
        
        \node[] at (3.75, -8.75) { $1$};
        \node[] at (1.25, -8.75) { $2$};
        \node[] at (-1.25, -8.75) { $3$};

       \node[sqBLACK](F1-6) at (-8, -10) {$\tau^{F1-5},35$};

        \node[sqBLACK](F1-2) at (-5, -10) {$\tau^{F1-2},35$};

         \node[sqGRAY](F2-2) at (-2, -10) {$\tau^{F2-2},30$};

         \node[sqGRAY](F2-1) at (0.5, -10) {$\tau^{F2-1},20$};

        \node[sqGRAY2](F3-1) at (3.5, -10) {$\tau^{F3-1},10$};

         \node[sqWHITE](F4-1) at (6, -10) {$\tau^{F4-1},10$};

		\tikzset{thick edge/.style={-, black, fill=none, thick, text=black}}
		\tikzset{thick arc/.style={->, black, fill=black, thick, >=stealth, text=black}}

		\draw[line width=0.4mm, black,-> ] (R1-3) to[out=-90,in=90] (I1-2);
		\draw[line width=0.4mm, black,-> ] (R1-3) to[out=-90,in=90] (I2-1);
		\draw[line width=0.4mm, black,-> ] (I1-2) to[out=-90,in=90] (F1-6);
		\draw[line width=0.4mm, black,-> ] (I1-2) to[out=-90,in=90] (F1-2);
		\draw[line width=0.4mm, black,-> ] (I2-1) to[out=-90,in=90] (F1-6);
		\draw[line width=0.4mm, black,-> ] (I2-1) to[out=-90,in=90] (F1-2);

        \draw[line width=0.4mm, black,-> ] (R2-1) to[out=-90,in=90] (I3-1);
        \draw[line width=0.4mm, black,-> ] (I3-1) to[out=-90,in=90] (I5-2);
		\draw[line width=0.4mm, black,-> ] (I3-1) to[out=-90,in=90] (I5-1);
		\draw[line width=0.4mm, black,-> ] (I3-1) to[out=-90,in=90] (I6-1);
		\draw[line width=0.4mm, black,-> ] (I5-2) to[out=-90,in=90] (F2-2);
		\draw[line width=0.4mm, black,-> ] (I5-1) to[out=-90,in=90] (F2-1);
		\draw[line width=0.4mm, black,-> ] (I6-1) to[out=-90,in=90] (F2-2);
		\draw[line width=0.4mm, black,-> ] (I6-1) to[out=-90,in=90] (F2-1);
		\draw[line width=0.4mm, black,-> ] (I6-1) to[out=-90,in=90] (F3-1);

        \draw[line width=0.4mm, black,-> ] (R3-5) to[out=-90,in=90] (I4-1);
        \draw[line width=0.4mm, black,-> ] (I4-1) to[out=-90,in=90] (F4-1);

		\end{tikzpicture}
		}
	}
    \caption{\footnotesize{There are five countries labeled 1 to 5.  Country 1 is in the blue rectangle and Country 2 is in the yellow rectangle.   Each technology is indexed by the good it produces and then the country in which it is located. So technology $I1-2$ produces intermediate good $I1$ and is located in country $2$. Countries 3-5 operate a single technology. However Countries $1$ and $2$ each operate multiple technologies, as illustrated by the shaded areas.}}
    \label{fig:power_calculation}
  \end{figure}
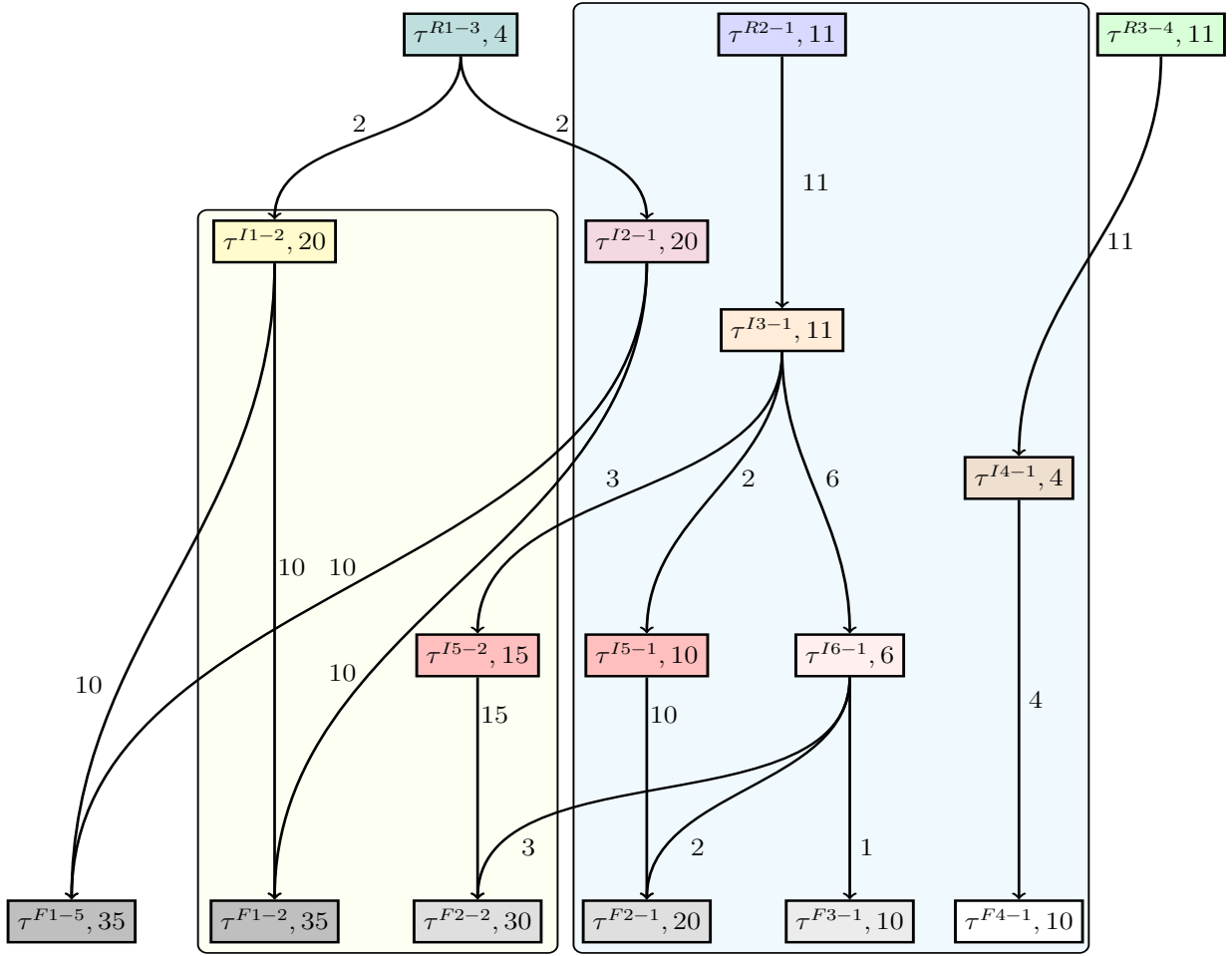


Consider the power Country $1$ has to disrupt production in Country $2$. By Proposition \ref{prop:disruption_frontier} we just need to consider individual technology disruptions.  And, as we remarked above, not even all of these need to be considered. In this case it is sufficient to consider individual disruptions to $\{\tau^{I2-1},\tau^{I3-1},\tau^{I6-1}\}$. As Country 1 operates 9 technologies and several of these technologies have multiple flows in and/or out of them there are many different combinations of disruptions that Country 1 could potentially take and so this simplifies matters considerably. 

First consider a disruption to $\tau^{I2-1}$. Reducing the flow $\tau^{I2-1}\tau^{F1-2}$ to $0$ reduces Country $1$'s GDP by $10 w_1$ and reduces Country $2$'s GDP by $45 w_2$. As $\tau^{F1-2}$ is no longer able to source input $I2$  it shuts down, leaving $35$ units of labor in Country 2 idle, and this reduces $\tau^{F1-2}$'s demand for input $I1$ leaving a further $10$ units of labor idle in Country 2. Thus this disruption grants Country $1$ power $4.5$ over Country $2$. 

An alternative disruption to $\tau^{I2-1}$ available to Country $1$ is to eliminate the flow $\tau^{I2-1}\tau^{F1-5}$. This will reduce the demand for $\tau^{I1-2}$'s output, causing Country 2 to lose $10w_2$ in GDP at a cost of $10w_1$ in GDP for Country 1. Thus this disruption is less effective, and as other disruptions to $\tau^{I2-1}$ are some convex combination of these two, the most power that disrupting $\tau^{I2-1}$ can give Country $1$ over Country $2$ is $4.5$. 

The second technology that Country $1$ might want to disrupt is $\tau^{I3-1}$. If Country 1 eliminated the flow $\tau^{I3-1}\tau^{I5-2}$, then propagating the shock both upstream and downstream, $15$ units of labor become idle in Country 1 at a cost to GDP of $15w_1$, while GDP in Country $2$ is reduced by $45 w_2$. This grants Country $1$ power $3$ over Country $2$ and is the most powerful disruption of technology $\tau^{I3-1}$ available to Country $1$.  

The final technology that Country $1$ might want to disrupt is $\tau^{I6-1}$. However, when the disruption to the flow $\tau^{I3-1}\tau^{I5-2}$ considered above is propagated it eliminates the flow $\tau^{I6-1}\tau^{F2-2}$ and disrupting either of these flows turn out to be equivalent. So again, this only gives Country $1$ power $3$ over Country $2$.

Thus, the power that Country $1$ has over Country $2$ is $\text{Power}_{12}=4.5$.

It is also instructive to consider the power that Country $2$ has over the other countries. It has $0$ power over Country 4 because it cannot impact its GDP at all. It has power $10$ over Country 5, because it can completely eliminate its GDP at a cost of $10\%$ of its own GDP by withholding the flow $\tau^{I1-2}\tau^{F1-5}$. The same disruption gives it power $2.5$  over Country 3, because by eliminating its flow $\tau^{I1-2} \tau^{F1-5}$ demand for $\tau^{R1-3}$'s output falls to $1$ causing a 25\% loss in GDP for Country 3 for the 10\% loss in own GDP. Finally, Country 2 has power $1$ over Country 1, which comes again from disrupting the flow $\tau^{I1-2}\tau^{F1-5}$. This is effectively no power, unless Country 2 is willing to forgo an equal percentage loss in GDP as it imposes on Country 2.

The example helps identify the characteristics of the supply network that grant Country $i$ power over Country $j$. Country $i$ will have a substantial amount of power over Country $j$ when: (i) Country $j$'s output is relatively highly concentrated in one supply chain (measured by its proportion of wage income associated with that supply chain); (ii) Country $i$ has the ability to disrupt this supply chain; and (iii) Country $i$ has relatively diversified production (measured by its proportion of wage income associated with other supply chains). 

An example of a powerful disruption is one in which a key good is relatively monopolized by one country, but 
the supply chain in which it is involved (up and/or downstream) lies nontrivially in another country.   

Note that it is possible for Countries $i$ and $j$ to both have substantial power over each other (potentially via completely separate parts of the supply network), and for the power relationship to be very asymmetric---so that one country has substantial power over the other but not vice versa.

\subsection{Strategic Power}\label{sec:strategicpower}

The definitions above work with non-strategic decisions at all flows except those by the country that is doing the disruption.  This does not provide scope for other countries responding to mitigate the damage or retaliate. While in the short run other countries may not have time to respond, by the medium run a response should be expected. Moreover, if the disruption has been anticipated by the target country then it might be in a position to immediately respond.
We now develop a definition of power that includes a strategic response.

To see some of the basic issues, consider the example in Figure \ref{fig:strategic}.   Here, by responding strategically with its intermediate good, Country $j$ can drastically blunt Country $i$'s power.

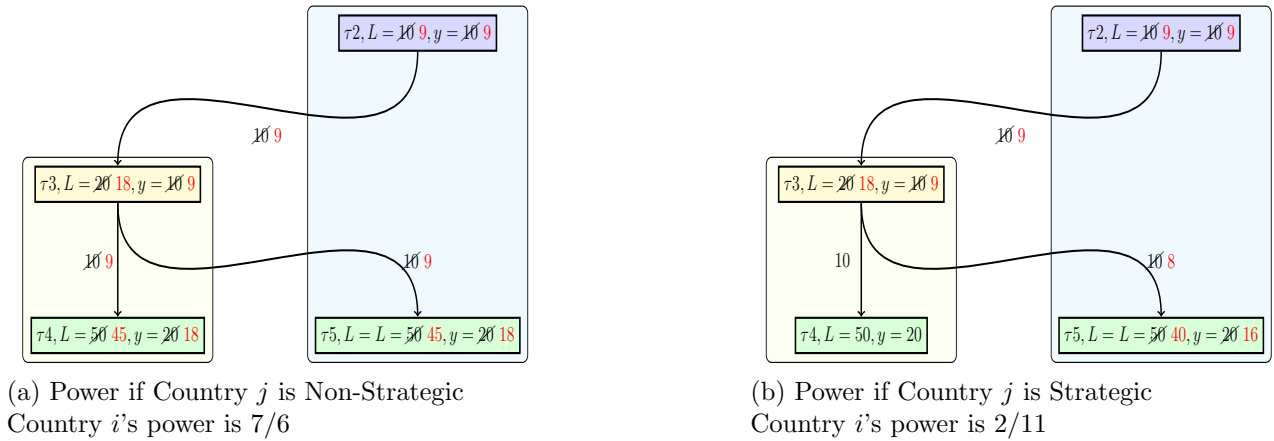
\begin{figure}[!ht]
    \subfloat[Power if Country $j$ is Non-Strategic \\
   $\hspace{1cm}$ Country $i$'s power is 7/6]
    {
\centering
\resizebox{2.7in}{1.92in}{		
    \begin{tikzpicture}[
sqRED/.style={rectangle, draw=black!240, fill=red!15, very thick, minimum size=7mm},
		sqREDb/.style={rectangle, draw=black!240, fill=red!15, very thick, dashed, minimum size=7mm},
		sqBLUE/.style={rectangle, draw=black!240, fill=blue!15, very thick, minimum size=7mm},
		sqBLUEb/.style={rectangle, draw=black!240, fill=blue!15, very thick, dashed, minimum size=7mm},
		roundnode/.style={ellipse, draw=black!240, fill=green!15, very thick, dashed, minimum size=7mm},
		sqGREEN/.style={rectangle, draw=black!240, fill=green!15, very thick, minimum size=7mm},
		sqGREENb/.style={rectangle, draw=black!240, fill=green!15, very thick, dashed, minimum size=7mm},
		roundRED/.style={ellipse, draw=black!240, fill=red!15, very thick, dashed, minimum size=7mm},
		roundREDb/.style={ellipse, draw=black!240, fill=red!15, very thick, minimum size=7mm},
		roundYELL/.style={ellipse, draw=black!240, fill=yellow!20, very thick, dashed, minimum size=7mm},
		sqYELL/.style={rectangle, draw=black!240, fill=yellow!20, very thick, minimum size=7mm},
		roundORA/.style={ellipse, draw=black!240, fill=orange!15, very thick, dashed, minimum size=7mm},		
		sqORA/.style={rectangle, draw=black!240, fill=orange!15, very thick, minimum size=7mm},
		sqORAb/.style={rectangle, draw=black!240, fill=orange!15, very thick, dashed, minimum size=7mm},
		sqMAG/.style={rectangle, draw=black!240, fill=magenta!5, very thick, minimum size=7mm},
		squaredBLACK/.style={rectangle, draw=black!240, fill=white!15, very thick, minimum size=7mm},	
		roundBROWN/.style={ellipse, draw=black!240, fill=brown!15, very thick, dashed, minimum size=7mm},
		sqBROWN/.style={rectangle, draw=black!240, fill=brown!15, very thick, minimum size=7mm},		
        ]

\node[](HG1A) at (-1.75, 0.25) {};
\node[](HG1B) at (-6.25, -3.25) {};
\node [he, draw=black!240, fill=lime!5, fit=(HG1A) (HG1B)] {};

\node[](HG2A) at (1.35, 3.25) {};
\node[](HG2B) at (6.65, -3.25) {};
\node [he, draw=black!240, fill=cyan!5, fit=(HG2A) (HG2B)] {};

        \node[sqBLUE](A2) at (4, 3) {$\tau2,  L=\cancel{10} \ \textcolor{red}{9},  y=\cancel{10} \ \textcolor{red}{9}$};
        \node[] at (0, 1) { $\cancel{10} \ \textcolor{red}{9}$};
		\node[] at (4, -1.5) { $\cancel{10} \ \textcolor{red}{9}$};
			
		\node[sqYELL](C1) at (-4, 0) {$\tau3, L=\cancel{20} \ \textcolor{red}{18}, y =\cancel{10} \ \textcolor{red}{9} $};
       \node[] at (-4.5, -1.5) { $\cancel{10} \ \textcolor{red}{9}$};

		
		\node[sqGREEN](F1) at (-4, -3) {$\tau4,L=\cancel{50} \ \textcolor{red}{45}, y=\cancel{20} \ \textcolor{red}{18}$};
	
		\node[sqGREEN](F2) at (4, -3) {$\tau5, L = L=\cancel{50} \ \textcolor{red}{45}, y=\cancel{20} \ \textcolor{red}{18}$};	
		
		\tikzset{thick edge/.style={-, black, fill=none, thick, text=black}}
		\tikzset{thick arc/.style={->, black, fill=black, thick, >=stealth, text=black}}
		
		\draw[line width=0.4mm, black,-> ] (A2) to[out=-90,in=90] (C1);
		
		\draw[line width=0.4mm, black,-> ] (C1) to[out=-90,in=90] (F1);
		\draw[line width=0.4mm, black,-> ] (C1) to[out=-90,in=90] (F2);

		
       
		\end{tikzpicture}
		}
        }
           \hfill
    \subfloat[Power if Country $j$ is Strategic \\
    \hspace{1cm} Country $i$'s power is 2/11]
    {
\centering
\resizebox{2.7in}{1.92in}{		
    \begin{tikzpicture}[
sqRED/.style={rectangle, draw=black!240, fill=red!15, very thick, minimum size=7mm},
		sqREDb/.style={rectangle, draw=black!240, fill=red!15, very thick, dashed, minimum size=7mm},
		sqBLUE/.style={rectangle, draw=black!240, fill=blue!15, very thick, minimum size=7mm},
		sqBLUEb/.style={rectangle, draw=black!240, fill=blue!15, very thick, dashed, minimum size=7mm},
		roundnode/.style={ellipse, draw=black!240, fill=green!15, very thick, dashed, minimum size=7mm},
		sqGREEN/.style={rectangle, draw=black!240, fill=green!15, very thick, minimum size=7mm},
		sqGREENb/.style={rectangle, draw=black!240, fill=green!15, very thick, dashed, minimum size=7mm},
		roundRED/.style={ellipse, draw=black!240, fill=red!15, very thick, dashed, minimum size=7mm},
		roundREDb/.style={ellipse, draw=black!240, fill=red!15, very thick, minimum size=7mm},
		roundYELL/.style={ellipse, draw=black!240, fill=yellow!20, very thick, dashed, minimum size=7mm},
		sqYELL/.style={rectangle, draw=black!240, fill=yellow!20, very thick, minimum size=7mm},
		roundORA/.style={ellipse, draw=black!240, fill=orange!15, very thick, dashed, minimum size=7mm},		
		sqORA/.style={rectangle, draw=black!240, fill=orange!15, very thick, minimum size=7mm},
		sqORAb/.style={rectangle, draw=black!240, fill=orange!15, very thick, dashed, minimum size=7mm},
		sqMAG/.style={rectangle, draw=black!240, fill=magenta!5, very thick, minimum size=7mm},
		squaredBLACK/.style={rectangle, draw=black!240, fill=white!15, very thick, minimum size=7mm},	
		roundBROWN/.style={ellipse, draw=black!240, fill=brown!15, very thick, dashed, minimum size=7mm},
		sqBROWN/.style={rectangle, draw=black!240, fill=brown!15, very thick, minimum size=7mm},		
        ]

\node[](HG1A) at (-1.75, 0.25) {};
\node[](HG1B) at (-6.25, -3.25) {};
\node [he, draw=black!240, fill=lime!5, fit=(HG1A) (HG1B)] {};

\node[](HG2A) at (1.35, 3.25) {};
\node[](HG2B) at (6.65, -3.25) {};
\node [he, draw=black!240, fill=cyan!5, fit=(HG2A) (HG2B)] {};

        \node[sqBLUE](A2) at (4, 3) {$\tau2,  L=\cancel{10} \ \textcolor{red}{9},  y=\cancel{10} \ \textcolor{red}{9}$};
        \node[] at (0, 1) { $\cancel{10} \ \textcolor{red}{9}$};
		\node[] at (4, -1.5) { $\cancel{10} \ \textcolor{red}{8}$};
			
		\node[sqYELL](C1) at (-4, 0) {$\tau3, L=\cancel{20} \ \textcolor{red}{18}, y =\cancel{10} \ \textcolor{red}{9} $};
       \node[] at (-4.5, -1.5) { $10$};

		
		\node[sqGREEN](F1) at (-4, -3) {$\tau4,L=50, y=20$};
	
		\node[sqGREEN](F2) at (4, -3) {$\tau5, L = L=\cancel{50} \ \textcolor{red}{40}, y=\cancel{20} \ \textcolor{red}{16}$};	
		
		\tikzset{thick edge/.style={-, black, fill=none, thick, text=black}}
		\tikzset{thick arc/.style={->, black, fill=black, thick, >=stealth, text=black}}
		
		\draw[line width=0.4mm, black,-> ] (A2) to[out=-90,in=90] (C1);
		
		\draw[line width=0.4mm, black,-> ] (C1) to[out=-90,in=90] (F1);
		\draw[line width=0.4mm, black,-> ] (C1) to[out=-90,in=90] (F2);

		
       
		\end{tikzpicture}
		}
        }
    \caption{\footnotesize{Calculation of Country $i$'s power over Country $j$.   Country $j$ is on the left and Country $i$ is on the right.  
    If Country $j$ is non-strategic, then Country $j$ loses 7 units of labor to Country $i$'s 6 units of labor and Country $i$'s power is 7/6.   In contrast, if Country $j$ is strategic, then Country $j$ loses 2 units of labor and Country $i$ loses 11,  and Country $j$'s power is only 2/11.  (In this example, that is actually the inverse of Country $i$'s power over Country $j$ which is 11/2). }
    \label{fig:strategic}}   
  \end{figure}

Accounting for strategic reactions can actually simplify power calculations on the one hand, and complicate them on the other.   
The simplification is that optimal strategic disruptions throughout the supply network on the margin can be calculated via ``threads''---that have one entrance and one exit from each technology on the thread and go between a raw material (i.e., an intermediate good with no incoming flows) and a final good.  
A complication is that the only relevant threads are those that are solutions to an equilibrium and nested max (or min) problem,  where each country at each node is minimizing its domestic disruption anticipating the consequences given other countries' strategies at every node, and potentially maximizing the disruption of some other countries.  
A second complication is that we now need to specify what third-party countries do when they have choices of whom to harm, holding their own harm constant.

To keep things uncluttered, we focus on the case of two countries, but one can extend this to put third-party countries in one of the two blocks (with some complications if countries care about which country within the other block they do the most harm to).

We first note a complication that requires us to specify an equilibrium definition, where actions at each technology are specified, and each has to be a best response given the others.
To see why this is true, consider the following example that shows that the optimal strategy at one technology depends on what is being done at others. 

Suppose that the labor endowment of both countries is the same and that Country $j$ is seeking to minimize the power that Country $i$ has over it. Country $j$ has the choice at a given technology to either disrupt 3 units of labor in each country, or else to disrupt 2 units of labor in the other country and 1 unit of its own labor.  If all other paths involve only minor labor disruptions, then the second option results in $i$ having a power of approximately $1/2$ and is preferable.  However, if 5 units are already disrupted in $j$ and none in $i$, then the first gives $i$ a power of $8/3$ while the second gives $i$ a power of $6/2$, and so now the first disruption minimizes $i$'s power. 

To define an equilibrium, we can adjust the definition of a consistent disruption above to find the {\sl strategic power} of Country $i$ over Country $j$ from disrupting some initial technology.  
The only difference is that we now allow Country $j$ to have discretion over its flows as well.
Note that if Country $j$ acts to minimize $i$'s power, then this is a zero-sum game (Country $i$'s payoff is its power and Country $j$ acts to minimize that power and so has a payoff equal to minus Country $i$'s power). Thus, it has a unique Nash equilibrium outcome and associated power. 

That is, to define the strategic power of Country $i$ over $j$ from some initial technology, we adjust the previous definition of a consistent disruption to eliminate the requirement that if $\tau^{\ell'}\in T_j$ then there is proportional rationing.  Everything else in the definition is as it was. Then we define a game for each $\tau^i\in T_i^{max}$ such that $i$ and $j$ choose flows that satisfy consistency as defined above (but not requiring $j$ to equalize disruptions), and $i$ maximizes its power while $j$ minimizes it. As this is a zero sum game it has a unique equilibrium power for each $\tau^i\in T_i^{max}$,
and $i$'s strategic power over $j$ is the max of these powers over  
$\tau^i\in T_i^{max}$.

Computing the Nash equilibrium outcome given any initial $\tau^i\in T_i^{max}$ can be challenging, given that even in the absence of (directed) cycles a technology can be reached via multiple paths from another technology---as it could be reached from both upstream and downstream from paths originating at the other technology.   
If we focus on supply chains with no {\sl undirected} cycles in the supply chain, then a technology can only be reached from one direction from some other technology, and so it has a well-defined direction of disruption to consider if it is reached.
From any starting technology, we can then calculate the equilibrium power as follows.   
Consider all technologies that are on undirected paths from this technology, and thus lie on some combination of paths going up and downstream.
Given the absence of undirected cycles,  this means that each of these technologies is either reached via some path coming from upstream or coming from downstream, but not both.   Not all such technologies will be reached, but we have to specify what they would do if reached in order to find the equilibrium disruption consequence.  
We can work by backward induction.  We specify what a country will do at maximum distance for every given set of choices of technologies at lower distance.  We then move steps back iteratively, given what we have found at greater distances and for every possible set of choices at technologies of smaller distances.  
Equilibrium requires that decisions 
maximize the power of Country $i$ when the technology is in Country $i$ and minimize the power of $i$ when the technology is in Country $j$. Note that in this case, one can find a pure strategy equilibrium by selecting a pure best response at each node in the backward induction. 

This analysis of strategic power is summarized as follows.
\begin{itemize}
\item The strategic power of Country $i$ over Country $j$ is well-defined and the (unique) value of a zero sum game in which $i$ picks some technology to disrupt and chooses consistent flows at each technology in that supply chain that it controls (as defined above), and $j$ also chooses consistent flows at each technology. 
\item In the case of undirected cycles that power can be found in pure strategies and via backward induction.
\end{itemize}

Our notion of strategic power gives a maximum power for any given disruption, presuming that the two countries stay within the part of the supply chain associated with the initial disruption. There is a richer game in which $j$ could respond with some other disruptions of its own, and so forth.  To solve that game, one has to introduce some more widely encompassing preferences, and we leave that for further research.  Our analysis provides an important building block for any further extension.

\section{Equilibrium}\label{sec:equilibrium}


It is helpful to define the transportation-cost adjusted prices that technology $\tau$ faces. We let $\widehat{p}_{\tau} \in \Re_+^{1+M+F}$ denote these prices. The first entry of this vector records the adjusted cost of labor for technology $\tau$. This is given by $\widehat{p}_{\tau L}=\min_{n'\in N}\theta_{n'\tau}p_{n'}$. The next $M$ entries record the adjusted sourcing costs of inputs, with entry $\widehat{p}_{\tau m}=\min_{\tau':O(\tau')=m}\theta_{\tau'\tau}p_{\tau'}$. The final $F$ entries are redundant, as by assumption final goods are never used as inputs, but for consistency we set $\widehat{p}_{\tau f}=\min_{\tau':O(\tau')=f}p_{\tau'}=p_f$.

An {\sl equilibrium} of an economy is $(N, M, F, \{T_n\}_n, {\{L_{n}\}_n}, \theta )$ is a specification of
\begin{itemize}
\item prices $p \in \Re_+^{N+T}$ for labor and technologies (whether in use or not), and
\item for all countries $n$ and technologies $\tau\in T^{n}$:
\begin{itemize}
\item a corresponding amount of the output $O(\tau)$ for each  $\tau\in T$ denoted by $y^{\tau }\geq 0$,
\item amounts (in units shipped) of inputs for use by the technology $\tau$
$x_{k \tau}\in  \Re_+$ for all $k\in N\cup T$, and
\end{itemize}
\item the total amount of each final good $f$ consumed in each country $n$, $c_{fn}$;
\end{itemize}
that satisfy the following conditions:

\begin{itemize}

\item \textbf{Consumers optimize:}
 Labor is fully supplied at the highest available wage (accounting for transportation costs) and wages are
 spent on final goods to maximize utility. That is, laborers in country $n$ choose consumption $(c_{1n}, \ldots, c_{Fn})$ to maximize
$$
U(c_{1n}, \ldots, c_{Fn}) {\rm \ \ subject \ \  to \ \ } \sum_{f\in F} p_{f} c_{fn} = L_n p_n.
$$

\item \textbf{Producers optimize (and earn zero profits):} For each active technology $\tau$ (such that $y^{\tau}>0$)
$$\widehat{p}_{\tau} \cdot \tau =0,$$ and for each inactive technology $\tau$ (such that $y^{\tau}=0$)
$$\widehat{p}_{\tau}\cdot \tau \leq 0.$$

\item \textbf{Production Plans are Feasible} For each active technology $\tau$
\begin{equation}
\sum_{\tau':O(\tau')=k} \frac{x_{\tau' \tau}}{\theta_{\tau' \tau} }=  -  \tau_k  y_{\tau},
\label{feasibleproduction}
\end{equation}
which says that the total inputs sourced, adjusted for shipping losses, equal the amount of inputs required under the technology to produce the desired level of output.

\item \textbf{Markets clear:}

\begin{itemize}
\item Labor markets clear: the total amount of labor from each country $n$ being used in production throughout the world is equal to the its endowment:
\begin{equation}
L_n = \sum_{\tau} x_{n \tau}.
\label{clearlabor}
\end{equation}

\item Intermediate goods markets clear: For each technology $\tau$ such that $O(\tau)\in M$, the amount being used in production across all countries is equal to the amount produced. So:
\begin{equation}
\sum_{\tau'} x_{\tau \tau'} = y_{\tau}.
\label{clearintermediate}
\end{equation}

\item Final goods markets clear:  For each technology $\tau$ such that $O(\tau)\in F$, the amount being consumed across all countries is equal to the amount produced. So:
\begin{equation}
\sum_{n} c_{fn} = \sum_{\tau \in T: O(\tau )=f} y_{\tau}.
 \label{clearfinal}
\end{equation}
\end{itemize}
\end{itemize}

Given constant returns to scale, in equilibrium active technologies make zero profits and hence all goods that are produced in positive amounts are priced at their respective unit costs. In particular, for all countries $n$ and technologies $\tau\in T^{n}$, if $y_{\tau}>0$ and $O(\tau)=k$, then
$$p_{\tau} =  \sum_{k'\neq k} -\tau_{k'} \widehat{p}_{\tau k'}.$$

Note that the left hand side is $p$ and the right hand side is $\widehat{p}$, capturing that this is the cost of output in $n$, while inputs are sourced from their cheapest source.

\subsection{Existence and Uniqueness of Consumption}

To ensure equilibrium existence, it is necessary that the available technologies avoid money pumps.
We assume that each technology uses a positive amount of labor,  and given the limited supply of labor, that is sufficient for avoiding money pumps and is quite natural, although weaker conditions could be used.\footnote{ We could instead just assume that each final good supply chain uses a positive amount of labor, but at the expense of some technicalities (including the possibility of zero prices for some intermediate goods) that lead to more complex and obscure proofs.}

We also assume that $T$ is such that every final good has a viable supply chain, meaning that there is some combination of technologies that produce a positive amount of it.  This ensures that every final good has a finite price in equilibrium.

Lemma \ref{lemma:exists} in Appendix \ref{sec:proofs} shows that there exists an equilibrium and the same bundle of final goods is produced in all equilibria.  We show that by expanding the space of technologies to let them take source-specific inputs, transportation costs can be incorporated into them. Then, with this transformation in hand, standard existence results apply.

We note that every final good is consumed and has a finite, positive price in equilibrium, and that the same is true of each country's labor.  These follow from the fact that every final good has a viable supply chain, each country has a positive endowment of labor that is completely supplied in equilibrium, and preferences are increasing.  This ensures that all final goods are produced and consumed, and so must have positive, finite prices (given increasing preferences), and labor cannot have either a 0 price (or some technology would use an infinite amount of it\footnote{Note that this follows from the proof in that there are equivalent primitive technologies in terms of labor units, and 0 profits ensure that it is used in}), or an infinite price (or some consumer's problem would result in infinite demand and markets would not clear).

Although there is unique total consumption in equilibrium, this can sometimes be supported by multiple wage profiles, and the existence of multiple wage profiles is generic. To see this, suppose there are three countries, $A$, $B$ and $C$ and each have a unit of labor endowment.  There is one final good $f$ and two intermediate goods $1$ and $2$.
Suppose that labor transportation costs are sufficiently large that labor never moves, but in contrast intermediate goods have transportation costs of $1$, so that they ship with no loss. Suppose country $A$ can only produce intermediate good $1$, and it takes $1$ unit of its labor to produce one unit of it.  Country 2 can only produce intermediate good $2$, and it takes $1$ unit of its labor to do so. Country $3$ can only produce the final good, and to produce one unit it has to combine $1$ unit of intermediate good $1$,  $1$ unit of intermediate good $2$, and $1$ unit of its labor. Thus it takes $1$ unit of labor from each country to produce $1$ unit of the consumption good, and in equilibrium exactly $1$ unit of the final good is produced. However, any vector of wages that sums to $1$ can be supported as part of an equilibrium. These wages just pin down the relative amounts of consumption that the three different countries can afford, and the consumption good market clears (as well as the labor good markets).

Note that this example is ``robust'' in the following sense: the multiplicity is robust to the ratios at which intermediate goods are combined to produce the final good, and how much labor is needed by each country in its production technology.\footnote{Suppose we added in a technology for each country that converted its labor into the final good $f$. For country $i$ suppose that $k_i>3$ units of labor are needed to produce one unit of good $f$. Then the equilibrium wages that could be supported would be any positive vector summing to $1$ such that the wage in country $i$ is weakly greater than $1/k_i$.}  If instead, labor transportation costs are sufficiently low that labor in one country can substitute for that in another, then that begins to pin down relative wages, until with sufficient mobility  wages are equal across countries.

\section{Additional Proofs: Representative Consumer and Equilibrium Existence}\label{sec:proofs}

\begin{lemma}\label{lemma:representative_consumer}
In any equilibrium aggregate consumption is equivalent to the consumption choice of a single representative consumer with preferences represented by the utility function $U(c_1, \ldots, c_F)$ and with wealth $\sum_n L_n p_{Ln}$.
\end{lemma}

\noindent {\bf Proof of Lemma \ref{lemma:representative_consumer}:}

Country $n$'s consumer solves
$$
\max_{c_{1n}, \ldots, c_{Fn}} U(c_{1n}, \ldots, c_{Fn}) {\rm \ \ subject \ \  to \ \ } \sum_{f\in F} p_{f} c_{fn} = L_n p_n.
$$
Given finite prices for all final goods and labor (which follows in equilibrium, as shown below),
and given that the utility function is $U(\cdot)$ is increasing and strictly quasi-concave there is a unique solution to this problem. Further, as consumers' utility functions are identical and all consumers in all countries face the same prices for final goods (because there are no transportation costs on final goods), each country's consumer solves the same problem except for differences in their labor endowments inducing differences in their wealth levels. Thus, as $U(\cdot)$ is homogeneous of degree $1$, the solution to each country's consumer problem is a re-scaling of the same bundle of goods. Moreover, it follows that a representative consumer with utility function $U(c_1, \ldots, c_F)$ and with wealth $\sum_n L_n p_{Ln}$ chooses a re-scaling of the same bundle of goods, and that re-scaling is the aggregate consumption.\eproof

\subsection{Proof of Lemma~\ref{lemma:exists}}

\begin{lemma}\label{lemma:exists}
There exists an equilibrium. Moreover, the same bundle of final goods is produced in all equilibria.   
\end{lemma}

\begin{proof}
We prove existence of an equilibrium by mapping our economy into one without transportation costs, and then applying standard results. In order to do this we use the transportation costs to map each technology $\tau\in \Re^{1+M+F}$ into source-specific technologies $t\in\Re^{N+T}$. For each technology $\tau$ that requires inputs $I(\tau)$, we create $\prod_{k\in I(\tau)}|\{\tau:O(\tau)=k\}|$ source-specific technologies allowing for each possible combination of sourcing choices from different technologies across the different inputs (where $|\cdot|$ denotes cardinality).
Further, we adjust the number of units of each input required to represent the number of units that need to be sourced from that technology, including those units that will be lost to transportation costs. So, if technology $\tau$ requires $k$ units of good $g$, and the transportation costs associated with sourcing good $g$ from $\tau'$ are $2$, then corresponding technology source-specific technologies $t$ that source good $g$ from technology $\tau'$, requires $2k$ units of good $g$; and similarly for the $N$ sources of labor.

Replacing transportation costs and technologies with technology source-specific technologies, existence of equilibrium in our environment can be proven using standard techniques (noting that each technology can then be represented as a production possibility set by adding free disposal), such as that used to prove Theorem 17BB2 of \cite{mas1995microeconomic}. This establishes existence.

An equilibrium of our economy must be Pareto efficient by the first Welfare Theorem. Further, as by Lemma \ref{lemma:representative_consumer} our economy admits a representative consumer with utility $U(c_1, \ldots, c_F)$, all Pareto efficient allocations must maximize this utility function subject to feasibility constraints. We show that the set of feasible consumption bundles is convex and compact, and hence, as the representative consumer's utility function is increasing and quasi-concave, that there is a unique bundle of final goods that solves the Pareto problem.

A consumption bundle $c\in\Re_+^F$ is feasible if
\begin{enumerate}
\item $c_f\leq \sum_{\tau:O(\tau)=f} y_{\tau}$ for all $f$,
\item $y_{\tau}\leq \min_{k\in I(\tau)} \sum_{\tau':O(\tau')=k} \left(\frac{x_{\tau'\tau}}{-\tau_k\theta_{\tau'\tau}}\right)$  for all $\tau$
\item $\sum_{\tau} x_{n\tau}\leq L_n$ for all $n$.
\end{enumerate}

The first conditions requires that enough final goods are produced. The second condition requires that sufficient inputs are sourced to support the required output for each technology. The final condition requires that the use of labor satisfies the labor endowments.

Consider two feasible consumption bundles $c$ and $c'$. We first show that this implies the consumption bundle $c''=\lambda c + (1-\lambda) c'$ is also feasible for all $\lambda\in[0,1]$. First suppose that we reduced the labor endowments of all countries to $\lambda$ their initial levels. As all technologies are constant returns to scale, and it was feasible before the reduction to produce the bundle $c$, it must be feasible after the reduction to produce the bundle $\lambda c$. This can be obtained by reducing all inputs (and hence all outputs) of all used technologies to $\lambda$ their initial levels. Equivalently, if all labor endowments are reduced to $(1-\lambda)$ their initial levels, then the bundle $(1-\lambda)c'$ is feasible. Hence, the bundle $c''$ is feasible with the initial labor endowments. This shows that set of feasible consumption bundles is convex.

We now show that the set of feasible consumption bundles is compact. Take any feasible use of technologies that produces a non-zero consumption bundle $c$. As all technologies use a strictly positive amount of labor, producing this bundle uses a strictly positive amount of labor. Let $L(c)\in \Re_+^n$ denote the vector of labor inputs used across countries. Fixing the use of technologies, as all technologies are constant returns to scale, we can increase all inputs to $\lambda\geq 1$ their initial level to produce the consumption bundle $\lambda c$. Thus there exists a unique $\bar \lambda \geq 1 $ that maximizes $\lambda c $ subject to $\lambda L(c)_n \leq L_n$ for all $n$. As the consumptions must be non-negative and the $0$ bundle is feasible, the set of feasible consumption bundle is compact.

The Pareto problem therefore involves maximizing a strictly quasi-concave function subject to a convex and compact constraint set, and thus has a unique solution. Hence, in all equilibria, the same aggregate consumption must occur.
\end{proof}

\end{document}